\documentclass[aps,prl,twocolumn,showpacs,groupedaddress,notitlepage,superscriptaddress,floatfix,longbibliography,10pt]{revtex4-1}
\usepackage{amsmath,amsthm,amssymb,amsfonts,enumitem,float,graphicx,physics,hyperref}

\usepackage{color}
\usepackage[dvipsnames]{xcolor}
\usepackage[most]{tcolorbox}

\usepackage[normalem]{ulem}

\newcommand{\PRLsec}[1]{\textit{#1}.---}

\setlist[itemize]{leftmargin=*}
\setlist[enumerate]{leftmargin=*}

\usepackage{aliascnt}

\newtheorem{theorem}{Theorem}[section]
\newtheorem*{theorem*}{Theorem}

\newaliascnt{lemma}{theorem}
\newtheorem{lemma}[lemma]{Lemma}
\aliascntresetthe{lemma}
\newtheorem*{lemma*}{Lemma}

\newaliascnt{corollary}{theorem}
\newtheorem{corollary}[corollary]{Corollary}
\aliascntresetthe{corollary}

\newaliascnt{result}{theorem}
\newtheorem{result}[result]{Result}
\aliascntresetthe{result}
\newtheorem*{result*}{Result}

\newaliascnt{example}{theorem}

\aliascntresetthe{example}

\newaliascnt{proposition}{theorem}
\newtheorem{proposition}[proposition]{Proposition}
\aliascntresetthe{proposition}

\newtheorem{definition}{Definition}[section]

\newtheorem{remark}{Remark}[section]
\newtheorem*{remark*}{Remark}
\newtheorem{problem}{Problem}

\usepackage{cleveref}

\Crefname{theorem}{Theorem}{Theorems}
\Crefname{problem}{Problem}{Problems}
\Crefname{result}{Result}{Results}
\Crefname{lemma}{Lemma}{Lemmas}
\Crefname{corollary}{Corollary}{Corollaries}
\Crefname{proposition}{Proposition}{Propositions}
\Crefname{example}{Example}{Examples}
\Crefname{section}{Appendix}{Appendices}

\newcommand{\poly}[1]{\operatorname{poly}\!\left(#1\right)}

\newcommand{\lr}[1]{\left(#1\right)}
\newcommand{\E}{\mathbb{E}}
\newcommand{\EE}[1]{\mathbb{E}\left[#1\right]}

\newcommand{\C}{\mathbb{C}}
\newcommand{\D}{\mathbb{D}}

\newcommand{\cJ}{\mathcal{J}}

\newcommand{\cN}{\mathcal{N}}

\newcommand{\cR}{\mathcal{R}}

\newcommand{\cU}{\mathcal{U}}

\newcommand{\iprhob}[1]{\left\langle #1 \right\rangle_{\beta}}

\newcommand{\supp}{\operatorname{supp}}
\newcommand{\Id}{I}

\newcommand{\ad}{\operatorname{ad}}
\newcommand{\OTOCorder}{L}
\newcommand{\locality}{k}

\newcommand{\sgn}{\operatorname{sgn}}
\newcommand{\pf}{\operatorname{Pf}}

\hypersetup{
    colorlinks=true, 
    linkcolor=cyan,
    citecolor=magenta, 
    filecolor=magenta, 
    urlcolor=cyan,
    runcolor=cyan
}

\begin{document}

\title{SYK thermal expectations are classically easy at any temperature}

\author{Alexander Zlokapa}
\affiliation{Center for Theoretical Physics -- a Leinweber Institute, MIT}
\author{Bobak T. Kiani}
\affiliation{Bowdoin College}

\begin{abstract}
Estimating thermal expectations of local observables is a natural target for quantum advantage. We give a simple classical algorithm that approximates thermal expectations for Gibbs states of local Hamiltonians, and we show it has quasi-polynomial cost $n^{O(\log (n/\epsilon))}$ for all temperatures above a phase transition in the free energy. For many natural models, this coincides with the entire fast-mixing, quantumly easy phase. Our results apply to the Sachdev-Ye-Kitaev (SYK) model at any constant temperature due to its absence of a phase transition---despite its entanglement, sign problem, and polynomial quantum circuit lower bounds. Beyond SYK, we rigorously establish a universal classically easy high-temperature phase for all local, bounded-degree Hamiltonians and show that it extends to temperatures strictly colder than the death of entanglement transition.
\end{abstract}

\maketitle

A central question in quantum computation asks if thermal states admit problems featuring quantum advantage over classical computers.
Recent advances in quantum algorithms~\cite{temme2011quantum,chen2021fast,chen2023efficient,rall2023thermal,zhang2023dissipative,rouze2024optimal,jiang2024quantum,bakshi2024high,gilyen2024quantum,ding2025end,ding2025efficient,chen2025efficient,chen2025quantum} enable the efficient preparation of Gibbs states 
\begin{align}
    \rho_\beta = \frac{e^{-\beta H}}{Z(\beta)} \; \text{ for } \; Z(\beta) = \Tr(e^{-\beta H})
\end{align}
for a variety of models.
Sampling from Gibbs states can yield quantum advantage under complexity-theoretic conjectures \cite{chen2024local}; however, it remains open if the more physically relevant task of estimating local thermal expectations,
\begin{align}
    \Tr(O\rho_\beta) \;\text{ for bounded local } \;O,
\end{align}
can yield quantum advantage at constant temperature.

We provide evidence that for many natural models---including those previously proposed for quantum advantage---this task is classically easy. 
Our argument is based on Barvinok's interpolation method, which approximates the partition function in quasi-polynomial time~\cite{barvinok2014computing,barvinok2016computing,barvinok2016approximating,barvinok2016combinatorics,barvinok2018approximating}. This succeeds when $Z(\beta)$ is analytic along a path from $\beta=0$ to the target inverse temperature, i.e., in the absence of a phase transition~\cite{lee1952statistical,fisher1965nature}. Our results bypass common classical barriers highlighted in Fig.~\ref{fig:box}
because they rely only on analyticity of $Z$ rather than on sampling or low‑entanglement structure. Since a phase transition often leads to a slow-mixing phase that obstructs quantum algorithms~\cite{dennis2002topological,gamarnik2021overlap,basso2022performance,gamarnik2024slow,anschuetz2025efficient,zlokapa2025average,chen2025information}, our results suggest that efficient classical algorithms often exist throughout the entire quantumly easy phase.

We apply this generic argument to the SYK model, one of the most well-studied models in condensed matter and high-energy physics~\cite{sachdev1993gapless,parcollet1999non,georges2000mean,georges2001quantum,kitaev2015SYK,maldacena2016remarks,maldacena2016bound,kitaev2018soft}. 
The SYK model is a strongly interacting fermionic model defined by randomly coupling all possible $q$-body interactions between $n$ Majorana fermions. For $J_{i_1\cdots i_q} \sim_\mathrm{iid} \cN\lr{0, (q-1)! J^2/ n^{q-1}}$ and fermions satisfying $\{\psi_i,\psi_j\} = \delta_{ij}$, a Hamiltonian instance is given by
\begin{align}
    H = i^{q/2} \sum_{1 \leq i_1 < \cdots < i_q \leq n} J_{i_1\cdots i_q}\psi_{i_1} \cdots \psi_{i_q}.
\end{align}
In the context of quantum algorithms, the SYK model has been considered a candidate Hamiltonian for quantum advantage due to various obstructions to classical algorithms outlined in Fig.~\ref{fig:box}. Moreover, it is expected to be quantumly easy due to its absence of a phase transition: it has non-glassy, replica symmetric physics at any (constant) temperature~\cite{almheiri2024universal,zhang2019evaporation,maldacena2021syk,chen2021fast,schuster2025cooling}.

\begin{figure}[H]
    \centering
    \begin{tcolorbox}[
        colback=gray!10,
        colframe=gray!50,
        boxrule=0.5pt,
        arc=10pt,
        left=6pt,
        right=6pt,
        top=6pt,
        bottom=6pt
    ]
    \begin{itemize}
        \item \emph{Sign problem.} Obstructs quantum Monte Carlo evaluation of $e^{-\beta H}$ as a probabilistic process~\cite{troyer2005computational}.
        \item \emph{Large entanglement and magic.} Obstructs tensor networks~\cite{vidal2003efficient} and stabilizer circuits~\cite{gottesman1998heisenberg}.
        \item \emph{Circuit complexity lower bounds.} Obstruct direct classical simulation~\cite{markov2008simulating,cirac2021matrix} and variational ansatzes such as Gaussian states~\cite{hastings2022optimizing,hastings2023field}.
        \item \emph{Instance-to-instance fluctuations.} Obstructs algorithms that output the ensemble average or make other mean-field approximations~\cite{mezard1987spin,chen2024sparse,google2025observation}.
    \end{itemize}
    \end{tcolorbox}
    \caption{Common obstructions to classical algorithms. The SYK model satisfies all of the above properties~\cite{liu2018quantum,anschuetz2025strongly,hastings2023field,ramkumar2025high,bettaque2026magic}.}
    \label{fig:box}
\end{figure}

We argue that Barvinok’s interpolation method estimates thermal expectations within a broad zero-free regime that extends well beyond settings accessible to standard classical heuristics. 
Our main application is the SYK model, where standard but non-rigorous large-$n$, large-$q$ path-integral methods locate the complex zeros of the partition function and thereby determine the $\beta$-dependence of Barvinok’s interpolation runtime \cite{maldacena2016remarks,kitaev2015SYK,takahashi2011replica,takahashi2013zeros,khramtsov2021spectral}. 
We support this picture with fully rigorous results for bounded-degree Hamiltonians showing a universal high-temperature, computationally easy phase that extends to temperatures lower than the sudden death of entanglement~\cite{bakshi2024high}. A family of provably entangled Gibbs states certifies the separation between these computationally easy and separable temperatures.

\begin{figure*}
    \centering
    \includegraphics[width=0.9\linewidth]{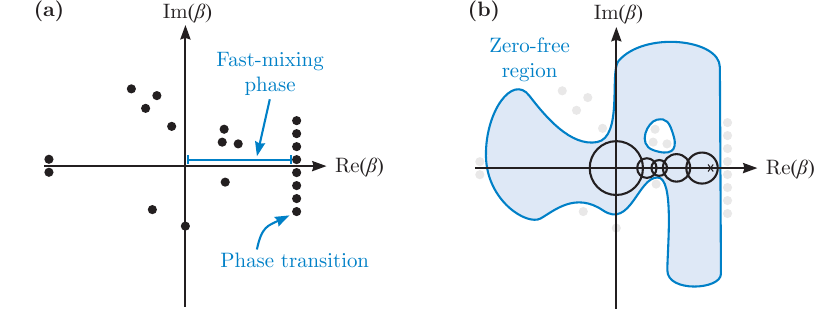}
    \caption{(a) Cartoon of a phase transition identified by the complex (Fisher) zeros of $Z(\beta)$ asymptotically intersecting the real axis in the thermodynamic limit. (b) Illustration of Barvinok's algorithm to estimate $Z(\beta)$ via an analytic continuation of a Taylor expansion from the origin. Because the high-temperature phase is characterized by Fisher zeros bounded away from the real axis, a zero-free strip of constant height is always available for Barvinok's algorithm to analytically continue along. For constant $\beta$, this yields a quasi-polynomial time to estimate (up to inverse polynomial additive error) $\log Z(\beta)$ and, if the zero-free region is stable upon perturbation by local observables, thermal expectations. For many natural systems (including SYK), the high-temperature phase corresponds to the non-glassy, fast-mixing phase.}
    \label{fig:barvinok}
\end{figure*}

\PRLsec{Results}%
To analyze our classical algorithms, the main objects we study are the zeros of the partition function $Z(\beta)$; i.e., complex-valued $\beta_j$ such that $Z(\beta_j) = 0$. Note that all such $\beta_j$ must have an imaginary component: since $Z(\beta) = \sum_k e^{-\beta E_k}$ for eigenvalues $E_k$ of the Hamiltonian, the partition function is always nonzero for any choice of real $\beta$. However, if the complex zeros $\beta_j$ approach the $\Re \beta$ axis at some $\beta_*$ in the thermodynamic limit $n\to\infty$, then the free energy $-\frac{1}{\beta} \log Z(\beta)$ becomes non-analytic in the vicinity of $\beta_*$. This corresponds to a phase transition, as first introduced by Lee and Yang~\cite{lee1952statistical} for complex-valued external fields and later extended by Fisher~\cite{fisher1965nature} to complex-valued temperatures.

The \emph{absence} of a phase transition implies that the partition function admits a zero-free region (Fig.~\ref{fig:barvinok}a). While this suffices to estimate the partition function (Fig.~\ref{fig:barvinok}b), we require here a mild strengthening in order to estimate thermal expectations of local observables. 
We will refer to $\beta<\beta_*$ as a \emph{stable high-temperature phase} if there exists constant $b_0>0$ such that for all sufficiently large $n$ and sufficiently small $\lambda$, the following holds
\begin{equation}\label{eq:stable}
	Z_n(\beta,\lambda)\neq 0  \text{ in } \cR := \{\beta\,:\, \Re \beta \in [0,\beta_*], |\Im \beta| < b_0\}
\end{equation}
and in an arbitrarily small (constant) neighborhood of $\cR$.
In many models, this high-temperature phase (in particular the inclusion of $+\lambda O$) can be established by bounding the influence of interactions through combinatorial bounds from cluster expansions~\cite{kotecky1986cluster,fernandez2007cluster}. 
It is also implied by the Dobrushin condition~\cite{dobrushin1985completely}, which implies rapid mixing of Glauber dynamics in classical models~\cite{hayes2006simple}.

The classical algorithm, based on Barvinok's interpolation method introduced for counting problems~\cite{barvinok2014computing,barvinok2016computing,barvinok2016approximating,barvinok2016combinatorics,barvinok2018approximating}, evaluates derivatives of $\log Z(\beta)$ at $\beta=0$ and estimates the partition function at a target temperature via an analytic continuation. Within a zero-free disk---i.e., $Z(\beta) \neq 0$ for $|\beta|<R$---the partition function can be estimated to multiplicative error $\epsilon$ by truncating the series at order $m = O(\log(n/\epsilon))$, yielding an algorithm that costs $n^{O(m)}$ to achieve accuracy $|\log \widetilde{Z} - \log Z(\beta)| \leq \epsilon$. In general, a zero-free disk might not exist. However, the absence of a phase transition guarantees that the complex zeros do not pinch the real axis, producing a rectangular zero-free region: for $\Re \beta < \beta_*$, one can always find some constant $b_0 > 0$ such that a zero-free region exists for all $|\Im \beta| < b_0$ \footnote{More carefully, the region may include $O(1)$ zeros---even asymptotically close to the real axis---which do not cause the free energy to be non-analytic in the thermodynamic limit. However, Barvinok's algorithm can efficiently avoid a constant number of zeros by trying a constant number of analytic continuations of constant width, as described in~\cite{eldar2018approximating}.}. An analytic continuation yields the same scaling $m = O(\log(n/\epsilon))$. In the \emph{stable} high-temperature phase of \eqref{eq:stable}, the same argument applies to $Z(\beta, \lambda)$. Thermal observables $\Tr[O\rho_\beta]$ can thus be efficiently estimated by a finite-difference approximation
\begin{align}
    -\frac{1}{\beta}\partial_\lambda \log Z(\beta,\lambda) \Big|_{\lambda=0} \approx -\frac{\log Z(\beta,\lambda) - \log Z(\beta,0)}{\beta\lambda}
\end{align}
for sufficiently small $\lambda$ \cite{bravyi2021complexity}.

Before addressing the SYK model, we first give rigorous evidence that Barvinok's algorithm can bypass typical obstructions to classical algorithms. Prior work suggests that Barvinok's algorithm can overcome the sign problem to efficiently compute the permanent of random matrices~\cite{eldar2018approximating} and to contract random tensor networks~\cite{jiang2025positive}. In the latter case, Barvinok's algorithm fails at a phase transition in entanglement~\cite{chen2025sign}.

Our first result rigorously establishes that Barvinok's method can in fact bypass the entanglement obstruction of Fig.~\ref{fig:box}. Previous applications of Barvinok's interpolation for quantum Gibbs states have been restricted to identifying a high-temperature phase universal to arbitrary $k$-local, degree-$D$ spin Hamiltonians, $\beta < \beta_{\rm zero-free} = \Theta(1/kD)$~\cite{harrow2020classical,mann2021efficient,yao2022polynomial,mann2024algorithmic}, where $D$ is the maximum number of Hamiltonian terms acting on a single site. In comparison, for Pauli Hamiltonians, Ref. \cite{bakshi2024high} shows that all Gibbs states are separable at inverse temperature $\beta < \beta_{\rm sep}$ for $\beta_{\rm sep} \geq 1/(100k^2D)$~\footnote{The notion of degree $D$ here is different from the dual interaction graph degree $\mathfrak{d}$ used in~\cite{bakshi2024high}.}; it is not known if the bound on $\beta_{\rm sep}$ is tight. A natural question is whether the universal classically easy phase $\beta < \beta_{\rm zero-free}$ is contained in a trivial separable phase. We show that at temperatures above $1/\beta_{\rm zero-free}$, Gibbs states can be provably entangled, ruling out a confluence of phase transitions in separability and classical easiness.
\begin{theorem*}[Universal high-temperature phase]
    For $k$-local degree-$D$ Hamiltonians, a quasi-polynomial classical algorithm can estimate the free energy to inverse polynomial error throughout a universal high-temperature phase $\beta_{\rm zero-free} \geq \frac{1}{2e(k(D-1)+1)}$. Moreover, for $k \geq 8$, the inequality $\beta_{\rm zero-free} > \beta_{\rm sep}$ holds.
\end{theorem*}
The bound on $\beta_{\rm zero-free}$ that we report (proven in \Cref{app:hight}) is a slight improvement over the temperature of prior quasi-polynomial classical and quantum algorithms~\cite{harrow2020classical,mann2021efficient,yao2022polynomial,rouze2024optimal,wild2023classical,mann2024algorithmic,bakshi2024high,bakshi2025dobrushin}. The proof that the zero-free temperature is strictly lower than the separability transition is given in \Cref{prop:swap}. Note that the above theorem is a universal result regarding the zero-freeness and separability for \emph{all} $k$-local degree-$D$ Hamiltonians above some temperature. With additional structure, zero-freeness and entanglement may happen at lower temperatures. As an example of this, we discuss the case of stabilizer code Hamiltonians in \Cref{app:codes} where analyticity extends to all temperatures despite an entanglement transition at constant temperature.

Our second result shows that Barvinok's interpolation algorithm is efficient for the SYK model, and in particular controls the complex zeros of the SYK Gibbs state to determine the $\beta$-dependence of the algorithm. At constant temperature, the SYK Gibbs state is highly entangled~\cite{liu2018quantum} and contains magic~\cite{bettaque2026magic}; moreover, typical eigenstates have at least polynomially growing circuit complexity and matrix product state bond dimension \cite{anschuetz2025strongly}. This causes variational ansatzes, such as mixtures of Gaussian states, to poorly approximate the Gibbs state \cite{hastings2023field,ramkumar2025high}. Quantum Monte Carlo approaches also fail due to the sign problem \cite{troyer2005computational}. Finally, mean-field approximations that simply return the observable in expectation fail to capture instance-to-instance fluctuations, which are inverse-polynomially large. Nonetheless, we show that Barvinok's algorithm succeeds at any constant temperature and resolves these fluctuations (\Cref{app:syk}).

\begin{result*}[Classical algorithm for SYK]
With probability $1-o(1)$, a large-$q$ SYK instance satisfies for any local observable $O$ and any $\lambda = o(1)$
\begin{align}
	\Tr(e^{-\beta (H+\lambda O)}) \neq 0 \;\; \text{if}\;\; |\Im \beta| < 1.3 \;\; \text{or} \;\; \Re \beta \neq 0\label{eq:sykregion}.
\end{align}
Consequently, for all (constant) real $\beta > 0$, Barvinok's interpolation method estimates thermal expectations to additive error $\epsilon$ with cost $n^{O\lr{\beta \log (\beta n/\epsilon)}}$.
\end{result*}

This result is exponential in $\beta$; in comparison, a naive application of Barvinok's algorithm in the absence of a phase transition would give a doubly exponential dependence in $\beta$ (\Cref{app:disk-to-strip-mapping}). We also note that the instance-to-instance fluctuations of local observables are inverse polynomial in size, and thus the algorithm remains quasi-polynomial even when $\epsilon$ is small enough to recover disorder-dependent thermal expectations.

We show the result via a generalization of the standard SYK Euclidean path integral computation~\cite{maldacena2016remarks,kitaev2018soft} to complex $\beta$. Rather than computing $\E \log |Z(\beta)|$ with the replica trick, we use Jensen's formula to bound the number of complex zeros in terms of the annealed quantity $\E |Z(\beta)|^2$. Solving the saddle point equations in the large-$q$ limit under a replica symmetric ansatz allows us to identify a zero-free region, which we verify numerically by exact diagonalization (Fig.~\ref{fig:sykzeros}). We note that these techniques, although standard in the physics literature, are non-rigorous due to the difficulty of formally controlling the path integral.

To emphasize that Barvinok's algorithm is broadly applicable to systems with a sign problem and large entanglement, we give numerical examples of zero-free regions in \Cref{app:num} for the quantum Heisenberg model~\cite{georges2001quantum} and the doped Fermi-Hubbard model~\cite{zhang1999finite}. Since experimentally accessible phases often correspond to the fast-mixing regime before a phase transition, we expect these classical algorithms to remain quasi-polynomial in relevant settings for quantum simulators and fault-tolerant quantum Gibbs sampling algorithms.

\begin{figure}
    \centering
    \includegraphics[width=0.85\linewidth]{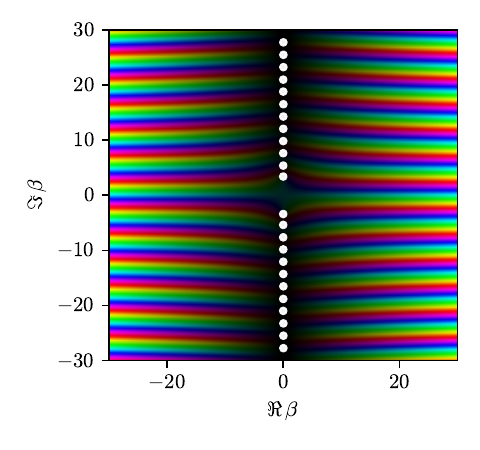}
    \caption{Complex zeros of $Z(\beta)$ obtained from exact diagonalization of a single $n=30, q=4$ SYK instance ($J=1$). Brightness indicates magnitude and hue indicates phase; zeros are indicated in white. The region $|\Re \beta| \gtrsim 0.01$ is zero-free, as well as $|\Im \beta| \lesssim 3.4$, exceeding the zero-free region identified analytically in \eqref{eq:sykregion}. In Fig.~\ref{fig:sykzeros300} of \Cref{app:syk}, we verify that the zero-free region persists up to at least $|\Re \beta|, |\Im \beta| \leq 300$ and when $H$ is perturbed by a local observable.}
    \label{fig:sykzeros}
\end{figure}

\PRLsec{Related work}%
Barvinok's interpolation method and related cluster expansions have been applied to several settings, including bounded‑degree spin systems \cite{harrow2020classical,mann2024algorithmic}, weakly interacting fermions \cite{chen2025convergence} and the guided local Hamiltonian problem \cite{zhang2024dequantized}, yielding quasi‑polynomial or even polynomial runtimes under zero‑freeness guarantees.
For dense models, arguments based on Jensen's formula control complex zeros via annealed moments \cite{hughes2008zeros,eldar2018approximating,bencs2025zeros}, an approach we extend to the SYK setting.
The computational complexity of the SYK model was analyzed in Ref.~\cite{hastings2022optimizing}, which gave a quantum algorithm achieving constant factor approximation to the ground energy and showed Gaussian states cannot achieve this constant factor approximation. 
This evidence of classical hardness was later strengthened to show convex combinations of Gaussian states fail to represent the constant-temperature SYK Gibbs state~\cite{hastings2023field,herasymenko2023optimizing,ramkumar2025high}.
Ref.~\cite{anschuetz2025strongly} showed that low-energy states of the SYK model have polynomially growing circuit complexity lower bounds (indicative of classical hardness) and that quenched and annealed free energies agree at all temperatures (indicative of quantum easiness). 
Replica and numerical analyses give strong evidence against a finite‑temperature glass transition in SYK \cite{gur2018does,gurau2017quenched,wang2019replica}, suggesting that it is quantumly efficient to prepare the Gibbs state~\cite{almheiri2024universal,zhang2019evaporation,maldacena2021syk,chen2021fast,schuster2025cooling}.
Our work builds on spectral form factor analyses that locate complex zeros on the imaginary axis, and is related to prior works that control the density of complex zeros via the replica trick \cite{takahashi2011replica,takahashi2013zeros,khramtsov2021spectral,winer2022hydrodynamic,bunin2024fisher}.

\PRLsec{Zero-freeness of the SYK model}%
The SYK model is known to have a single phase at all constant temperatures. As described in \Cref{app:disk-to-strip-mapping}, this fact in general allows Barvinok's algorithm to estimate local observables, but incurs a doubly exponential cost in $\beta$. Here, we show how to control the complex zeros of the SYK model to reduce this to $O(n^{\beta \log(\beta n/\epsilon)})$.

In the standard SYK Euclidean path integral computation over $n$ fermions and $s$ replicas, one computes disorder-averaged moments $\E Z^s$ for the partition function~\cite{kitaev2015SYK,maldacena2016remarks}
\begin{align}
    Z &= \int D\psi_i \exp\Bigg[-\int_0^\beta d\tau\Bigg(\frac{1}{2}\sum_j \psi_j \partial_\tau \psi_j \nonumber\\
    &\qquad + i^{q/2} \sum_{1 \leq i_1 < \cdots < i_q \leq n} J_{i_1\cdots i_q}\psi_{i_1} \cdots \psi_{i_q}\Bigg)\Bigg]
\end{align}
and commutes limits $s \to 0, n\to\infty$ in order to obtain the quenched free energy $-\frac{1}{\beta}\E \log Z$ via the replica trick~\cite{mezard1987spin}. Here, we instead control the complex zeros using an \emph{annealed} quantity. We count the number of complex zeros using Jensen's formula, which states that for the disk $\D(0, R)$ centered at the origin and of radius $R$,
\begin{align}
    \sum_{\beta\in \D(0,R):Z(\beta)=0} \log \frac{R}{|\beta|} = \frac{1}{2\pi} \int_{0}^{2\pi} d\theta \log \left|\frac{Z(Re^{i\theta})}{Z(0)}\right|.
\end{align}
If the Hamiltonian's zero-free region were a disk of size $R$, we could immediately apply Jensen's formula to count the number of zeros $N_r = \#\{\beta \in \D(0,r) : Z(\beta) = 0\}$ in a smaller disk of radius $r < R$:
\begin{align}
    N_r &\leq \lr{\log \frac{R}{r}}^{-1} \frac{1}{2\pi} \int_0^{2\pi} d\theta \log \left|\frac{Z(Re^{i\theta})}{Z(0)}\right|.
\end{align}
Taking the disorder average of both sides and applying Markov's inequality would then give a zero-freeness result that holds with high probability. However, this requires controlling the quenched quantity $\E \log |Z(Re^{i\theta})|$. We apply Jensen's inequality to obtain a bound in terms of the simpler annealed second moment, i.e.,
\begin{align}
    \E \log \left|\frac{Z(Re^{i\theta})}{Z(0)}\right| \leq \frac{1}{2} \log \E \left|\frac{Z(Re^{i\theta})}{Z(0)}\right|^2.
\end{align}
If the RHS is a harmonic function, then the Jensen integral vanishes, giving $N_r = 0$ in expectation.
To the best of our knowledge, this method was first used for random polynomials~\cite{hughes2008zeros}, then similarly applied by~\cite{eldar2018approximating} for the matrix permanent and by~\cite{bencs2025zeros} in a rigorous treatment of the Sherrington-Kirkpatrick (SK) model. The SK model is replica symmetric for all $0 < \beta < 1$, and it has a zero-free region in the $|\beta| < 1$ disk. In contrast, the SYK model is known to have complex zeros at $|\beta| = O(1)$ even if it is replica symmetric for all $0 < \beta < \infty$: these zeros correspond to the times at which the spectral form factor vanishes. Consequently, a more careful analysis than the above disk is required (\Cref{app:sykjensen}).

To compute $\E |Z(\beta)|^2$, we introduce two copies of the SYK model labeled by replica indices $a \in \{L,R\}$. Changing to the standard $G, \Sigma$ variables gives $\E |Z(\beta)|^2 = \int DG \, D\Sigma\, \exp[-n S[G, \Sigma]]$ for replica symmetric action
\begin{align}\label{eq:Smain}
    S[G,\Sigma] &= - \log \pf\lr{\partial_\tau -\hat\Sigma_{aa}} \nonumber\\
    &\quad + \frac{1}{2} \int_0^{|\beta|} d\tau_1\,d\tau_2\, \Bigg(\Sigma_{aa}(\tau_1,\tau_2) G_{aa}(\tau_1,\tau_2) \nonumber\\
    &\quad - \frac{J^2}{q} s_{aa} G_{aa}(\tau_1,\tau_2)^q\Bigg),
\end{align}
where $\pf$ denotes the Pfaffian, $\hat \Sigma$ denotes the integral operator, and
\begin{align}
    s_{LL} = \lr{\frac{\beta}{|\beta|}}^2, \quad s_{RR} = \lr{\frac{\bar\beta}{|\beta|}}^2.
\end{align}
In the large-$q$ limit, the translation-invariant solution can be written as
\begin{align}
	G_{aa}(\tau_1,\tau_2) &= \frac{1}{2}\sgn(\tau_1-\tau_2)\lr{1 + \frac{g_{aa}(\tau_1,\tau_2)}{q}}
\end{align}
for (using standard substitution $\cJ^2 = J^2 q / 2^{q-1}$)
\begin{align}
    e^{g_{aa}(\tau)} &= \lr{\frac{\cos(c_a/2)}{\cos\left[c_a \lr{\frac{1}{2} - \frac{\tau}{|\beta|}}\right]}}^2\\
    c_a^2 &= s_{aa} \lr{|\beta|\cJ}^2 \cos^2 \frac{c_a}{2}.
\end{align}
Integrating out $\Sigma$ simplifies the action to
\begin{align}
	S[g] &= - \frac{s_{aa}|\beta|\cJ^2}{2q^2} \int_0^{|\beta|/2} d\tau\,  \lr{1 - \frac{1}{2}g_{aa}(\tau)}e^{g_{aa}(\tau)},
\end{align}
which we saddle point approximate by identifying the leading solutions for $c_L, c_R$ that extremize the action and yield a physical (non-diverging) partition function. We show that when a single leading saddle dominates, any complex zeros are confined to the imaginary $\beta$ axis. The first of these zeros on the imaginary axis corresponds to the $|\Im \beta| < 1.3$ condition of \eqref{eq:sykregion} and matches the first dip of the SYK spectral form factor~\cite{khramtsov2021spectral}. For all $\Re \beta \neq 0$, we analytically show that a unique saddle dominates in both the $|\beta| \cJ \ll 1$ and $|\beta| \cJ \gg 1$ limits; we argue (and numerically verify) that it always dominates. Finally, we show that a single dominating saddle corresponds to a harmonic function in the integrand of Jensen's formula, yielding with probability $1-o(1)$ the zero-free region claimed in \eqref{eq:sykregion} and confirmed numerically in Fig.~\ref{fig:sykzeros}. Finally, we extend this argument to the perturbed Hamiltonian $H \to H+\lambda O$ for local observable $O$, which we use to obtain thermal expectations with Barvinok's algorithm.

\PRLsec{Discussion}%
Our results suggest that for many natural models, efficient classical algorithms can estimate local thermal expectations until a phase transition occurs.
Such a phase transition often exhibits glassiness or a divergence of correlation length that obstructs local samplers~\cite{dennis2002topological,gamarnik2021overlap,basso2022performance,gamarnik2024slow,anschuetz2025efficient,zlokapa2025average,chen2025information}, analogous to the classical case~\cite{mezard2009information,gamarnik2021overlap,huang2025strong}. We expect similar obstructions to hold for physical thermalization processes~\cite{davies1974markovian}, e.g., as obtained by quantum simulators~\cite{altman2021quantum}. This suggests that algorithms with quantum advantage for Gibbs state problems must ultimately exploit regimes beyond a phase transition where analyticity fails. Alternatives include using more nonlocal algorithms for Gibbs sampling~\cite{jiang2024quantum,chen2025quantum} or identifying other Gibbs-related tasks, such as rapid mixing within a phase~\cite{bergamaschi2026rapid} or measuring properties of metastable states~\cite{bergamaschi2025structural}.

Our classical algorithm currently has runtime quasi-polynomial in system size.
It may be possible to improve upon this: recent works have found polynomial-time classical algorithms for sampling from high-temperature spin models~\cite{kuwahara2020clustering,mann2021efficient,mann2024algorithmic} and weakly interacting fermions~\cite{chen2025convergence}.
More sophisticated algorithms, such as message-passing algorithms~\cite{biroli2001quantum}, as well as more refined analyses of Barvinok's algorithm~\cite{patel2017deterministic,yao2022polynomial,mann2024algorithmic}, may also achieve polynomial runtimes.

Even if a polynomial-time classical algorithm were obtained, quantum polynomial speedups may still persist. 
An Arrhenius law generically gives a mixing time of $e^{O(\beta)}$ due to a free energy barrier~\cite{bovier2016metastability}, whereas Barvinok's algorithm generically applied on a zero-free strip has doubly exponential dependence on $\beta$ (\Cref{app:disk-to-strip-mapping}). 
For the SYK model, the zero-free region available yields exponential dependence in $\beta$ classically (\Cref{app:disk-to-wedge-mapping}), whereas the quantum mixing time may only be polynomial in $\beta$~\cite{almheiri2024universal,zhang2019evaporation,maldacena2021syk,chen2021fast,schuster2025cooling}. However, proving polynomial quantum mixing times is often difficult, although recent results achieve quasi-polynomial mixing times for lattice models assuming a spectral gap~\cite{bergamaschi2025quantum}.

Gibbs-related tasks beyond thermal expectations may also provide an avenue for quantum advantage. Sampling (i.e., measuring in a fixed basis) is known to be sufficient for quantum advantage at constant temperature~\cite{bergamaschi2024quantum}, although its relevance for physically relevant tasks is less apparent. Quantities that require real time evolution (e.g., estimating the density of states) may also be promising candidates for quantum advantage for natural Hamiltonian ensembles~\cite{brandao2008entanglement,chen2024sparse}.
Although analytic continuation can be applied to $e^{-iHt}$, the exponential (or greater) cost in $t$ and presence of zeros along the imaginary axis prevents Barvinok's algorithm from efficiently computing quantities such as out-of-time correlators~\cite{larkin1969quasiclassical,swingle2016measuring,roberts2017chaos,google2025observation} beyond constant time. 
At constant time, however, we note in~\Cref{app:otoc} that Barvinok's method yields an exponential improvement in error dependence compared to a naive Lieb-Robinson truncation of Heisenberg evolution.

A central theme in the classical theory is that zero-freeness, correlation decay, fast mixing, and algorithmic tractability often coincide in several families of models \cite{dobrushin1985completely,jerrum1993polynomial,weitz2006counting,ding2009mixing,barvinok2016combinatorics,regts2018zero,shao2021contraction,regts2023absence,liu2025correlation,huang2025strong}. These relationships are subtle and can be physically informative: for some natural models, thermal expectation estimation can remain efficient beyond the efficient sampling regime~\cite{bencs2025zeros}.
In the quantum setting, relations between these notions are less clear~\cite{brandao2019finite} and further complicated by phenomena such as entanglement and the sign problem~\cite{chen2025sign,jiang2025positive}. 
Recent work indicates that clarifying these links yields new structural understanding of quantum Gibbs states~\cite{zlokapa2025average,bakshi2025dobrushin}.
Our work suggests that understanding such properties can also help relate the algorithmic transitions of quantum and classical algorithms.

\PRLsec{Acknowledgments}%
We thank Eric Anschuetz, Jordan Cotler, Ike Chuang, Daniel Lee, and Kuikui Liu for enlightening conversations. We especially thank Chi-Fang Chen, Aram Harrow, Saeed Mehraban, and Jiaqing Qiang for comments on a preliminary draft of this work, and Francisco Pernice for introducing us to the classical literature as well as for subsequent discussions and comments. AZ is supported by a Hertz Fellowship.

\bibliography{References.bib}
\bibliographystyle{apsrev4-1}

\clearpage
\onecolumngrid
\appendix
\setcounter{secnumdepth}{3}

\section{Barvinok's interpolation method for quantum Hamiltonians}
\subsection{Partition function and local observables}

To estimate thermal expectations to additive error, we estimate the partition function of a perturbed Hamiltonian to a particular relative error. We describe this procedure following the setting of \cite{bravyi2021complexity}.
\begin{problem}[Quantum Partition Function (QPF)] \label{problem:QPF}
Given a Hamiltonian $H$ acting on $n$ qubits and a precision parameter $\epsilon = O(\poly{1/n})$, compute an estimate $\widetilde{Z}$ to $Z_{H}(\beta) := \Tr[e^{-\beta H}]$ such that
\begin{equation}
(1 - \epsilon)Z_{H}(\beta) \leq \widetilde{Z} \leq (1 + \epsilon)Z_{H}(\beta).
\end{equation}    
\end{problem}

Ref. \cite{bravyi2021complexity} shows that the hardness of QPF (\Cref{problem:QPF}) is essentially equivalent to that of estimating the expected value of observables with respect to the Gibbs measure, which we formalize below.

\begin{problem}[Quantum Mean Value (QMV)] \label{problem:QMV}
Given a Hamiltonian $H$ acting on $n$ qubits, a Pauli or fermionic monomial $A$, and precision parameter $\delta = O(\poly{1/n})$, compute an estimate $\widetilde{\mu}$ to $\iprhob{A}:= \Tr[A \rho_\beta]$ such that
\begin{equation}
\left| \widetilde{\mu}  - \iprhob{A} \right| \leq \delta.
\end{equation}
\end{problem}

The value of an observable can be estimated by solving QPF via estimating a derivative of the log partition function for real-valued $\beta$, i.e.,
\begin{equation}
    \iprhob{A} = -\frac{1}{\beta} \frac{\partial}{\partial \lambda} \log Z_{H + \lambda A}(\beta) \bigg\vert_{\lambda = 0}.
\end{equation}
This was formally shown in \cite{bravyi2021complexity} and restated below.

\begin{lemma}[QMV from QPF; implicit in Lemma 11 of \cite{bravyi2021complexity}] \label{lem:QMV_from_QPF}
    Given a Hamiltonian $H$ acting on $n$ qubits and a Pauli or fermionic monomial $A$, the observable expectation $\iprhob{A}$ in QMV (\Cref{problem:QMV}) can be estimated to error $\delta = 1/\poly{n}$ by solving QPF (\Cref{problem:QPF}) for $Z_{H}(\beta)$ and $Z_{H+\delta A}(\beta)$ to precision $\epsilon=O(\delta^2)$.
\end{lemma}

In other words, the classical algorithm that we study can estimate any bounded observable written as a sum of Pauli or fermionic operators assuming the perturbed model is zero-free. Notably, this does not require an explicit representation of a state (e.g., expressed via a circuit or tensor network).

\subsection{Barvinok's interpolation}
Barvinok's interpolation \cite{barvinok2016combinatorics} is a method introduced to compute $Z_H(\beta)$ from its Taylor series at $\beta = 0$.
This was used in the quantum setting by Ref. \cite{harrow2020classical} and various other works \cite{bravyi2021classical,mann2024algorithmic,wild2023classical}.
Barvinok's interpolation method parametrizes $Z_H(t)$ along the real segment $t\in[0,\beta]$ and approximates $\log Z_H$ by truncating its Taylor series at $t=0$. 
The algorithm succeeds in approximating the complex-valued partition function in its zero-free region. 
We begin by describing the algorithm in the following two regions.

\begin{definition}[Closed and open disk]
    Given $R>0$, the open disk $\mathbb D_R$ is the set $\mathbb D_R:=\{t\in\mathbb{C}:|t| < R\}$ and the closed disk $\overline{\D_R}$ is the set $\overline{\mathbb D_R}:=\{t\in\mathbb{C}:|t| \le R\}$.
\end{definition}

Central to Barvinok's interpolation is an error guarantee of the truncated Taylor series. 
We depart from the standard presentation of Barvinok's interpolation and provide a bound on the truncation error of arbitrary complex functions in their analytic regions below; this will eventually be useful for approximating out-of-time-order correlators. 
Later, we adapt this to apply to partition functions as well.

\begin{lemma}[Taylor truncation under a zero-free/analytic disk]\label{lem:disk-analytic}
Let $f:\mathbb D_R\to\mathbb C$ be holomorphic on the open disk $\mathbb D_R:=\{z\in\mathbb C:|z| < R\}$ with $R>1$.
Choose any radius $\rho$ with $1<\rho<R$. 
Let $T_K(1):=\sum_{r=0}^K \frac{f^{(r)}(0)}{r!}$ and $M_\rho:=\sup_{|z|\le \rho} |f(z)|$.
Then the truncation error is bounded by
\begin{equation}\label{eq:disk-analytic-error}
    \bigl|f(1)-T_K(1)\bigr| \le  \frac{M_\rho}{\rho^{K+1}(1-\frac{1}{\rho})}.
\end{equation}
Equivalently, for any $\epsilon\in(0,1)$ it suffices to choose
\begin{equation}\label{eq:disk-analytic-K}
    K \ge 
    \left\lceil \frac{\log \left(M_\rho\right) - \log\left(\epsilon(1-1/\rho))\right)}{\log \rho} - 1\right\rceil
    \end{equation}
to ensure $\bigl|f(1)-T_K(1)\bigr|\le \epsilon$.
\end{lemma}

\begin{proof}
Denote the coefficients of the Maclaurin series by $a_r$, i.e., $f(z)=\sum_{r\ge 0}a_r z^r$.
Cauchy's estimate on the circle $|z|=\rho$ gives $|a_r|=|f^{(r)}(0)|/r!\le M_\rho/\rho^r$. 
Therefore, $\sum_{r\ge K+1}|a_r|\le M_\rho \sum_{r\ge K+1}\rho^{-r}= M_\rho \rho^{-(K+1)}/(1-1/\rho)$, which yields \eqref{eq:disk-analytic-error} and \eqref{eq:disk-analytic-K}.
\end{proof}

For the free energy $-\frac{1}{\beta} \log Z_H$ to be analytic in the setting of the above lemma, we need $Z_H(t) \neq 0$ in a disk (or the complex $b$-neighborhood of $[0,\beta]$).
Approximating $\log Z_H(t)$ using a Taylor series truncation to order $K=O(\log(n/\epsilon))$ then yields a $(1\pm \epsilon)$ multiplicative approximation to $Z_H(\beta)$.
These Taylor coefficients can be computed for a Hamiltonian
\begin{equation}
H=\sum_{a=1}^{m} h_a,
\end{equation}
by computing moments $\Tr(H^r)$ for $r \leq K$; if each $h_a$ is a Pauli or fermionic monomial, then this is efficient. We now detail the guarantees of Barvinok's interpolation when $Z_H$ is zero-free in a closed disk. 
\begin{theorem}[Barvinok interpolation under a zero-free disk]\label{thm:disk-barvinok}
    Let $H=\sum_{i=1}^m J_i P_i$ be an $n$-qubit Hamiltonian with $m$ Pauli or fermionic monomials and let $Z_H(t)=\Tr(e^{-tH})$. Fix $R>0$ and suppose $Z_H(t)\neq 0$ for all $t$ in the open disk $\mathbb D_R$. For any $|\beta|<R$, we can choose
    \begin{equation}\label{eq:K-choice-disk}
    K = O\left(\frac{1}{1-|\beta|/R}\,\left[\log \left(n+R\|H\|\right)+\log(1/\epsilon)\right]\right), 
    \end{equation}
    so that the implementation of Ref.~\cite{harrow2020classical} of Barvinok's interpolation solves \Cref{problem:QPF}, computing an estimate $\widetilde{Z}$ such that
    \begin{equation}
    (1-\epsilon) Z_H(\beta)   \le   \widetilde{Z}   \le   (1+\epsilon) Z_H(\beta),
    \end{equation}
    using the moments $M_r':=\Tr(H^r)$ for $r\le K$ and total time $O  \left(m^{K} \mathrm{poly}(n,K)\right)$. 
    In particular, if $m,\|H\|,R=\poly{n}$, $\epsilon = 1/\poly{n}$, and $|\beta|\leq (1-\alpha)R$ for $\alpha>0$, then the runtime is $n^{O(\log n)}$ (quasi‑polynomial).
\end{theorem}

\begin{proof}
    This result is implicit in \cite{harrow2020classical} though we provide a self contained proof here to have explicit dependencies on $\beta$, $R$, and $\|H\|$. 
    
    Following the proof of \Cref{lem:disk-analytic}, define $f(t):=\log Z_H(t)$ with expansion $f(t) = \sum_{r=0}^\infty a_r t^r$ for $a_r := f^{(r)}(0)/r!$. Define the degree-$K$ Taylor truncation $T_K(\beta):=\sum_{r=0}^K a_r \beta^r$. 
    Define the following:
    \begin{equation}
        \rho:=\frac{R+|\beta|}{2}\in(|\beta|,R),\qquad q:=\frac{|\beta|}{\rho}\in(0,1),\qquad 
        M_\rho:=\sup_{|t|\le \rho}|f(t)|.
    \end{equation}
    By Cauchy's estimate on the circle $|t|=\rho$,
    \begin{equation}
    |a_r|=\frac{|f^{(r)}(0)|}{r!}\le \frac{M_\rho}{\rho^r}.
    \end{equation}
    We therefore have
    \begin{equation}
    \left| f(\beta)-T_K(\beta)\right|
      \le   \sum_{r=K+1}^\infty |a_r| |\beta|^r
      \le   M_\rho \sum_{r=K+1}^\infty q^{ r}
      =   M_\rho \frac{q^{ K+1}}{1-q}.
    \end{equation}
    Note that $|\beta| < \rho < R $ and by the Borel-Carathéodory theorem,
    \begin{equation}
    M_\rho \leq \frac{2\rho}{R-\rho} \sup_{|t|\le R}\Re f(t) + \frac{R+\rho}{R-\rho}\,|f(0)|
    \leq \frac{2\rho}{R-\rho} \left(n\log 2 + R\|H\|\right) + \frac{R+\rho}{R-\rho} n\log 2,
    \end{equation}
    since $|Z_H(t)|\le 2^n e^{R\|H\|}$. 
    Choose $K$ as in \eqref{eq:K-choice-disk} so that $| f(\beta)-T_K(\beta)|\le \log(1+\epsilon)$. 
    More explicitly, it suffices to choose
    \begin{equation}
        K\ge\left\lceil \frac{\log\left(\dfrac{M_\rho}{\log(1+\epsilon)}\right) - \log(1-q)}{\log(\rho/|\beta|)} \right\rceil.
    \end{equation}
    Note that asymptotically, it thus suffices to choose $K$ satisfying
    \begin{align}\label{eq:kdisk}
        K = \Theta\lr{\frac{1}{1-|\beta|/R}\log \frac{n + R\norm{H}}{\epsilon}}.
    \end{align}
    Our estimate is then $\widetilde{Z}:=\exp  \left(T_K(\beta)\right)$, and
    \begin{equation}
    \frac{\widetilde{Z}}{Z_H(\beta)}   =   \exp  \left(T_K(\beta)-f(\beta)\right)
      \in   \left[e^{-\log(1+\epsilon)}, e^{\log(1+\epsilon)}\right]
      =   \left[(1+\epsilon)^{-1}, 1+\epsilon\right]
      \subseteq   [ 1-\epsilon, 1+\epsilon ]
    \end{equation}
    for $\epsilon\in(0,1/2)$. 

    Now we bound the runtime to compute $f(t)$ for the branch with $f(0)=n\log 2$.
    The implementation of Ref.~\cite{harrow2020classical} computes $a_r$ for $r\le K$ as follows. 
    First compute the moments $\mu_r':=\Tr(H^r)$, which are the coefficients of $Z_H(t)=\sum_{r\ge 0} \mu_r' t^r/r!$, and then convert $\{\mu_r'\}_{r\le K}$ to the \emph{cumulants} $\{\kappa_r\}_{r\le K}$, which are the derivatives of $\log Z_H$ at $0$, via the moment-to-cumulant recurrence costing $O(K^2)$ arithmetic operations:
    \begin{equation}
    \kappa_1=\mu_1',\qquad
    \kappa_r=\mu_r'-\sum_{i=1}^{r-1}\binom{r-1}{i-1}\kappa_i \mu_{r-i}'\quad (r\ge 2),
    \end{equation}
    and $a_r=\kappa_r/r!$. 
    The moment step costs $O(m^r\,\mathrm{poly}(n))$ for each $r\le K$, hence $O(m^K\,\mathrm{poly}(n,K))$ overall. By \eqref{eq:kdisk}, this yields a quasi-polynomial algorithm for $m, \norm{H}, R, 1/\epsilon = O(\poly{n})$ and $|\beta| < (1-\alpha)R$.
\end{proof}

\begin{remark}[Fully polynomial time approximation scheme (FPTAS)]
    To achieve a fully polynomial time approximation scheme (FPTAS), the computation of the $K$-th moment must run in polynomial time in $n$ and $1/\epsilon$. 
    For local, low degree Hamiltonians, computing degree $K=O(\log n)$ moments in polynomial time is often possible from cluster expansions or other methods \cite{mann2024algorithmic,patel2022approximate,wild2023classical}. 
    We hope that future work can also provide polynomial time approximations to moments for the dense and nonlocal instances that we study here.
\end{remark}

\begin{figure}
    \centering
    \begin{minipage}{0.42\textwidth}
        \centering
        \includegraphics[width=\textwidth]{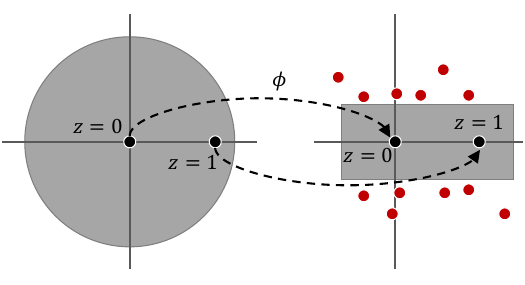}\\
        \small\textbf{(a)} Strip mapping (\Cref{lem:barvinok-strip-map})
        \label{fig:interpolator_strip}
    \end{minipage}
    \hfill
    \begin{minipage}{0.42\textwidth}
        \centering
        \includegraphics[width=\textwidth]{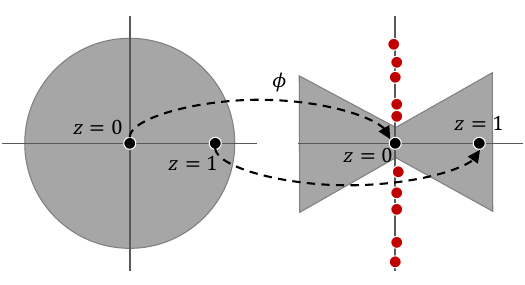}\\
        \small\textbf{(b)} Wedge mapping (\Cref{lem:wedge-avoiding-mapping})
        \label{fig:interpolator_wedge}
    \end{minipage}

    \caption{Illustrative schematic of the mapping $\phi$ used to avoid zeros of the partition function $Z_H(z)$ when applying Barvinok's interpolation. Red dots are example zeros that are avoided by the mapping.}
    \label{fig:interpolator_illustrative}
\end{figure}

If the zero-free regime is a disk, then \Cref{thm:disk-barvinok} can be applied directly.
In practice, however, the zero-free strip may be something else, whereby one then estimates the partition function by composing it with a zero avoiding map $\phi:\mathbb{D}_R \to \mathbb{C}$ that is holomorphic and avoids the zeros in the complex plane \cite{barvinok2016combinatorics,barvinok2018approximating}.
Here, \Cref{thm:disk-barvinok} is applied to the function $f(z) = \log Z_H(\beta \phi(z))$ with a mapping such that $\phi(0)=0$ and $\phi(1)=1$.
Two particular maps $\phi:\mathbb{D}_R \to \mathbb{C}$ that we consider are the mapping from a disk to a strip and a mapping from a disk to a wedge-avoiding region (see \Cref{fig:interpolator_illustrative}). 
We detail these two mappings in the sections that follow.

\subsection{Disk to strip}
\label{app:disk-to-strip-mapping}

We include here Barvinok's mapping of a disk to a strip.

\begin{lemma}[Disk to strip; Lemma 2.2.3 of \cite{barvinok2016combinatorics}] \label{lem:disk-to-strip-barvinok}
    For $0 < \rho < 1$, define
    \begin{equation}
        \alpha = \alpha(\rho) := 1 - e^{-\frac{1}{\rho}}, \quad \beta = \beta(\rho) := \frac{1 - e^{-1-\frac{1}{\rho}}}{1 - e^{-\frac{1}{\rho}}} > 1,
    \end{equation}
    and let
    \begin{equation}
    N = N(\rho) := \left\lfloor \left(1 + \frac{1}{\rho}\right) e^{1+\frac{1}{\rho}} \right\rfloor \ge 14, \quad \sigma = \sigma(\rho) := \sum_{m=1}^N \frac{\alpha^m}{m}.
    \end{equation}
    Define
    \begin{equation}
    \phi(z) = \phi_{\rho}(z) := \frac{1}{\sigma} \sum_{m=1}^N \frac{(\alpha z)^m}{m}.
    \end{equation}
    Then $\phi(z)$ is a polynomial of degree $N$ such that $\phi(0) = 0$, $\phi(1) = 1$,
    \begin{equation}
    -\rho \le \Re \phi(z) \le 1 + 2\rho \quad \textit{and} \quad |\Im \phi(z)| \le 2\rho \quad \textit{provided} \quad |z| \le \beta.
    \end{equation}
\end{lemma}

As a corollary, we show the asymptotic dependence of the mapping above on the strip width $\sigma$. 
Later, we will show that this asymptotic dependence, $\deg(\phi) = O(e^{C_1 /\sigma}) $, is optimal.
\begin{lemma}[Adapted from Lemma 2.2.3 of \cite{barvinok2016combinatorics}]\label{lem:barvinok-strip-map}
Given $\sigma>0$ constant, there exists a radius $\hat r=1+e^{-\Theta(1/\sigma)}$ and a polynomial $\phi$ of degree 
$\deg(\phi) = O(e^{C_1 /\sigma}) $ for constant $C_1>0$ and asymptotically small $\sigma$, such that $\phi(0)=0$, $\phi(1)=1$, and for all $|z| \leq \hat r$, $|\Im \phi(z) | \leq \sigma$.
\end{lemma}
\begin{proof}
    This is a direct application of Lemma 2.2.3 of \cite{barvinok2016combinatorics} which we also reproduce in \Cref{lem:disk-to-strip-barvinok} above for completion.
\end{proof}

In our setting, the strip width will always be constant, so this exponentially poor scaling contributes a constant. We now show that the disk-to-strip mapping asymptotically gives truncation degree
\begin{align}
    K = \exp[O(\beta)] \cdot O\lr{\log \frac{n}{\epsilon}},
\end{align}
which is optimal given only a zero-free strip.

\subsubsection{Optimality of zero-free strip mapping}
\label{app:disk-to-strip-optimality}

We now show that Barvinok's disk to strip polynomial in \Cref{lem:barvinok-strip-map} is essentially optimal. This implies that improvements to the classical algorithm for local observables in the absence of a phase transition (which only guarantees a zero-free strip) and for out-of-time-order correlators (which also have a zero-free strip) are optimal unless a better zero-free region can be shown. In the cases of the SYK model and stabilizer codes, we will have a larger zero-free region, leading to the disk-to-wedge map we study in the following subsection.

\begin{proposition}[Disk to strip limitation]\label{prop:disk-strip-optimality}
Set $0<\rho<1$ and $R>1$. 
Assume $\phi$ is holomorphic on the open disk $\mathbb D_R:$, satisfies $\phi(0)=0$, $\phi(1)=1$ and maps $\mathbb D_R$ into the horizontal strip
$\{w\in\mathbb{C}: |\Im w|< \rho\}$.
Then it must hold that
\begin{equation}\label{eq:key}
  \frac{\pi}{2\rho}  \le  \log \left(\frac{R+1}{R-1}\right),
\end{equation}
which implies the asymptotic relationship $R \leq 1+e^{-\Theta(1/\rho)}$ sending $1/\rho \to \infty$.
\end{proposition}

\begin{proof}
Define
\begin{equation}
    g(z) :=  \exp \left(\frac{\pi}{2\rho} \phi(z)\right).
\end{equation}
Our proof will show bounds on the integral of the $\log$ of the function above from $t=0$ to $t=1$ to obtain the stated result: 
\begin{equation}\label{eq:SP_upper_lower}
    \begin{split}
        \frac{\pi}{2\rho} \le \int_0^1 \left|\frac{d}{dt}\log g(t)\right| \, d t \le \log\left( \frac{R+1}{R-1} \right).
    \end{split}
\end{equation}
Because $|\Im \phi|\le \rho$ on $\mathbb D_R$, we have $-\frac{\pi}{2}<\Im \left(\frac{\pi}{2\rho}\phi\right)<\frac{\pi}{2}$, and $\Re g(z)>0$ for all $z\in \mathbb D_R$. 
Therefore $g$ maps $\mathbb D_R$ into the right half‑plane $H:=\{w:\Re w>0\}$.

Define the function $C(w)=(w-1)/(w+1)$ which maps $H\to\mathbb{D}$, and define the composite function $F:=C\circ g$, which maps $\mathbb D_R$ into the unit disk $\mathbb{D}$. 
The Schwarz-Pick inequality gives
\begin{equation}\label{eq:SPbeta}
    |F'(z)|  \le  \frac{R}{R^2-|z|^2} \left(1-|F(z)|^2\right)
\end{equation}
for all $z\in \mathbb D_R$.
Using $C'(w)=2/(w+1)^2$ and the bound below for $\Re w>0$,
\begin{equation}
    1-|C(w)|^2  =  \frac{4 \Re w}{|w+1|^2},
\end{equation}
we have $F'(z)=C'(g(z)) g'(z)$ and $1-|F(z)|^2=4 \Re g(z)/|g(z)+1|^2$. Plugging these into \eqref{eq:SPbeta} gives
\begin{equation}\label{eq:metric-contract}
  \frac{|g'(z)|}{\Re g(z)}  \le  \frac{2R}{R^2-|z|^2}.
\end{equation}
Along the segment of the real line $[0,1]\subset \mathbb D_R$,
$\left|\frac{d}{dt}\log g(t)\right| \leq \frac{|g'(t)|}{\Re g(t)}$
and 
\begin{equation}
  \int_0^1 \frac{|g'(t)|}{\Re g(t)} dt
   \ge  \int_0^1 \left|\frac{d}{dt}\Re\log g(t)\right| dt
   \ge  \left|\Re\log g(1)-\Re\log g(0)\right|.
\end{equation}
Note that $\log g(z) = \frac{\pi}{2\rho} (\phi(z)-\phi(0))$. 
Using the assumptions $\phi(0)=0$ and $\phi(1)=1$ with \eqref{eq:metric-contract} and integrating from $t=0$ to $t=1$ gives
\begin{equation}
  \frac{\pi}{2\rho}
   \le  \int_0^1 \frac{2R}{R^2-t^2} dt
   =  \log \left(\frac{R+1}{R-1}\right),
\end{equation}
which gives \eqref{eq:key}.
\end{proof}

\begin{corollary}\label{cor:strip-optimal-trunc}
Let us implement Barvinok's interpolation on $Z_H(\beta)$ assuming that $Z_H(z)$ has zeros along the strip at points $z=a+ib$ for $|b|\ge b_0$ and $\|H\|=O(n)$. 
To avoid these zeros we estimate $F(z):=\log Z_H(\beta \phi(z))$
with $\phi$ as in \Cref{prop:disk-strip-optimality} setting $\rho:=b_0/\beta$ with $R=1+e^{-\Theta(\beta)}$. 
Then for any $r$ such that $1<r<R$, $F$ is analytic on $\overline{\mathbb D_r}$, and by \Cref{lem:disk-analytic} the degree‑$K$ truncation of $F$ evaluated at $z=1$ attains error at most $\epsilon$ only if $K \ge  \Omega(\exp(\Omega_\beta(\beta))\log(n/\epsilon))$,
asymptotically matching the guarantees from the disk to strip map of \Cref{lem:barvinok-strip-map}.
\end{corollary}

\subsection{Disk to wedge}
\label{app:disk-to-wedge-mapping}
We now provide a holomorphic mapping from the disk to a wedge which avoids zeros near the imaginary axis.

\begin{lemma}[Wedge-avoiding disk mapping]\label{lem:wedge-avoiding-mapping}
Fix $0<\rho<1$ and a wedge half-angle $\delta_\theta\in(0,\pi/2)$. Set $p:=1-\frac{2\delta_\theta}{\pi}\in(0,1)$ and define
\begin{equation}
    Q(s):=\frac{1+s}{1-s},\qquad
    G(z):=Q\left(\left(\frac{z}{R}\right)^2\right)^p.
\end{equation}
The wedge-avoiding interpolator takes the form
\begin{equation}
    \phi(z) := \rho z \sqrt{ \frac{G(z)^2-1}{z^2} }
\end{equation}
where $(\cdot)^p$ denotes the principal branch. The square root is chosen to be the branch holomorphic on $\mathbb{D}_R$ that is positive at $z=0$ ensuring $\phi(0)=0$.
Set $R>1$ to solve
\begin{equation}\label{eq:radius-wedge-explicit}
    R^2 = \frac{M+1}{M-1},\qquad
    M:=\bigl(1+\rho^{-2}\bigr)^{1/(2p)},
\end{equation}
so that $G(1) = \sqrt{1+\rho^{-2}}$.
Then $\phi$ is holomorphic on $\mathbb D_R:=\{|z|<R\}$, satisfies $\phi(0)=0$, $\phi(1)=1$, and
\begin{equation}
\begin{split}
    \phi(\mathbb D_R) \subset\ \Omega_{\rho,\delta_\theta}
    &:=\mathbb{C}\setminus \left( W^{(+)}_{\rho,\delta_\theta} \cup W^{(-)}_{\rho,\delta_\theta}\right), \\
    W^{(+)}_{\rho,\delta_\theta} &= i\rho+\{re^{i\theta}: r\ge0, |\theta-\pi/2|\le \delta_\theta\}, \\
    W^{(-)}_{\rho,\delta_\theta} &= -i\rho+\{re^{i\theta}: r\ge0, |\theta-3\pi/2|\le \delta_\theta\}.
\end{split}
\end{equation}
Asymptotically as $\rho \to 0$, $R = 1+\Theta(\rho^{1/p})$ and setting $r = (1+R)/2$, $\max_{z \in \mathbb{D}_r} |\phi(z)| \leq 2^p(1+O(\rho^{1/p}))$.
\end{lemma}

\begin{proof}
To shorten notation, we define $t(z) := (z/R)^2$.
For $z \in \mathbb D_R$, $|t(z)|<1$ and thus $Q(t(z))$ maps $\mathbb D_R$ to the right half-plane $H=\{\Re w>0\}$. 
Using the principal power, $G(z)=Q(t(z))^p$ is holomorphic on $\mathbb D_R$ with $\arg G(z)\in(-p\pi/2,p\pi/2)$ and $G(0)=1$.
Note that for $z \in \mathbb D_R$, by Taylor expansion,
\begin{equation}
    G(z)^2-1 = \frac{4p}{R^2} z^2 + O(z^4),
\end{equation}
which has a zero of order 2 at $z=0$.
Consider $h(z):=\frac{G(z)^2-1}{z^2}$ with removable singularity at $z=0$ where we set $h(0)=4p/R^2$. 
By Taylor expansion, $G(z)^2-1 = \frac{4p}{R^2} z^2 + O(z^4)$.
Since $G(z)$ maps to the sector $|\arg w| < p\pi/2 < \pi/2$, $G(z)^2=1$ implies $G(z)=1$, which requires $Q(t(z))=1$ and thus $z=0$. Since $h(z)$ vanishes if and only if $G(z)^2 = 1$, this implies that $h(z)$ is non-vanishing on $\mathbb D_R$.
Since $\mathbb{D}_R$ is simply connected and $h(z)$ is non-vanishing, a holomorphic branch of the logarithm $\log h(z)$ exists. We define $\sqrt{h(z)} := \exp(\frac{1}{2} \log h(z))$ using the branch where $\sqrt{h(0)} > 0$.
Finally, defining $\phi(z) := \rho z \sqrt{h(z)}$ yields a function holomorphic on $\mathbb D_R$ satisfying $\phi(0)=0$.
The condition $\phi(1)=1$ follows from the choice of $R$.

To see that $\phi$ avoids the wedges $W^{(\pm)}$,
consider the map $\Psi(w):=1+(w/\rho)^2$. The forbidden wedges at $\pm i\rho$
are mapped by $\Psi$ onto the set
\begin{equation}
    S_{\delta_\theta}:=\{re^{i\vartheta}: r\ge0, |\vartheta-\pi|\le 2\delta_\theta\}.
\end{equation}
For our map, $\Psi(\phi(z))=G(z)^2$, so
$\arg\Psi(\phi(z))=\arg(G(z)^2)\in(-p\pi,p\pi)$.
Because $p=1-\frac{2\delta_\theta}{\pi}$, we have
$(-p\pi,p\pi)\subset(-\pi+2\delta_\theta, \pi-2\delta_\theta)$. 
Therefore, $\Psi(\phi(z))$ avoids
$S_{\delta_\theta}$ so $\phi(z)$ avoids the corresponding wedges in the $w$-plane proving
$\phi(\mathbb D_R)\subset\Omega_{\rho,\delta_\theta}$.

For the asymptotics, note that as $\rho\downarrow0$ we have
$M=(1+\rho^{-2})^{1/(2p)}=\rho^{-1/p}(1+o(1))$ and $R = 1 + \Theta(1/M)$.
Furthermore, by maximum modulus principle, $\max_{z \in \mathbb{D}_r} |\phi(z)| = \max_{\theta \in [0,2\pi]} |\phi(re^{i\theta})|$ and
\begin{equation}
\begin{split}
    \max_{\theta \in [0,2\pi]} |\phi(re^{i\theta})| &\le \rho \sqrt{|G(re^{i\theta})|^2+1} \\
    &\le \rho \sqrt{\left(\frac{1+|\frac{r}{R}|^2}{1-|\frac{r}{R}|^2} \right)^{2p}+1}.
\end{split}
\end{equation}
Setting $R=1+\epsilon$, we have $\frac{1+|\frac{r}{R}|^2}{1-|\frac{r}{R}|^2}=2/\epsilon + O(1)$ as $\epsilon \to 0$. Since $\epsilon = \rho^{1/p} + O(\rho^{2/p})$,
\begin{equation}
\begin{split}
    \max_{\theta \in [0,2\pi]} |\phi(re^{i\theta})| &\le 2^p (1+O(\rho^{1/p})).
\end{split}
\end{equation}
\end{proof}

When applying the above mapping, we will always consider Hamiltonian ensemble that have a zero-free region containing all $|\Re \beta| > 0$. For this region, we now show that the disk-to-wedge mapping asymptotically gives truncation degree (for asymptotically large $n, \beta, 1/\epsilon$)
\begin{align}
    K = O\lr{\beta \log \frac{\beta n}{\epsilon}},
\end{align}
which is optimal given zeros on the imaginary axis.

\subsubsection{Optimality of wedge mapping}

Now, we show that as long as there are zeros on the imaginary axis that must be avoided, the mapping in \Cref{lem:wedge-avoiding-mapping} is asymptotically optimal.

\begin{proposition}\label{prop:optimality-wedge-mapping}
Let $0<\rho<1$ and $R>1$. Define regions
\begin{equation}
    \Omega_\rho:=\mathbb C\setminus\big(i[ \rho,\infty)\cup i(-\infty,-\rho]\big),\qquad
    \mathbb D_R:=\{z\in\mathbb C:\ |z|<R\}.
\end{equation}
If $\phi$ is holomorphic on $\mathbb D_R$, satisfies $\phi(\mathbb D_R)\subset\Omega_\rho$, and
$\phi(0)=0$, $\phi(1)=1$, then 
\begin{equation}\label{eq:main_imaginary_impossible}
    \log \left(\frac{R+1}{R-1}\right)
    \ge \tfrac12\log \big(1+\rho^{-2}\big).
\end{equation}
Therefore, asymptotically $R \le 1+2 \rho + O(\rho^2)$ as $\rho \to 0$. 
\end{proposition}

\begin{proof}
We use the principal branch of the square root defined on
$\mathbb C\setminus(-\infty,0]$ so that $\Re\sqrt{z}>0$. Define the function
\begin{equation}
    g(z):=\sqrt{ 1+(\phi(z)/\rho)^2 }.
\end{equation}
The image of $g$ is contained in the right half-plane $H:=\{z:\Re z>0\}$.
More formally, for $w\in\Omega_\rho$ we claim $1+(w/\rho)^2\notin(-\infty,0]$. 
Indeed, if
$1+(w/\rho)^2\in(-\infty,0]$ then $(w/\rho)^2=-t$ with $t\ge1$ resulting in
purely imaginary $w$ satisfying $|\Im w| \geq \rho $, which is a contradiction.
Thus, $g$
is holomorphic on $\mathbb D_R$ and satisfies $\Re g(z)>0$.

Now we proceed to apply the Schwarz-Pick inequality.
Let $C(w)=(w-1)/(w+1)$ be the Cayley map $H\to\mathbb D$. 
Then $F:=C\circ g$ maps $\mathbb D_R$ into
$\mathbb D$, and the Schwarz-Pick inequality gives for $z\in \mathbb D_R$:
\begin{equation}
    |F'(z)| \le \frac{R}{R^2-|z|^2} \left(1-|F(z)|^2\right).
\end{equation}
Using $C'(w)=2/(w+1)^2$ and $1-|C(w)|^2=4 \Re w/|w+1|^2$ for $\Re w>0$, we obtain again for $z\in \mathbb D_R$:
\begin{equation}\label{eq:local}
  \frac{|g'(z)|}{\Re g(z)} \le \frac{2R}{R^2-|z|^2}.
\end{equation}
For a holomorphic branch of $\log g$
where $(\log g)'=g'/g$, we have for $t\in[0,1]$,
\begin{equation}
  \left|\frac{d}{dt}\Re\log g(t)\right|
  \le \left|\frac{d}{dt}\log g(t)\right|
  = \frac{|g'(t)|}{|g(t)|}
  \le \frac{|g'(t)|}{\Re g(t)}.
\end{equation}
Integrating and using $g(0)=1$, $g(1)=\sqrt{1+\rho^{-2}}$ yields
\begin{equation}
\begin{split}
    \log \big(\sqrt{1+\rho^{-2}}\big)
    &=\left|\Re\log g(1)-\Re\log g(0)\right|
    \le \int_0^1 \frac{|g'(t)|}{\Re g(t)} dt \\
    &\le \int_0^1 \frac{2R}{R^2-t^2} dt
    = \log \left(\frac{R+1}{R-1}\right),
\end{split}
\end{equation}
which is \eqref{eq:main_imaginary_impossible}. 
\end{proof}

\begin{corollary}\label{cor:wedge-optimal-trunc}
Let us implement Barvinok's interpolation on $Z_H(\beta)$ assuming that $Z_H(z)$ has zeros along the imaginary axis at points $z=ib$ for $|b|\ge b_0$ (i.e., beyond $\pm i b_0$) and $\|H\|=O(n)$. 
To avoid these zeros we estimate $F(z):=\log Z_H(\beta \phi(z))$
with $\phi$ as in \Cref{prop:optimality-wedge-mapping} setting $\rho:=b_0/\beta$ with $R=1+\Theta(1/\beta)$. 
Then for any $1<r<R$, $F$ is analytic on $\overline{\mathbb D_r}$, and by \Cref{lem:disk-analytic} the degree‑$K$ truncation of $F$ evaluated at $z=1$ attains error at most $\epsilon$ only if $K \ge \Omega_\beta(\beta \log(n/\epsilon))$,
asymptotically matching the guarantees from the wedge‑avoiding map of \Cref{lem:wedge-avoiding-mapping} in the $\delta_\theta=0$ case.
\end{corollary}

\section{Bounded-degree Hamiltonians}
\label{app:hight}
In this section, we consider a Hilbert space of dimension $2^n$ corresponding to either $n$ spins/qubits or $2n$ fermionic Majorana modes. We will study $k$-local, degree $D$ Hamiltonians of the form $H=\sum_{a\in\mathcal{A}} h_a$ indexed by a set $\mathcal{A}$ under the following assumptions:
\begin{itemize}
    \item Locality: each $h_a$ is Hermitian and supported on at most $k$ sites, i.e., $|\supp h_a|\le k$. 
    For fermionic Hamiltonians, we also assume the size of the support $|\supp h_a|$ is even for all terms.
    \item Bounded norms: $\|h_a\|\le 1$ for all terms in the Hamiltonian.
    \item Bounded degree: each site belongs to at most $D$ terms from $\mathcal A$, i.e. for each site $i$, it holds that $\left| \left\{ a \in \mathcal{A}: i \in \supp h_a \right\} \right| \leq D$.
\end{itemize}
For example, for spin systems, the terms above may take the form of Pauli monomials $h_a\in\{I,X,Y,Z\}^{\otimes n}$.
In the fermionic setting, $h_a$ consists of Majorana monomials $h_a=\Gamma_{S_a}:=i^{|S_a|/2}\prod_{j\in S_a}\gamma_j$ on $2n$ Majorana modes.
In either the spin or fermionic case, the set of these monomials is an orthonormal basis of operators with respect to the normalized Hilbert-Schmidt inner product, e.g. for Majorana monomials $\{\Gamma_S:\ S\subseteq[2n]\}$:
\begin{equation}
\label{eq:HS-orthonormality}
    \frac{1}{2^n}\Tr\big(\Gamma_S^\dagger \Gamma_T\big)=1[S=T].
\end{equation}
We show a zero-free region for the partition function $Z_H(\beta)$ for Hamiltonians taking the above form.

\begin{theorem}\label{thm:bounded_degree}
Given a degree $D$, $k$-local Hamiltonian $H=\sum_{a\in\mathcal{A}} h_a$ where $\|h_a\|\leq 1$ as detailed above, there exists constant $C_{\max}=(2e)^{-1}$ such that
\begin{equation}
    Z_H(\beta)\neq 0 \quad\text{for all}\quad |\beta|< \frac{C_{\max}}{k(D-1)+1} , \quad C_{\max}:=\frac{1}{2e} \approx 0.18394 .
\end{equation}
Equivalently, on the disc $|\beta|<C_{\max}/(k(D-1)+1)$ there is a single-valued analytic branch of $f(\beta)=\log Z_H(\beta)$. 
\end{theorem}

This is an improvement over prior work focused on the spin setting from
Mann and Minko \cite{mann2024algorithmic} who show that $Z_H(\beta)$ is zero-free whenever $|\beta|\le (e^4D\binom{k}{2})^{-1} =  O(1/(k^2D))$. 
This also improves by a constant factor over the result in Ref. \cite{harrow2020classical} which proves zero-freeness whenever $|\beta|\le (e(2e-1)kD)^{-1}$. 
As far as we are aware, our bound on $|\beta|$ effectively subsumes all efficient regimes previously proven for other algorithms applied to quantum Hamiltonians of bounded degree and locality \cite{harrow2020classical,mann2021efficient,yao2022polynomial,rouze2024optimal,wild2023classical,mann2024algorithmic,bakshi2024high,bakshi2025dobrushin}. 

\begin{remark}[Optimality]
    The dependence $\beta = O(1/(kD))$ is related to the fact that $\Theta(kD)$ sites ``influence" a given qubit.
    Under complexity theoretic assumptions, the dependence of our bound on degree cannot be improved beyond a constant factor as \cite{sly2014counting} show that estimating the partition function of a classical degree $D$ Ising model is NP-Hard at $\beta = \Theta(1/D)$ corresponding to a computational hardness threshold established for models on trees with degree $D$ \cite{weitz2006counting} (see also \cite{patel2023near} and Theorem 9 of \cite{mann2024algorithmic} for related results). 
\end{remark}

To prove this result, we trotterize the dynamics and convert the partition function to a hard-core model. Then we use the Koteck\'y--Preiss (KP) bound to show that the model is zero-free for $|\beta|$ small enough. 
Our proof follows a similar strategy to that in a series of works \cite{park1982cluster,ueltschi2003cluster,mann2024algorithmic,nguyen2024high} specializing the technique to a specific Trotterization of the model and optimizing the choices for applying the KP bound.
Convergence criteria subsuming the KP bound exist, such as the Dobrushin or Fernandez-Procacci conditions \cite{fernandez2007cluster}. 
We were unable to improve the explicit bound below using those conditions however.

\subsection{Trotterization and factorization}
Fix an ordering of the elements in $\mathcal{A}=\{1,\dots,A\}$ and let $m\in\mathbb{N}$. 
Set the first-order Trotterization as
\begin{equation}
    U_m(\beta) := \left(\prod_{a\in\mathcal{A}} e^{-(\beta/m) h_a}\right)^{ m}.
\end{equation}
By the Trotter product formula, $e^{-\beta H} = \lim_{m\to\infty} U_m(\beta)$ \cite{trotter1959product}.

Let $\mathcal{U}:=[A]\times[m]$ be the set of \emph{events} in both space and time.
For each tuple $(j,s)\in\mathcal{U}$, we define
\begin{equation}\label{eq:Ba-def}
    B_{j,s} := e^{-(\beta/m) h_j} - I.
\end{equation}
Note that $B_{j,s}$ depends only on $j$, but we keep the index $s$ to indicate that $B_{j,s}$ and $B_{j,s'}$ correspond to different events for $s\neq s'$.
For any $\varepsilon_1 > 0$ there exists $m_0(\varepsilon_1,\beta)$ such that, for all $m\ge m_0$,
\begin{equation}\label{eq:T-bound}
    \|B_{j,s}\| \le e^{|\beta|/m}-1 \le (1+\varepsilon_1) \frac{|\beta|}{m} =: T.
\end{equation}
Our application of the KP condition will use $T$ to upper bound the magnitude of influence from terms relating to $B_{j,s}$.

For $F\subseteq\mathcal U$, define the ordered product
\begin{equation}\label{eq:ordered-prod}
    \prod_{(j,s)\in F}^{\rightarrow} B_{j,s}
     := \prod_{s=1}^m \prod_{j=1}^A \widetilde B_{j,s},
    \qquad
    \widetilde B_{j,s} := \begin{cases}
    B_{j,s}, & (j,s) \in F,\\
    I, & (j,s) \notin F.
    \end{cases}
\end{equation}
Note that $|F|$ counts multiplicity across time slices since $(j,s) \neq  (j,s')$ when $s\neq s'$.
Expanding the factor $e^{-(\beta/m)h_j}=I+B_{j,s}$ in each Trotter layer and multiplying layers yields
\begin{equation}\label{eq:Zm-factor}
    \Tr U_m(\beta)  =  2^n \sum_{F\subseteq \mathcal{U}} \zeta_m(F),\qquad \zeta_m(F):= \frac{1}{2^n}\Tr \left[\prod_{(j,s)\in F}^{\rightarrow} B_{j,s}\right].
\end{equation}

To bound $|\zeta_m(F)|$,
use \Cref{eq:T-bound} and submultiplicativity of the operator norm so that
\begin{equation}\label{eq:zeta-basic-bound}
    |\zeta_m(F)|
    \le \left\|\prod_{(j,s)\in F}^{\rightarrow} B_{j,s}\right\|
    \le T^{|F|}.
\end{equation}
We also note that the trace over disjoint $\zeta_m(F)$ and $\zeta_m(F')$ factorizes.
Namely, for $F\subseteq\mathcal U$, we define its spatial support as
\begin{equation}\label{eq:supp-F}
    \supp(F) := \bigcup_{(j,s)\in F}\supp(h_j).
\end{equation}
If $\supp(F)\cap\supp(G)=\emptyset$, then all operators in the two ordered products act on disjoint tensor components and thus
\begin{equation}\label{eq:zeta-factorization}
    \supp(F)\cap\supp(G)=\emptyset
    \quad\Longrightarrow\quad
    \zeta_m(F\cup G)=\zeta_m(F) \zeta_m(G).
\end{equation}
For sake of completeness, we formalize and prove this fact below.

\begin{lemma}[Factorization over disjoint supports] \label{lem:factorization_disjoint}
    Given $\zeta_m(F)$ as defined above for $F\subseteq\mathcal U$, assume one of the following settings holds:
    \begin{enumerate}
        \item[(i)] \textbf{(Spins)} each $h_a$ is supported on some subset of qubits,
        \item[(ii)] \textbf{(Fermions)} each local term $h_a$ consists of an \emph{even} number of Majorana operators (i.e., parity-preserving).
    \end{enumerate}
    Then, if $F,G \subseteq \mathcal U$ satisfy $\supp(F)\cap\supp(G)=\emptyset$, it holds that $\zeta_m(F\cup G) = \zeta_m(F) \zeta_m(G)$.
\end{lemma}
\begin{proof}
Let us denote for simplicity,
\begin{equation}
    A_F:=\prod_{u\in F}^{\rightarrow}B_u,
    \qquad
    A_G:=\prod_{u\in G}^{\rightarrow}B_u.
\end{equation}
In case (i), operators supported on disjoint qubit sets commute.
In case (ii), even operators supported on disjoint Majorana sets also commute.
Thus, in either case every term in $A_F$ commutes with every term in $A_G$, so by swapping terms, 
\begin{equation}
    \prod_{u\in F\cup G}^{\rightarrow}B_u = A_F A_G.
\end{equation}
Now we show the normalized trace factorizes.
In case (i), we can choose a basis of the vector space so that $A_F=A_F^{(F)}\otimes I\otimes I$ and $A_G=I\otimes A_G^{(G)}\otimes I$, which implies 
\begin{equation}
    2^{-n}\Tr(A_F A_G) = \left(2^{-|\supp(F)|}\Tr A_F^{(F)}\right) \left(2^{-|\supp(G)|}\Tr A_G^{(G)}\right)
    = (2^{-n}\Tr A_F) (2^{-n}\Tr A_G).
\end{equation}
In case (ii), expanding $A_F$ and $A_G$ in the orthonormal Majorana basis (\Cref{eq:HS-orthonormality}), we write
\begin{equation}
    A_F=\sum_{X\subseteq \supp(F)} \alpha_X \Gamma_X,
    \qquad
    A_G=\sum_{Y\subseteq \supp(G)} \beta_Y \Gamma_Y.
\end{equation}
We have $\Gamma_X\Gamma_Y =(-1)^{s(X,Y)}\Gamma_{X\cup Y}$ where $s(X,Y) \in \{0,1\}$ and $s(\emptyset,\emptyset)=0$. 
Using $2^{-n}\Tr(\Gamma_S)=\mathbf 1[S=\emptyset]$, we obtain
\begin{equation}
    2^{-n}\Tr(A_F A_G) = \sum_{X,Y}(-1)^{s(X,Y)}\alpha_X\beta_Y 2^{-n}\Tr(\Gamma_{X\cup Y})
    =\alpha_\emptyset \beta_\emptyset
    =(2^{-n}\Tr A_F)(2^{-n}\Tr A_G).
\end{equation}
\end{proof}

\subsection{Event graph and polymers}

We embed a graph structure on the set of events $\mathcal{U}=[A]\times[m]$ by connecting $(j,s)$ and $(j',s')$ whenever their supports overlap. 
More formally, define the event graph $G_m=(\mathcal U,E_m)$ by putting an edge between
$(j,s)$ and $(j',s')$ if and only if $\supp(h_j) \cap \supp(h_{j'}) \neq \emptyset$
(time indices do not affect adjacency).
In particular, $(j,s)$ is adjacent to $(j,s')$ for all $s'\neq s$.
Let $N(u)$ denote the open neighborhood of $u\in\mathcal U$, and define the maximum closed neighborhood size
\begin{equation}\label{eq:closed_neighborhood_def}
    \Delta_m := \max_{u\in\mathcal U} \left|\{u\}\cup N(u)\right|.
\end{equation}
With locality bounded by $k$ and the degree bound $D$, each term $h_j$ overlaps with at most $k(D-1)$ other terms, so the number of terms overlapping $h_j$ including itself is at most $k(D-1)+1$. 
For each such term there are $m$ time slices. 
Thus
\begin{equation}\label{eq:Delta-hat-m}
    \Delta_m \le m \bigl(k(D-1)+1\bigr).
\end{equation}

Following terminology in the physics literature, a \emph{polymer} is a nonempty connected vertex set $\gamma\subseteq\mathcal U$ in the event graph $G_m$.
Let $\mathcal C_m$ denote the set of all polymers.
For a subset $F\subseteq\mathcal U$, let $F=\gamma_1\sqcup\cdots\sqcup\gamma_r$
be the decomposition of $F$ into connected components of the induced subgraph $G_m[F]$.
By definition of $G_m$, disconnected components have disjoint spatial supports, and \Cref{eq:zeta-factorization} gives
\begin{equation}\label{eq:zeta-components}
    \zeta_m(F)=\prod_{i=1}^r \zeta_m(\gamma_i).
\end{equation}

The event graph induces a graph on the vertex set $\mathcal C_m$ which we call the polymer graph $G(\zeta_m)$. 
$G(\zeta_m)$ has vertex set $\mathcal C_m$ with edges connecting polymers $\gamma\sim\gamma'$ if and only if their spatial supports overlap, i.e. $\supp(\gamma)\cap\supp(\gamma')\neq\emptyset$.
Let $\mathcal I(G(\zeta_m))$ contain the independent sets of $G(\zeta_m)$.
Then the map $F\mapsto\{\gamma_1,\dots,\gamma_r\}$ is a bijection between subsets $F\subseteq\mathcal U$ and independent sets $\Gamma\in\mathcal I(G(\zeta_m))$, with inverse $\Gamma\mapsto \bigcup_{\gamma\in\Gamma}\gamma$.
Thus, \Cref{eq:Zm-factor} can be written as a hard-core polymer model
\begin{equation}\label{eq:cluster_polymer_model}
    \Tr U_m(\beta) = 2^n \Xi_m(\beta), \qquad
    \Xi_m(\beta)
     := \sum_{\Gamma\in\mathcal I(G(\zeta_m))}\ \prod_{\gamma\in\Gamma}\zeta_m(\gamma).
\end{equation}
Note, that in physics literature, it is more common to denote adjacent polymers as incompatible with the relationship $\gamma \nsim \gamma'$ if the union of their supports overlap.
We will instead adopt the notation here to more clearly denote independent sets.

We will need to bound the number of polymers of a given size.  
To do this, we will use an exponential size bound on the number of connected sets of size $s$ which we share below.
\begin{lemma}[Connected-set counting]\label{lem:connected-count}
    Let $G$ be any finite graph with maximum degree $\Delta\ge 2$. Fix a vertex $u$. Let $N_s(u)$ be the number of connected vertex sets $S\subseteq V(G)$ with $u\in S$ and $|S|=s$. Then
    \begin{equation}
        N_s(u) \le \frac{1}{s} \binom{s\Delta}{s-1}.
    \end{equation}
\end{lemma}
\begin{proof}
Note that every connected vertex set has a spanning tree. Lemma 2.1 of \cite{borgs2013left} shows that the number of subtrees of $s$ nodes containing $u$ is at most $\frac{1}{s} \binom{s\Delta}{s-1}$.
Note, Ref. \cite{borgs2013left} also show that $N_s(u)  \le  \frac{(e\Delta)^{ s-1}}{2}$ for all $s\ge 2$.
\end{proof}

\subsection{Koteck\'y--Preiss bound}

We will apply the KP convergence criteria which is formalized below.

\begin{theorem}[Kotecký-Preiss condition \cite{kotecky1986cluster}]\label{prop:KP}
Let $\mathcal{C}$ denote the set of polymers in the hardcore model of \eqref{eq:cluster_polymer_model}.
If there exists a function $a:\mathcal{C} \to[0,\infty)$ such that for all $\gamma \in \mathcal{C}$,
\begin{equation}\label{eq:KP}
\sum_{\gamma' \sim \gamma} |\zeta(\gamma')| e^{a(\gamma')} \le\ a(\gamma),
\end{equation}
then the cluster expansion for $\log\Xi_m$ converges absolutely and $\Xi_m(\beta) \neq 0$.
\end{theorem}

\begin{lemma}\label{lem:KP_bound}
    In the setting above, with $\varepsilon_1\in(0,1)$ as in \Cref{eq:T-bound}, there exists a positive integer $m_0(\varepsilon_1,\beta)$ such that
    for $m > m_0(\varepsilon_1, \beta)$, the Kotecký-Preiss condition in \Cref{eq:KP} holds when 
    \begin{equation}\label{eq:smallness-y-new}
        (1+\varepsilon_1)|\beta|  <  \frac{C_{\max}}{k(D-1)+1}, \qquad C_{\max}:=\frac{1}{2e}\approx 0.18394.
    \end{equation}
\end{lemma}

\begin{proof}
As a reminder, $N(u)$ denotes the (open) neighbor set of $u\in\cU$ in the event graph and $\Delta_m \le   m\bigl(k(D-1)+1\bigr)$ is the degree bound in the event graph.
Choose $a(\gamma)=\theta|\gamma|$ with the choice $\theta=\frac12$.
For any polymer $\gamma$,
\begin{equation}
\begin{split}
    \sum_{\gamma'\sim\gamma} |\zeta_m(\gamma')|e^{a(\gamma')}
    &\le \sum_{u\in \gamma} \sum_{v \in \{u\} \cup N(u)} \sum_{\gamma'\ni v}
    |\zeta_m(\gamma')| e^{\theta|\gamma'|} \\
    &\le \Delta_m |\gamma|  \sup_{v\in\cU}\sum_{\gamma'\ni v}
    |\zeta_m(\gamma')| e^{\theta|\gamma'|}.    
\end{split}
\end{equation}
For each $v\in\cU$, using $|\zeta_m(\gamma')|\le T^{|\gamma'|}$ and the connected set bound $N_s(v)\le \frac{1}{s}\binom{s\Delta_m}{s-1}$ in \Cref{lem:connected-count} gives
\begin{equation}
    \sum_{\gamma'\ni v}|\zeta_m(\gamma')|e^{\theta|\gamma'|}
     \le  \sum_{s\ge 1} N_s(v) (Te^\theta)^s 
     \le \sum_{s\ge 1} \frac{1}{s} \binom{s\Delta_m}{s-1} (Te^\theta)^s 
     = \mathcal{T}_{\Delta_m}(Te^\theta),
\end{equation}
where the power series $\mathcal{T}_{\Delta}(x)$ is defined as
\begin{equation}
    \mathcal{T}_{\Delta}(x) := \sum_{s\ge 1} \frac{1}{s} \binom{s\Delta}{s-1} x^s.
\end{equation}
Therefore
\begin{equation}
    \sum_{\gamma'\sim\gamma} |\zeta_m(\gamma')|e^{a(\gamma')}
     \le  \Delta_m |\gamma| \mathcal{T}_{\Delta_m}(Te^\theta),
\end{equation}
and the KP condition in \Cref{eq:KP} holds if
\begin{equation}\label{eq:KP-suffices-new}
    \mathcal{T}_{\Delta_m}(Te^\theta) \le \theta/\Delta_m  = \frac{1}{2\Delta_m}.
\end{equation}
We now use the Fuss-Catalan functional equation to evaluate $\mathcal{T}_{\Delta_m}(Te^\theta)$.
Following Ref. \cite{gessel2016lagrange}, the $n$-th order $p$ Fuss-Catalan number is $\frac{1}{pn+1}\binom{pn+1}{n}$. 
Setting $n=s$ and $p=\Delta$ obtains $\frac{1}{\Delta s+1}\binom{\Delta s+1 }{ s} = \frac{1}{s}\binom{s\Delta }{ s-1}$ as in our setting.
Let $c_{\Delta}(x):=1+\mathcal{T}_{\Delta}(x)$. 
Section 3.3 of Ref. \cite{gessel2016lagrange} shows that $c_{\Delta}$ satisfies $c_{\Delta}(x) = 1+x c_{\Delta}(x)^{\Delta}$ or equivalently,
\begin{equation}\label{eq:fuss-eq}
    \mathcal{T}_{\Delta}(x) = x\left(1+\mathcal{T}_{\Delta}(x)\right)^{\Delta}.
\end{equation}
Using the choice $\theta=\frac12$,
let $z_0:=\theta/\Delta_m=1/(2\Delta_m)$ and set $x_0:=z_0/(1+z_0)^{\Delta_m}$ so that $z_0=x_0(1+z_0)^{\Delta_m}$ solving \Cref{eq:fuss-eq}.
As the polynomial expansion of $\mathcal{T}_{\Delta_m}$ has nonnegative coefficients, $\mathcal{T}_{\Delta_m}(x) \le \mathcal{T}_{\Delta_m}(x_0)$ for all $x \le x_0$. 
This implies $\mathcal{T}_{\Delta_m}(x) \le z_0$ for all $x \leq x_0$.
Using $(1+z_0)^{\Delta_m} \le e^{\Delta_m z_0} = e^{1/2}$, we get
\begin{equation}
    x_0 \ge  z_0 e^{-1/2} =  \frac{1}{2\Delta_m}e^{-1/2}.
\end{equation}
Therefore, if $Te^{1/2}\le \frac{1}{2\Delta_m}e^{-1/2}$, then
$\mathcal{T}_{\Delta_m}(Te^{1/2})\le z_0$ and $\Delta_m \mathcal{T}_{\Delta_m}(Te^{1/2}) \le \frac12$, which is \Cref{eq:KP-suffices-new}.

We now show the bound on $T$ suffices to verify the KP condition $T\le \frac{1}{2e \Delta_m}$.
For $m$ large enough, \Cref{eq:T-bound} gives $T\le (1+\varepsilon_1)|\beta|/m$.
Combining with the bound on $\Delta_m$ gives
\begin{equation}
    \Delta_m T  \le  (1+\varepsilon_1)\bigl(k(D-1)+1\bigr)|\beta|.
\end{equation}
So if \Cref{eq:smallness-y-new} holds and $m$ is taken large enough, then $\Delta_m T<1/(2e)$ and the KP condition follows.
\end{proof}

\begin{remark}\label{rem:theta_choice}
In the proof of \Cref{lem:KP_bound} we choose $\theta=1/2$ which arises because it maximizes the equation $\theta e^{-2\theta}$.
More specifically, assume $\theta>0$ is a free parameter. 
Using the equation
$\mathcal{T}_\Delta(x)=x(1+\mathcal{T}_\Delta(x))^\Delta$ yields the condition
\begin{equation}
    \Delta_m T \le \theta e^{-\theta}\left(1+\frac{\theta}{\Delta_m}\right)^{-\Delta_m}.
\end{equation}
We have $\Delta_m\to\infty$ as $m\to\infty$, so the right-hand side
converges to $\theta e^{-2\theta}$ which is the quantity optimized to obtain $\theta=1/2$. 
\end{remark}

We are finally ready to prove \Cref{thm:bounded_degree}.

\begin{proof}[Proof of \Cref{thm:bounded_degree}]
Let $C_{\max}$ be as in \Cref{lem:KP_bound}.
For $\beta$ with $(1+\varepsilon_1)|\beta|<C_{\max}/(k(D-1)+1)$,
by \Cref{lem:KP_bound}, there is $m_0$ such that for all $m\ge m_0$ the polymer partition function $\Xi_m(\beta)$ is holomorphic and zero‑free in the open disk $(1+\varepsilon_1)|\beta|<C_{\max}/(k(D-1)+1)$.

Since in finite dimension, the Trotterization converges uniformly on compact $\beta$‑sets, $U_m(\beta)\to e^{-\beta H}$ in operator norm and thus $\Xi_m(\beta)\to \Xi(\beta):=2^{-n}\Tr e^{-\beta H}$. 
By Hurwitz's theorem, if $\Xi(\beta)$ has a zero in the neighborhood of $\beta$, then so must $\Xi_m(\beta)$ for $m$ large enough. Therefore, $Z_H(\beta)=2^n\Xi(\beta)\neq 0$ for the interior of the disk $(1+\varepsilon_1)|\beta|<C_{\max}/(k(D-1)+1)$.
Because $\varepsilon_1$ was arbitrary, sending $\varepsilon_1\to 0$ finishes the proof that $Z_H(\beta)\neq 0$ for all $|\beta|<C_{\max}/(k(D-1)+1)$. 
\end{proof}

\subsection{Comparison of explicit constants}

We compare several explicit sufficient conditions for the zero-freeness of the partition function $Z_H(\beta)=\Tr(e^{-\beta H})$ in the complex $\beta$-plane, under the normalization $\|h_a\|\le 1$ and the assumption that each site participates in at most $D$ interaction terms and each term acts on at most $k$ sites.
These comparisons are shown in \Cref{tab:beta-radius-comparison}. 
Explicit estimates can be hard to directly compare so for completeness, we describe our methodology below.

We compare directly to bounds in three works \cite{harrow2020classical,mann2024algorithmic,nguyen2024high}.
Comparisons to Refs. \cite{harrow2020classical,mann2024algorithmic} are straightforward as they are in the same setting as ours.
For Ref. \cite{harrow2020classical}, their proof gives a zero-free region whenever $|\beta| < \frac{1}{e(2e-1)kD}$. Technically, in their exposition, they give the simplified bound $|\beta| < (5ekD)^{-1}$, but a careful look at their proof obtains the slightly improved bound. 
In Ref. \cite{mann2024algorithmic}, a sufficient condition for convergence is their stated result $|\beta| < \frac{1}{e^4D\binom{k}{2}}$ which we use directly.

We also compare to results in Ref. \cite{nguyen2024high} which are stated in a more abstract setting.
Obtaining explicit constants can be challenging from their formulation, but we follow the methodology in \cite[Section 5.3]{nguyen2024high}. 
Improved bounds can likely be obtained if one uses a more specialized procedure.
We describe how we obtain bounds below for both lattice settings and general spin settings:
\begin{itemize}
    \item \textbf{Nearest-neighbor lattice specialization ($k=2$)}: 
    Nguyen-Fern\'andez \cite{nguyen2024high} derive for translation-invariant nearest-neighbor
    interactions on $\mathbb{Z}^d$ the sufficient condition (their Eq.~(5.36))
    \begin{equation}
        e^{|\beta| \|\Phi\|}-1 \le
        \frac{\zeta}{\bigl(1+2d \zeta\bigr)^2 (1+\zeta)^{4d-2}},
    \end{equation}
    optimized over $\zeta>0$ (their Eq.~(5.37), and numerically summarized in their Table~5.1).
    Under our normalization $\|\Phi\|\le 1$ and $D=2d$, so
    \begin{equation}
        |\beta| \le \log (1+b_d),
        \qquad
        b_d:=\max_{\zeta>0}\frac{\zeta}{(1+2d\zeta)^2(1+\zeta)^{4d-2}}.
    \end{equation}
    For $D=4$ (i.e. $d=2$) and $D=8$ (i.e. $d=4$) we use the values from Table~5.1 \cite{nguyen2024high}.
    For $D=16$ (i.e. $d=8$) we evaluate $b_d$ by optimizing the above formula at its zero derivative.
    \item \textbf{General spin models}:
    Nguyen-Fern\'andez (see \cite[Cor.~5.6]{nguyen2024high}) provide a sufficient condition of the form
    \begin{equation}
        |\beta| \, \|\Phi\|_\alpha e^{|\beta| \, \|\Phi\|_\alpha} \le \frac{\alpha}{4C_{\alpha/2}},
    \end{equation}
    where $\|\Phi\|_\alpha := \sup_{i \in [n]} \sum_{h_a: i \in \supp h_a} \|h_a\| e^{\alpha |\supp h_a|}$ is their interaction norm and
    $C_{\alpha/2}:=\sup_{i \in [n]} \sum_{h_a: i \in \supp h_a }e^{-(\alpha/2)|\supp h_a|}$.
    To compare to our bounded-degree $k$-local setting, we apply it
    assuming each term has locality exactly $k$, norm $\|h_a\|=1$, and each site is supported on $D$ terms.
    Then $\alpha>0$,
    \begin{equation}
        \|\Phi\|_\alpha = D e^{\alpha k},
        \qquad
        C_{\alpha/2} = D e^{-(\alpha/2)k}.
    \end{equation}
    Substituting this in gives the explicit sufficient condition
    \begin{equation}
        |\beta|
        \le
        \max_{\alpha>0} 
        \frac{1}{D e^{\alpha k}} W\left(\frac{\alpha e^{(\alpha/2)k}}{4D}\right),
    \end{equation}
    where $W(\cdot)$ is the Lambert-$W$ function (defined by $W(x)e^{W(x)}=x$).
    In \Cref{tab:beta-radius-comparison} we evaluate this bound by numerically optimizing over $\alpha$ for each $(k,D)$.
\end{itemize}

\begin{table}[t]
\centering
\small
\setlength{\tabcolsep}{6pt}
\begin{tabular}{l|ccc|ccc}
\hline
 & \multicolumn{3}{c|}{$k=2$} & \multicolumn{3}{c}{$k=3$}\\
Method & $D=4$ & $D=8$ & $D=16$ & $D=4$ & $D=8$ & $D=16$\\
\hline
Ours
& 0.0263 & 0.0123 & 0.00593
& 0.01839 & 0.00836 & 0.00400
\\
HMS \cite{harrow2020classical}
& 0.0104 & 0.0052 & 0.00259
& 0.00691 & 0.00346 & 0.00173
\\
MM \cite{mann2024algorithmic}
& $0.0046$ & $0.0023$ & $0.00114$
& $0.00153$ & $0.00076$ & $0.00038$
\\
NF (lattice) \cite{nguyen2024high}
& 0.0286 & 0.0129 & 0.00633
& -- & -- & --
\\
NF (general) \cite{nguyen2024high}
& $0.0051$ & $0.0013$ & $0.00035$
& $0.00351$ & $0.00091$ & $0.00023$
\\
\hline
\end{tabular}
\caption{
Comparison of explicit sufficient zero-free radii in the complex $\beta$-plane under $\|h_a\|\le 1$.
The Nguyen-Fern\'andez ``lattice'' row assumes translation-invariant nearest-neighbor interactions on $\mathbb{Z}^d$ (so $k=2$ and $D=2d$); the entry at $D=16$ (i.e. $d=8$) is computed from their nearest-neighbor criterion rather than read directly from their table.
The Nguyen-Fern\'andez ``general'' row is a crude (and likely pessimistic) specialization of their Cor.~5.6 to bounded-degree $k$-uniform interactions.
}
\label{tab:beta-radius-comparison}
\end{table}

\subsection{Zero-free versus separability thresholds}
\label{app:sep-vs-zf}

Here, we compare the zero-free regime for bounded-degree Hamiltonians with the separability threshold established in \cite{bakshi2024high}. To avoid ambiguity, we remind the reader of these two thresholds: 
\begin{itemize}
    \item The zero-free threshold $\beta_{\rm zero-free}$ indicates that $Z_H(z)\neq 0$ for every $k$-local, degree-$D$ Hamiltonian in the class whenever $|z|<\beta_{\rm zero-free}$.
    \item The universal separability threshold $\beta_{\rm sep}$ is a sufficient condition guaranteeing that the Gibbs state $\rho_\beta$ is separable for every Hamiltonian in the same class and all $0 \le \beta <\beta_{\rm sep}$.
\end{itemize}

Our \Cref{thm:bounded_degree} gives the former type of guarantee. Ref. \cite{bakshi2024high} gives the latter separability guarantee for Hamiltonians represented as sums of local Pauli terms, although their locality and degree parameters are defined on the interaction (dual) graph rather than in the conventions used here. 
Translated into our notation, their result implies the bound $\beta_{\rm sep} > 1/(100k^2D)$ assuming each term in the Hamiltonian is a Pauli.
As we will see below, for Hamiltonians where each term is a bounded $k$-local operator (not necessarily a Pauli), $\beta_{\rm sep}$ can decay exponentially with $k$.

The proposition below compares these two thresholds showing that the zero-free threshold can surpass the separability threshold especially at large locality $k$. We study the SWAP Hamiltonian only because it admits a simple closed-form threshold, but remark that many other Hamiltonians have the same behavior. 
For example, a $k$-local Hamiltonian $U P U^\dagger$ consisting of any $k$-local Pauli $P$ conjugated by a Haar random $2^k \times 2^k$ unitary $U$ also suffices to prove a similar result to the one below with high probability.

\begin{proposition}[SWAP Hamiltonian separability and zero-free threshold]\label{prop:swap}
Fix an even locality parameter $k=2m$. Consider $k$ qubits split into two registers
$A$ and $B$ of size $m=k/2$ qubits each, so that
$\dim A=\dim B = d_{\mathrm{loc}} = 2^{k/2}$.
Let $F$ denote the $k$ qubit swap operator acting on $A\otimes B$:
 \begin{equation}
F (|x\rangle_A\otimes |y\rangle_B) = |y\rangle_A\otimes |x\rangle_B.
 \end{equation}
Let the Hamiltonian be the single $k$-local term $H=F$ (so the degree is $D=1$ and $\|H\|=1$).
Then, for $\beta > \operatorname{arctanh}(2^{-k/2}) = \Theta(2^{-k/2})$, the Gibbs state $\rho_\beta$ is entangled. Since $H$ consists of a single $k$-local term, it has degree $D=1$, and \Cref{thm:bounded_degree} guarantees that $Z_H(z)\neq 0$ for all $|z|<1/(2e)$. In particular, for $k\ge 8$, the explicit choice $\beta=2^{1-k/2}$ lies inside this zero-free regime and already yields an entangled Gibbs state.
\end{proposition}

\begin{proof}

To prove the entanglement threshold, note that for any product state $|a\rangle_A\otimes |b\rangle_B$,
\begin{equation}
    \langle a,b|F|a,b\rangle = |\langle a|b\rangle|^2 \ge 0.
\end{equation}
By convexity, for any separable mixed state $\sigma$, one has $\Tr(F\sigma)\ge 0$. Thus, if $\Tr(F\rho_\beta)<0$ then $\rho_\beta$ is necessarily entangled.

Since $F^2=\Id$, we have
\begin{equation}\label{eq:exp-swap}
    e^{-\beta F}=\cosh(\beta) \Id-\sinh(\beta) F.
\end{equation}
As a reminder, $\rho_\beta = e^{-\beta F}/\Tr(e^{-\beta F})$.
We now compute $\Tr(F\rho_\beta)$. Using \cref{eq:exp-swap} and $F^2=\Id$:
 \begin{equation}
    \Tr(F\rho_\beta)
    =
    \frac{\cosh\beta \Tr(F)-\sinh\beta \Tr(\Id)}{\cosh\beta \Tr(\Id)-\sinh\beta \Tr(F)}.
\end{equation}
Since $\Tr(\Id)=d_{\mathrm{loc}}^2$ and $\Tr(F)=d_{\mathrm{loc}}$, this simplifies to
\begin{equation}
    \Tr(F\rho_\beta) = \frac{\cosh\beta-d_{\mathrm{loc}}\sinh\beta}{d_{\mathrm{loc}}\cosh\beta-\sinh\beta}.
\end{equation}
For $\beta\ge 0$ the denominator is positive, so $\Tr(F\rho_\beta)<0$ iff
$\cosh\beta<d_{\mathrm{loc}}\sinh\beta$, i.e.\ $\tanh\beta>1/d_{\mathrm{loc}}=2^{-k/2}$.
Equivalently, $\rho_\beta$ is entangled whenever $\beta>\operatorname{arctanh}(2^{-k/2})$.
Since $\operatorname{arctanh}(x)=x+O(x^3)$ as $x\to 0$, we have $\operatorname{arctanh}(2^{-k/2})=\Theta(2^{-k/2})$.

Now consider the explicit choice $\beta=2^{1-k/2}$. We have $\beta>\operatorname{arctanh}(2^{-k/2})$ because $\operatorname{arctanh}(x)<2x$ for all $0<x<1/2$. For $k\ge 8$, one has $2^{1-k/2}<1/(2e)$, so \Cref{thm:bounded_degree} guarantees that the partition function is zero-free at this temperature. 

\end{proof}

\begin{remark}[Why this does not contradict death of entanglement]
    The death of entanglement theorem of \cite{bakshi2024high} does not directly apply to the Hamiltonian $H=F$ as the swap operator $F$ is not a $k$-local Pauli. If one expands $F$ into Pauli strings, the number of Pauli terms required is $2^k-1$ and the degree of the Hamiltonian become $\Theta(2^k)$, so the separability guarantee from \cite{bakshi2024high} becomes $\beta < O\left((k^2 2^k)^{-1} \right)$.
\end{remark}

\section{Stabilizer code Hamiltonians}
\label{app:codes}
Consider an $[[n,k,d]]$ quantum stabilizer code which is generated by $m:=n-k$ Pauli check matrices $C_1, \dots, C_{m}$ where each check matrix is an $n$-qubit tensor product of Pauli matrices $C_i=(-1)^{b_i} P_1^{(i)} \otimes P_2^{(i)} \otimes \cdots \otimes P_n^{(i)} \in \mathcal{P}_n$ with $b_i\in\{0,1\}$ \cite{gottesman1997stabilizer}. 
By construction, $C_1, \dots, C_{m}$ pairwise commute with each other.
We say $C_i$ has weight $w_i$ if $w_i$ of the $n$-many single qubit Pauli matrices are not equal to identity.
We define the Hamiltonian for this code as 
\begin{equation}\label{eq:stabilizer_hamiltonian}
    H=\sum_{i=1}^{m} C_i.
\end{equation}

\subsection{Estimating the partition function}

As $C_1, \dots, C_{m}$ are independent check matrices that pairwise commute and form a code, the eigenspaces of $H$ are isomorphic to those of the transverse field model $\sum_{i=1}^{m}Z_i$. 
In fact, eigenspaces $D_s$ of $H$ can be placed in correspondence with bitstrings $s \in \{-1,+1\}^{m}$ with eigenvalue $\sum_i s_i$ and spanned by states
\begin{equation}
    D_s=\left\{\ket{\phi}: C_i \ket{\phi} = s_i\ket{\phi} \, \forall i \in \{1,\dots,m\} \right\}.
\end{equation}
Thus, the Hamiltonian $H=\sum_{i=1}^{m} C_i$ has eigenvalues $-m+2j$ for $j\in\{0,\dots,m\}$ with corresponding eigenspace of dimension $2^{k}\binom{m}{j}$.
The complex valued partition function evaluates to
\begin{equation}
    Z_H(z)=2^{m+k}\cosh(z)^{m},
\end{equation}
and the Gibbs state $\rho_\beta := e^{\beta H} / \Tr[e^{\beta H}]$ at inverse temperature $\beta>0$ has energy
\begin{equation}
     \Tr[H \rho_\beta] = m\tanh(\beta).
\end{equation}

\begin{theorem}[Barvinok's interpolation for stabilizer Hamiltonians]
    Given an $[[n,k,d]]$ stabilizer code with stabilizer Hamiltonian $H$ of the form in \eqref{eq:stabilizer_hamiltonian} with $\beta > 0$, Barvinok's interpolation solves QPF with runtime $n^{O(\mathrm{poly}(\beta)\log(n/\epsilon))}$ returning $\widetilde{Z}$ such that $(1-\epsilon)Z_H(\beta) \le \widetilde{Z} \le (1+\epsilon)Z_H(\beta)$. 
\end{theorem}
\begin{proof}
    For complex valued $z$, $Z_H(z)=2^{m+k}\cosh(z)^{m}$ which is zero when $\cosh(z)=0$.
    This occurs only on the imaginary axis at values $i\pi/2 + \{i\pi \ell: \ell \in \mathbb{Z}\}$. 
    Using the function $\phi$ in \Cref{lem:wedge-avoiding-mapping} which maps the disk to the wedge, we run Barvinok's interpolation on the composite function $f(z)=Z_H(\beta \phi(z))$ which is zero-free in the disk with radius $R=1+O(\mathrm{poly}(1/\beta))$. 
    Apply \Cref{thm:disk-barvinok} to conclude.
\end{proof}
\begin{remark}
    Given the explicit simple form $Z_H(z)=2^{n}\cosh(z)^{m}$, it is trivial to calculate the partition function $Z_H(\beta)$ in polynomial time in $n$ as opposed to the quasi-polynomial runtime stated above. 
    Similarly, derivatives of $\log Z_H(z)$ can also trivially be calculated in polynomial time in $n$.
    Nonetheless, we report runtimes in the quasi-polynomial fashion as this is what would be achieved when calculating derivatives from moments $\Tr[H^k]$ without relying on the code structure. 
\end{remark}

\subsection{Measuring local observables of thermal state}
The results above guarantee a zero-free region for the partition function $Z_H(\beta)$.
To measure observables of a Gibbs state $\rho_\beta$ with respect to a local operator $A$, we need to apply \Cref{lem:QMV_from_QPF} which requires proving a zero-free region for a perturbed Hamiltonian $H+\delta A$ for $\delta$ which can be inverse polynomially small in $n$.
For any local Pauli $A$, we show below that $Z_{H+\delta A}(z)$ has a zero-free region $\mathcal{D} = \left\{ z \in \mathbb{C} : |\mathrm{Im}(z)| \le \Omega(1) \right\} \cup \left\{ z \in \mathbb{C} : | \Re z| > O(1) \right\}$ which is the set of all points $z=a+ib$ where either $|a|=\Omega(1)$ is larger than some constant threshold or $|b|<\Omega(1)$ and is contained inside a rectangle around the origin.
\begin{proposition}[Zero-free region for perturbed stabilizer Hamiltonian] \label{prop:zero-free-perturbed-stabilizer}
    Given an $[[n,k,d]]$ stabilizer Hamiltonian (see \eqref{eq:stabilizer_hamiltonian}) of the form $H=\sum_{i=1}^{m} C_i$ which is $O(1)$ local (each check matrix $C_i$ has bounded weight), let $A$ be an $n$-qubit Pauli matrix of constant locality (constant weight) and operator norm $\|A\|=1$.
    Define the set of anticommuting checks 
    \begin{equation}
        S(A) := \{ i \in [m] : \{C_i, A\} = 0 \}, \, \text{ with }\, r := |S(A)|.
    \end{equation}
    Then, for any fixed $b_* >0$ such that $ b_* <\pi/2$, there exist constants $\delta_0 > 0$ and $R_0 > 0$ (depending only on $b_*$ and $r$ and independent of $n$) such that for all perturbations $|\delta| \le \delta_0$, the partition function
    \begin{equation}
        Z_{H+\delta A}(z) = \Tr\left(e^{z(H + \delta A)}\right)
    \end{equation}
    is non-zero for all $z$ in the set
    \begin{equation}
        \mathcal{D} = \left\{ z \in \mathbb{C} : |\mathrm{Im}(z)| \le b_* \right\} \cup \left\{ z \in \mathbb{C} : | \Re z| > R_0 \right\}.
    \end{equation}
\end{proposition}
\begin{proof}
    Set $R:=[m]\setminus S(A)$, so $|R|=m-r$.
    The projectors onto eigenbases of $H$ can be decomposed into a set of orthogonal projectors $\{P_s\}_{s\in\{\pm1\}^m}$ indexed by bitstrings $s\in\{\pm1\}^m$ where $C_i P_s = s_i P_s$ and $\mathrm{dim}(P_s)=2^k$.
    Then in this basis, $H=\sum_s \lambda(s) P_s$ with $\lambda(s):=\sum_{i=1}^m s_i$.
    Let us write bitstrings as $s=(s_R, s_S)$ so that $s_R$ is the portion of the bitstring in the set $R$ and $s_S$ lies in $S(A)$. 
    Note that $AP_{(s_R, s_S)}=P_{(s_R, -s_S)}A$.
    Let $\mathcal H$ be a Hilbert space of dimension $2^k$.
    Fix $s_R\in\{\pm1\}^{|R|}$ and consider the vector space which is the direct sum:
    \begin{equation}
        \mathcal{H}(s_R) := \bigoplus_{s_S\in\{\pm1\}^{r}} P_{(s_R,s_S)}\mathcal H
         \cong \mathbb{C}^{2^k}\otimes\mathbb{C}^{2^r}.
    \end{equation}
    On $\mathcal{H}(s_R)$, we have that
    \begin{equation}
        H \cong \left(\sum_{i\in R} s_{R,i}\right) \Id_{2^k}\otimes \Id_{2^r} + \Id_{2^k}\otimes D,
    \end{equation}
    where $D=\mathrm{diag}(\{\lambda_S(s_S)\}_{s_S})$ with $\lambda_S(s_S):=\sum_{i\in S(A)} s_{S,i}$.
    Furthermore, on $\mathcal{H}(s_R)$, we have that
    \begin{equation}
        A \cong \Id_{2^k}\otimes K, \qquad K\ket{s_S} = \ket{-s_S}.
    \end{equation}
    In particular, for each pair $\{s_S,-s_S\}$ the restriction of $D+\delta K$ to the two dimensional subspace $\mathrm{span}\{\ket{s_S},\ket{-s_S}\}$ is the $2\times 2$ matrix $\lambda Z +\delta X$ with $\lambda=\lambda_S(s_S)$.
    It holds that
    \begin{equation}
        \begin{split}
            Z_{H+\delta A}(z) &= \sum_{s_R\in\{\pm1\}^{|R|}} \Tr_{\mathcal{H}(s_R)} \left( \exp\left[ z\left(\sum_{i\in R} s_{R,i}\right) \Id_{2^k}\otimes \Id_{2^r} + z\Id_{2^k}\otimes (D+\delta K) \right] \right) \\
            &= \sum_{s_R\in\{\pm1\}^{|R|}} e^{z\sum_{i\in R} s_{R,i}} \Tr\left( \exp\left[ z\Id_{2^k}\otimes (D+\delta K) \right] \right) \\
            &= 2^{|R|}(\cosh z)^{|R|} 2^k \Tr\left( \exp\left[ z (D+\delta K) \right] \right) .
        \end{split}
    \end{equation}
    $\cosh(z)$ has zeros only on the imaginary axis for $ib$ such that $|b|\geq \pi/2$. 
    We thus look at the zeros in the term $S_r(z, \delta):=\Tr\left( \exp\left[ z (D+\delta K) \right] \right)$.
    $D+\delta K$ is a block diagonal matrix of $2^{r-1}$ $2 \times 2$ blocks spanned by $\mathrm{span}\{\ket{s_S},\ket{-s_S}\}$.
    Therefore:
    \begin{equation}
        S_r(z, \delta) := \Tr\left( \exp\left[ z (D+\delta K) \right] \right) =  \sum_{s_S\in\{\pm\}^{r}} \cosh\left( z\sqrt{\lambda_S(s_S)^2+\delta^2 } \right).
    \end{equation}
    To identify the zeros of $S_r(z, \delta)$ we will split the analysis into two regions: (1) fixing an arbitrary $a_0 > 0$, we consider the rectangle $\mathcal{R} = \{ z : | \Re z| \le a_0, |\mathrm{Im}(z)| \le b_* \}$; (2) we consider the region where $\{ z: | \Re z| \ge a_0\}$.

    For the first region, we will use Rouché's Theorem.
    At $\delta=0$, the function $S_r$ simplifies to $S_r(z, 0) = (2\cosh z)^r$. On $\mathcal{R}$, this is bounded away from zero since $\min_{z \in \mathcal{R}} |S_r(z, 0)| \ge (2 \cos b_*)^r > 0$.
    To bound the perturbation, define the function $f(u) = \cosh(z\sqrt{\lambda^2 + u})$ for $u \in \mathbb{C}$. 
    The derivative with respect to $u$ is $f'(u) = \frac{z}{2} \frac{\sinh(z\sqrt{\lambda^2+u})}{\sqrt{\lambda^2+u}}$. 
    This function is holomorphic everywhere (the singularity at $\sqrt{\lambda^2+u}=0$ is removable). 
    Since $\mathcal{R}$ is compact, there exists a constant $K$ such that for sufficiently small $|u| \leq \delta_0^2$, $|f'(u)| \le K$.
    Summing over all $2^r$ strings $s_S$ and choosing $0 < \delta < \delta_0$, this gives
    \begin{equation}
        |S_r(z, \delta) - S_r(z, 0)| \le \sum_{s} \left| \cosh(z\sqrt{\lambda_S(s)^2+\delta^2}) - \cosh(z|\lambda_S(s)|) \right| \le 2^r K \delta^2.
    \end{equation}
    Choose $\delta_0$ such that $2^r K \delta_0^2 < (2 \cos b_*)^r$. 
    By Rouché's theorem, $S_r(z, \delta)$ has the same number of zeros as $S_r(z, 0)$ inside $\mathcal{R}$. 
    Since $(2\cosh z)^r$ has no zeros in $\mathcal{R}$, $S_r(z, \delta)$ is non-vanishing in $\mathcal{R}$.

    We now consider the second region of large $|\Re z|>a_0$. We will set $a_0$ large enough but still independent of $n$ so that $S_r(z, \delta)$ is nonzero in this region. 
    For ease of notation, define $\omega_s(\delta) := \sqrt{\lambda_S(s)^2 + \delta^2}$. 
    The function $S_r(z, \delta)$ is then equal to
    \begin{equation}
        S_r(z, \delta) = \sum_{s \in \{\pm 1\}^r} \frac{1}{2} \left( e^{z \omega_s(\delta)} + e^{-z \omega_s(\delta)} \right).
    \end{equation}
    Let $\omega_{\max} := \max_s \omega_s(\delta) = \sqrt{r^2 + \delta^2}$. 
    We factor out the dominant term $e^{\Re(z) \omega_{\max}}$, obtaining
    \begin{equation}
        S_r(z, \delta) = \frac12 e^{z \omega_{\max}} \lr{N_{\max} + \sum_{\omega_s < \omega_{\max}} e^{-z(\omega_{\max} - \omega_s)}} + \frac{1}{2}\sum_s e^{-z\omega_s},
    \end{equation}
    where $N_{\max} \ge 1$ is the degeneracy of the maximal term.
    Let $\Delta$ denote the difference between $\omega_\mathrm{max}$ and the second-largest $\omega_s$; this is a positive constant depending only on $r$ and $\delta_0$. For now, assume $\Re(z) > 0$. Since $|e^{-z(\omega_{\max}-\omega_s)}| \leq e^{-\Re(z)\Delta}$ and $|e^{-z\omega_s}| \leq e^{-\Re(z)|\delta|}$, we have
    \begin{align}
        |S_r(z, \delta)| &\geq \frac12 e^{\Re(z) \omega_{\max}} \lr{ N_{\max} - \sum_{\omega_s < \omega_\mathrm{max}} e^{-\Re(z)(\omega_{\max} - \omega_j)}} - \frac{1}{2}\sum_s e^{-\Re(z)\omega_s} \\
        &\geq \frac12 e^{\Re(z) \omega_{\max}} \lr{ N_{\max} - 2^r e^{-\Re(z)\Delta}} - 2^r e^{-\Re(z)|\delta|}.
    \end{align}
    Since $\Delta$ and $|\delta|$ are constant, there exists $R_0 > 0$ such that for all $\Re(z) > R_0$,
    \begin{equation}
        |S_r(z,\delta)| \geq \frac{N_\mathrm{max}}{8}e^{\Re(z)\omega_\mathrm{max}} > 0,
    \end{equation}
    i.e., that $S_r(z,\delta)$ is nonzero.
    By symmetry $z \to -z$, the same holds for $\Re(z) < -R_0$. (Note that if all $\omega_s$ are equal, then one can set $R_0=0$.)
    Set $a_0>R_0$ to conclude.
\end{proof}

\begin{corollary}[Local observables of stabilizer Hamiltonian]
    Let $\rho_\beta$ be the Gibbs state of a stabilizer Hamiltonian $H$ at inverse temperature $\beta>0$. 
    Let $A$ be a local observable with bounded norm $\|A\|=O(1)$.
    \Cref{prop:zero-free-perturbed-stabilizer} guarantees that the zeros of the perturbed model $H+\delta A$ are contained in a wedge which can be avoided using the mapping in \Cref{lem:wedge-avoiding-mapping}.
    Therefore, following \Cref{lem:QMV_from_QPF}, $\Tr[\rho_\beta A]$ can be evaluated to additive accuracy $\epsilon$ with runtime $n^{O(\mathrm{poly}(\log(1/\epsilon),\log(n), \beta))}$.
\end{corollary}

\subsection{Entangled phase}

Since Barvinok's interpolation can estimate the partition function at any inverse temperature $\beta$, an important question is at what threshold $\beta_c$ does the Gibbs state of a stabilizer code Hamiltonian become entangled (i.e., not separable). 
A separable state is a mixed state $\rho$ on $n$ qubits which can be written as a convex combination $\rho = \sum_i^{N_{sep}} p_i \sigma_i$ of tensor products of single qubit states.
Namely, $\rho = \sum_{i=1}^{N_{sep}} p_i \sigma_i$ is a decomposition of $\rho$ into finitely many product states $\sigma_i=\sigma_i^{(1)} \otimes \cdots \sigma_i^{(n)}$ and $p \in \mathbb{R}_+^{N_{sep}}$ is a distribution ($\sum_i p_i=1$) that denotes the weights in the mixture.

Refs. \cite{anshu2020circuit,anshu2023nlts,bravyi2025much} quantify the entanglement and circuit depth of low energy states in terms of the code parameters $n, k, d$. 
In our Gibbs state setting, we will show that there is a temperature threshold for the separable states for any $n$ qubit stabilizer code below.
In particular, our aim is to obtain simple bounds for the separability threshold $\beta_{\rm sep}$. 
We remark that potentially stronger and more explicit results (e.g., in terms of growing circuit depth or Von Neumann entropies) can be obtained by leveraging the findings in \cite{anshu2020circuit,anshu2023nlts,bravyi2025much}.

We first prove a general result bounding product state performance via weights of the code and then instantiate it for different models.
\begin{theorem}[Separability threshold of stabilizer Hamiltonians]\label{thm:stabilizer_separbility}
    Denote the set of $n$ qubit separable states as $\mathcal{S}^{\rm sep}_n$. Given an $[[n,k,d]]$ stabilizer Hamiltonian (see \eqref{eq:stabilizer_hamiltonian}) of the form $H=\sum_{i=1}^{m} C_i$ where check matrix $C_i$ has weight $w_i$, define $n$ vectors $v_1, \dots, v_n \in \mathbb{R}_+^3$ whose entries are indexed by the three single-qubit Paulis $X,Y,Z$ and are equal to
    \begin{equation}
        [v_a]_P = \sum_{i:[C_i]_a=P} \frac{1}{w_i}, \quad a \in [n], P \in \{X,Y,Z\}.
    \end{equation}
    Then the maximum energy achievable by separable states is bounded by
    \begin{equation}
        \sup_{\rho \in \mathcal{S}^{\rm sep}_n} \Tr[H\rho] \le \sum_{a=1}^n \|v_a\|.
    \end{equation}
\end{theorem}
\begin{proof}
    Any single qubit state can be written in its Bloch sphere representation as $\frac{1}{2}I + \frac{1}{2}(r_XX+r_YY+r_ZZ)$ where $r_X^2+r_Y^2+r_Z^2\le1$. 
    For any product state $\rho=\sigma_i^{(1)} \otimes \cdots \sigma_i^{(n)}$, denote the Bloch sphere representation of qubit $i$ with the vector $r^{(i)}=(r_X^{(i)},r_Y^{(i)},r_Z^{(i)})$
    Upper bounding $\Tr[H\rho]$ over all product states gives a bound on $\sup_{\rho \in \mathcal{S}^{\rm sep}_n} \Tr[H\rho]$ by convexity since any separable state is a convex mixture of product states.

    Write each check matrix $C_i=(-1)^{b_i} P_1^{(i)} \otimes P_2^{(i)} \otimes \cdots \otimes P_n^{(i)}$ explicitly as a tensor product of its $n$ single qubit Pauli matrices.
    Note that $\Tr[C_i \rho] = \prod_{j:[n], P_j^{(i)} \in \{X,Y,Z\}}\, r_{P_j^{(i)}}^{(j)}$. 
    Summing over all $C_i$, the maximum energy of a product state can be determined by solving the following optimization:
    \begin{equation}\label{eq:optimization_stabilizer_product}
    \begin{split}
        \max_{r^{(1)}, \dots r^{(n)}}& \; \sum_{i=1}^m\, \prod_{j:[n], P_j^{(i)} \in \{X,Y,Z\}}\, \left|r_{P_j^{(i)}}^{(j)}\right|,   \\
        \text{s.t. }& \; (r_X^{(j)})^2+(r_Y^{(j)})^2+(r_Z^{(j)})^2\le1 \quad \forall j \in [n].
    \end{split}
    \end{equation}
    Use the AM-GM inequality to obtain
    \begin{equation}
        \prod_{j:[n], P_j^{(i)} \in \{X,Y,Z\}}\, \left|r_{P_j^{(i)}}^{(j)}\right| \leq \left( \frac{1}{w_i}\sum_{j:[n], P_j^{(i)} \in \{X,Y,Z\}}\, \left|r_{P_j^{(i)}}^{(j)}\right| \right)^{w_i} \leq \frac{1}{w_i}\sum_{j:[n], P_j^{(i)} \in \{X,Y,Z\}}\, \left|r_{P_j^{(i)}}^{(j)}\right|.
    \end{equation}
    In the last inequality above we use the fact that the argument inside the parentheses is bounded in magnitude by $1$.
    Apply the equation above to upper bound \eqref{eq:optimization_stabilizer_product}, and with some reordering, we get
    \begin{equation}
    \begin{split}
        \max_{r^{(1)}, \dots r^{(n)}}& \; \sum_{a=1}^n\,\sum_{P \in \{X,Y,Z\}}\, \sum_{i:[C_i]_a=P} \frac{\left| r_{P}^{(a)} \right|}{w_i},   \\
        \text{s.t. }& \; (r_X^{(j)})^2+(r_Y^{(j)})^2+(r_Z^{(j)})^2\le1 \quad \forall j \in [n].
    \end{split}
    \end{equation}
    In the optimization above, we can optimize each qubit separately.
    Namely, for qubit $a$, we optimize $\sum_{P \in \{X,Y,Z\}} \left|r_{P}^{(a)}\right| (\sum_{i:[C_i]_a=P} \frac{1}{w_i} )$ subject to $(r_X^{(a)})^2+(r_Y^{(a)})^2+(r_Z^{(a)})^2\le1$. 
    This is maximized by setting the vector $r^{(a)}=v_a/\|v_a\|$ where as a reminder $[v_a]_P = \sum_{i:[C_i]_a=P} \frac{1}{w_i}$.
    Summing over all qubits $a \in [n]$, this provides an upper bound of $\sum_{a=1}^n \|v_a\|$ to the maximum energy over product states completing the proof.
\end{proof}

\begin{corollary}[Toric code]
    In the $2$-dimensional toric code, stabilizers are all weight $4$ and each qubit has two Pauli $X$ operators and two Pauli $Z$ operators applied to it.
    Applying \Cref{thm:stabilizer_separbility}, $v_a=(1/2,0,1/2)$ for each qubit $a \in [n]$, separable states achieve at most an energy $n/\sqrt{2}$.
    Since the Gibbs state $\rho_\beta$ has energy $m\tanh(\beta)$, the Gibbs state is entangled at all $\beta > \tanh^{-1}(1/\sqrt{2}) + o(1)\approx 0.8813$. 
\end{corollary}
\begin{corollary}[Any code with $d\ge 2$]
    Given an $[[n,k,d]]$ stabilizer code with check matrices $C_1, \dots, C_{m}$, let $\langle C_1, \dots, C_m \rangle$ be the stabilizer group of the code generated by matrix products from the check matrices. 
    Assume generators have maximum weight $w_{\rm max}$ so the Hamiltonian $H=\sum_{i=1}^m C_i$ is $w_{\rm max}$-local.
    If the weight of any nontrivial Pauli matrix in the stabilizer group is at least $2$ and distance $d \geq 2$, then it holds that 
    \begin{equation}
        \sup_{\rho \in \mathcal{S}^{\rm sep}_n} \Tr[H\rho] \le m w^*, \qquad w^*=1 - \frac{2-\sqrt{2}}{w_{\rm max}}\left(1+\frac{k}{m}\right).
    \end{equation}
    Therefore, the Gibbs state $\rho_\beta$ is entangled at all $\beta > \tanh^{-1}(w^*)$.
\end{corollary}
\begin{proof}
    Since $d \ge 2$ and no single qubit Pauli matrices are in the stabilizer group, there exist two generators $C_a,C_b$ whose single-qubit Paulis on site $j$ anticommute (e.g. $X$ and $Z$ on $j$) for all sites $j \in [n]$. In fact, we can show this by looking at two cases:
    (1) If every generator acts as the identity on qubit $j$, then every single-qubit $P_j$ commutes with all generators and is a weight-$1$ logical operator, contradicting $d \geq 2$. 
    (2) Suppose every generator acts as either the identity or some particular Pauli $Q_j$; then $Q_j$ commutes with every generator and is again a weight-$1$ logical operator. 
    We conclude that at least two distinct non-identity Paulis (which anticommute) must exist.
    
    Thus, for vectors $v_a$ as defined in \Cref{thm:stabilizer_separbility}, each $v_a$ must have at least two entries greater than $1/w_{\rm max}$. 
    For any nonnegative set of three numbers $(x,y,z)$ with $y$ the second largest coordinate, the inequality
    \begin{equation}
        \sqrt{x^2+y^2+z^2} \le \sqrt{(x+z)^2+y^2} \le (x+z)+(\sqrt{2}-1)y
    \end{equation}
    holds. 
    The second step is the bound $\sqrt{u^2+v^2}\le u+(\sqrt{2}-1)v$ for $u\ge v\ge 0$. 
    Applying this to $v_a$ and summing over $a$ gives
    \begin{equation}
        \sum_{a=1}^n \|v_a\|
        \le
        \sum_{a=1}^n \sum_{P}[v_a]_P - (2-\sqrt{2})\sum_{a=1}^n y_a
        \le
        m-\frac{2-\sqrt{2}}{w_{\max}} n,
    \end{equation}
    because $\sum_{a,P}[v_a]_P = m$ and $y_a\ge 1/w_{\max}$ for all $a$. Applying \Cref{thm:stabilizer_separbility} then gives the stated bound on $\sup_{\rho\in\mathcal{S}^{\rm sep}_n}\Tr(H\rho)$. 
    The entanglement threshold follows from $\Tr(H\rho_\beta)=m\tanh\beta$.
\end{proof}

\section{SYK model}
\label{app:syk}

\begin{definition}[SYK model]
    Given $N$ Majorana fermions, i.e., operators $\psi_i$ satisfying 
    \begin{align}
        \{\psi_i,\psi_j\} = \delta_{ij}, \qquad \psi_i = \psi_i^\dagger,
    \end{align}
    the $q$-local SYK Hamiltonian ensemble is given by
    \begin{align}
        H = i^{q/2} \sum_{1 \leq i_1 < \cdots < i_q \leq N} J_{i_1\cdots i_q}\psi_{i_1} \cdots \psi_{i_q}
    \end{align}
    for i.i.d. coefficients
    \begin{align}
        J_{i_1\cdots i_q} \sim \cN\lr{0, \frac{(q-1)! J^2}{N^{q-1}}}.
    \end{align}
    We will also commonly use $\cJ^2 = qJ^2/2^{q-1}$.
\end{definition}

Note that this is normalized such that with probability $1-o(1)$, an SYK Hamiltonian satisfies $\norm{H} = \Theta(N)$. In the standard SYK analysis \cite{maldacena2016remarks,khramtsov2021spectral}, the partition function is computed via the path integral
\begin{align}\label{eq:sykeuclid}
    Z = \int D\psi_i \exp[-\int_0^\beta d\tau\lr{\frac{1}{2}\sum_j \psi_j \partial_\tau \psi_j + i^{q/2} \sum_{1 \leq i_1 < \cdots < i_q \leq N} J_{i_1\cdots i_q}\psi_{i_1} \cdots \psi_{i_q}}].
\end{align}
In particular, the quenched free energy $\EE{\log Z}$ can be computed via the replica trick at constant $\beta$ by computing $\EE{Z^n}$ for integer $n$ via the saddle point method in the $N\to\infty$ limit. The saddle point equations can be solved in the large-$q$ limit under the replica symmetric ansatz, which is known to hold due to replicon stability and strong numerical agreement with exact diagonalization at finite $N$. Rigorously, the annealed and quenched free energies are known to be asymptotically equal at constant $\beta$, which is consistent with the assumption of replica symmetry.

In this work, we use a similar path integral computation to evaluate complex $\beta$ and control the complex zeros of the partition function. However, rather than computing the Euclidean path integral of \eqref{eq:sykeuclid}, we require control over complex $\beta$. In particular, we will use Jensen's formula in a similar manner to recent work showing the application of Barvinok's formula for the Sherrington-Kirkpatrick model~\cite{bencs2025zeros}. To show zero-freeness, it will suffice to compute the annealed expectation
\begin{align}
    \E \left|Z(\beta)\right|^2,
\end{align}
which we do by the saddle point method. We obtain the following result from the path integral computation, which states that the SYK model is zero-free for all $|\beta|=O(1)$ satisfying
\begin{align}
    \Re \beta \neq 0 \quad \text{or} \quad |\Im \beta| < \frac{1.3}{\cJ}.
\end{align}
\begin{result}[Zero-free region of the SYK model] \label{res:SYK_zero_free_region}
    Define the zero-free region 
    \begin{equation}
        S_R:=\mathbb{D}_R \setminus  \left\{a+ib\, :\, a=o(1),|b|\geq \frac{1.3}{\cJ}+o(1)\right\} \subset \mathbb{C}.
    \end{equation}
    For any $R=O(1)$, with probability $1-o(1)$,  $Z(\beta) \neq 0$ for all $\beta \in S_R$.
\end{result}

The zero-free region above guarantees a quasi-polynomial runtime for Barvinok's interpolation when evaluating the partition function of the SYK model at any $\beta=O(1)$. 
Furthermore, due to equivalences between estimating the quantum mean value (QMV, \Cref{problem:QMV}) and the quantum partition function (QPF, \Cref{problem:QPF}) detailed earlier, local observables and the energy of the Hamiltonian $H$ with respect to the Gibbs state $\rho_\beta$ can also be estimated efficiently in quasi-polynomial time up to any additive error $\epsilon=O(1/\mathrm{poly}(N))$. 

\begin{theorem}[Barvinok's interpolation for SYK] \label{thm:Barvinok_SYK}
    Let $H$ denote the Hamiltonian over $N$ Majoranas drawn from the SYK ensemble and fix $\beta > 0$. 
    Assuming \Cref{res:SYK_zero_free_region} rigorously holds, Barvinok's interpolation solves QPF with runtime $N^{O(\log(N/\epsilon)\mathrm{poly}(\beta))}$ returning $\widetilde{Z}$ such that $(1-\epsilon)Z_H(\beta) \le \widetilde{Z} \le (1+\epsilon)Z_H(\beta)$ with probability $1-o(1)$ over the disorder of the Hamiltonian. 
\end{theorem}
\begin{proof}
    \Cref{res:SYK_zero_free_region} guarantees that the zeros of $Z_H(\beta)$ are only on the imaginary axis for $n$ large enough at values  $ib$ with $|b|>\Omega(1)$. 
    Using the function $\phi$ in \Cref{lem:wedge-avoiding-mapping} which maps the disk to the wedge, we run Barvinok's interpolation on the composite function $f(z)=Z_H(\beta \phi(z))$ which is zero-free in the disk with radius $R=1+O(\mathrm{poly}(1/\beta))$. 
    Apply \Cref{thm:disk-barvinok} to conclude.
\end{proof}
\begin{corollary}[Local observables of SYK]
    Given a local observable $A$ with bounded norm $\|A\|=O(1)$, its expectation $\Tr[\rho_\beta A]$ on the Gibbs state $\rho_\beta$ can be evaluated to additive accuracy $\epsilon$ also with runtime $N^{O(\mathrm{poly}(\log(1/\epsilon),\log(N), \beta))}$.
    In fact, by equivalences between estimating the quantum mean value (QMV, \Cref{problem:QMV}) and the quantum partition function (QPF, \Cref{problem:QPF}) detailed earlier, $\Tr[\rho_\beta A]$ can be estimated by running QPF upon a perturbed Hamiltonian $H'=H+\lambda A$ for $\lambda = o(1)$. 
    This perturbation does not change the zero-free region in the replica calculations to leading order (see \Cref{sec:SYK_local_observables}).
\end{corollary}

\begin{remark}[Resolving instance-to-instance fluctuations]
    We note that fluctuations in the energy of the Gibbs state and its observables are also $\Omega(1/\mathrm{poly}(N))$ so the error guarantees of Barvinok's interpolation can estimate observables at values surpassing concentration thresholds. For example, consider fluctuations in the variance of the energy which are influenced by fluctuations in the Euclidean norm of the coefficients $\|J\|=\sqrt{\sum_{1 \leq i_1 < \cdots < i_q \leq N} J_{i_1\cdots i_q}^2}$. Namely, upon considering a normalized ensemble drawn as $H'=\frac{H}{\|J\|}$, then $H$ is drawn as $H=c H'$ for $c$ drawn from the square root of a Chi-squared distribution. The variance in $c$ is $\Omega(1/\mathrm{poly}(N))$ which also implies corresponding $\Omega(1/\mathrm{poly}(N))$ bounds in the variance of the energy.
\end{remark}

\subsection{Complex zeros from Jensen's formula}
\label{app:sykjensen}

Our starting point for showing zero-freeness is Jensen's formula, which has previously been used to show zero-freeness for permanents and the Sherrington-Kirkpatrick model~\cite{eldar2018approximating,bencs2025zeros}.
\begin{theorem}[Jensen's formula] \label{thm:jensens-formula}
    Let $\Omega \subset \C$ be an open set that contains the disk $\overline{\D(0,R)}$ for some $R > 0$. Let $Z:\Omega\to\C$ be an analytic function such that $Z(0) \neq 0$. Then
    \begin{align}
        \sum_{\beta\in \D(0,R):Z(\beta)=0} \log \frac{R}{|\beta|} = \frac{1}{2\pi} \int_{-\pi}^\pi d\theta \log \left|\frac{Z(Re^{i\theta})}{Z(0)}\right|.
    \end{align}
\end{theorem}
We can use Jensen's formula to count the number of zeros in a disk of smaller radius $r < R$ via the inequality (for a particular instance of disorder)
\begin{align}
    \#\{\beta \in \D(0,r) : Z(\beta) = 0\} \leq \lr{\log \frac{R}{r}}^{-1} \sum_{\beta \in \D(0,r):Z(\beta)=0} \log \frac{R}{|\beta|} = \lr{\log \frac{R}{r}}^{-1} \frac{1}{2\pi} \int_{-\pi}^\pi d\theta \log \left|\frac{Z(Re^{i\theta})}{Z(0)}\right|.
\end{align}
We wish to move from the instance-specific quantity to a statement that holds with high probability over the ensemble. The key insight of~\cite{eldar2018approximating} is that instead of computing the quenched quantity $\E \log \left|\frac{Z(Re^{i\theta})}{Z(0)}\right|$, it suffices by Jensen's inequality to compute an annealed expectation since
\begin{align}
    \E \log \left|\frac{Z(Re^{i\theta})}{Z(0)}\right| \leq \frac{1}{2} \log \E \left|\frac{Z(Re^{i\theta})}{Z(0)}\right|^2.
\end{align}
The above analysis on a disk correctly describes the glass transition of the SK model~\cite{bencs2025zeros}: it is replica symmetric for all real $0 < \beta < 1$ and has few complex zeros in the $|\beta| < 1$ disk. However, although the SYK model is replica symmetric for all (arbitrarily large) constant real $\beta$, it does \emph{not} have a zero-free disk of arbitrary constant radius. Indeed, the spectral form factor is known to vanish at constant value, corresponding to zeros in constant imaginary $\beta$. We thus require a more careful analysis for the SYK model.

We will ultimately show zero-freeness in a wedge in the complex plane which avoids the zeros that fall close to the imaginary axis.
This requires the use of Jensen's formula adapted to use the disk to wedge mapping of \Cref{lem:wedge-avoiding-mapping}, which we make precise here.

\begin{lemma}[Jensen's formula for wedge]\label{lem:jensen-via-wedge}
    Let $\phi$ be as in \Cref{lem:wedge-avoiding-mapping} with radius $R>1$, so $\phi:\mathbb D_R\to\mathbb C$ is holomorphic, $\phi(0)=0$, $\phi(1)=1$, and $\phi(\mathbb D_R)$ avoids the forbidden wedges. 
    Let $f$ be holomorphic in $\phi(\overline{\mathbb D_r})$ for some $1<r<R$ and
    let $N_{\phi(\mathbb D_{r_1})}(f)$ be equal to the number of zeros of $f$ (with multiplicity) in the set $\phi(\mathbb D_{r_1})$ for any $0<r_1<r$. Assume that $f(0) \neq 0$.
    Then,
    \begin{equation}\label{eq:jensen-phi}
        \mathbb{E} [N_{\phi(\mathbb D_{r_1})}(f)] \le
        \frac{1}{2\log(r/r_1)}\cdot \frac{1}{2\pi}\int_0^{2\pi}
        \log \mathbb{E}\left|\frac{f\left( \phi(re^{i\theta} )\right)}{f(0)}\right|^2 \,d\theta.
    \end{equation}
    Furthermore, if $\log \mathbb{E}\left[ \left|\frac{f(\beta)}{f(0)}\right|^2 \right]$ is harmonic on $\phi(\mathbb D_r)$, then $\mathbb{E} [N_{\phi(\mathbb D_{r_1})}(f)] = 0$.
\end{lemma}
\begin{proof}
Define $F(z):=f(\phi(z))$, which is analytic on $\mathbb D_r$.
Let $\{\beta_j\}$ be the multi-set of zeros of $f$ in $\phi(\mathbb D_{r_1})$, counted with multiplicity. 
For each $\beta_j$ choose $z_j\in\mathbb D_{r_1}$ with $\phi(z_j)=\beta_j$. Locally write $f(\beta)=(\beta-\beta_j)^{m_j}g(\beta)$ with $g(\beta_j)\neq 0$. Then
$F(z)=f(\phi(z))=(\phi(z)-\beta_j)^{m_j}g(\phi(z))$ shows that $F$ has a zero at $z_j$ of multiplicity at least $m_j$. 
Therefore
\begin{equation}\label{eq:count-ineq}
    N_{\phi(\mathbb D_{r_1})}(f) \le N_{\mathbb D_{r_1}}(F),
\end{equation}
where $N_{\mathbb D_{r_1}}(F)$ counts zeros of $F$ with multiplicity.
Jensen's formula (\Cref{thm:jensens-formula}) gives
\begin{equation}
    \sum_{|z_k|<r}\log\frac{r}{|z_k|} = \frac{1}{2\pi} \int_0^{2\pi} \log\left| \frac{F(re^{i\theta})}{F(0)} \right|\,d\theta.
\end{equation}
Since $\log\frac{r}{|z_k|}\ge \log\frac{r}{r_1}$ for all $|z_k|<r_1$, we get
\begin{equation}
    N_{\mathbb D_{r_1}}(F) \log\frac{r}{r_1} \le \frac{1}{2\pi} \int_0^{2\pi} \log\left| \frac{F(re^{i\theta})}{F(0)} \right|\,d\theta.
\end{equation}
Taking expectations, using $\log(|x|) = \frac12 \log(|x|^2)$, and applying Jensen's inequality to the logarithm yields
\begin{equation}
    \mathbb{E}[N_{\mathbb D_{r_1}}(F)] \le \frac{1}{2\log(r/r_1)} \frac{1}{2\pi} \int_0^{2\pi} \log \mathbb{E}\left[ \left|\frac{F(re^{i\theta})}{F(0)}\right|^2 \right]\,d\theta.
\end{equation}
Substitute $F(z)=f(\phi(z))$ and use the bound in \Cref{eq:count-ineq} to obtain \Cref{eq:jensen-phi}.

For the last claim, let $u(\beta):=\log \mathbb{E}\left[ \left|\frac{f(\beta)}{f(0)}\right|^2 \right]$ and assume $u$ is harmonic on $\phi(\mathbb D_r)$. 
Since $\phi$ is holomorphic on $\mathbb D_r$, $u\circ\phi$ is harmonic on $\mathbb D_r$. 
Applying the mean-value property of harmonic functions then gives $\frac{1}{2\pi}\int_0^{2\pi} u ( \phi(re^{i\theta}) )\,d\theta = (u \circ \phi)(0)= 0$.
\end{proof}

In the analysis below, we will compute
\begin{align}
    \log \E \left|\frac{Z(\beta)}{Z(0)}\right|^2
\end{align}
via the saddle point method and find that the leading-order action, of size $\Theta(N)$, is harmonic on a wedge region. While subleading $O(1)$ terms that are not harmonic may exist and produce $O(1)$ zeros, we argue that no such non-harmonic remainder exists. This is supported numerically (\Cref{app:num}), and thus we report our results in the main text assuming that the leading-action saddle point fully describes the physics. However, we note that if non-harmonic subleading terms exist, then the resulting $O(1)$ zeros in the wedge region can be avoided by trying a sufficiently large number of root-avoiding tubes as in~\cite{eldar2018approximating}.

\subsection{Complex zeros of SYK}

\subsubsection{Solving the saddle point equations}
We compute $\E |Z(\beta)|^2$ here via the path integral using the standard notation
\begin{align}
    \cJ^2 = \frac{qJ^2}{2^{q-1}}.
\end{align}
In the Euclidean path integral computation of $\EE{Z^n}$~\cite{maldacena2016remarks,kitaev2018soft}, one begins with the path integral
\begin{align}
    Z = \int D\psi_i \exp[-\int_0^\beta d\tau\lr{\frac{1}{2}\sum_j \psi_j \partial_\tau \psi_j + i^{q/2} \sum_{1 \leq i_1 < \cdots < i_q \leq N} J_{i_1\cdots i_q}\psi_{i_1} \cdots \psi_{i_q}}]
\end{align}
with antiperiodic boundary conditions $\psi_i(\beta) = -\psi_i(0)$. After making $n$ copies and evaluating the expectation over disorder via the Gaussian MGF, one inserts the identity
\begin{align}
    1 = \int DG \, D\Sigma \exp\left[-\frac{N}{2} \sum_{ab}\int_0^\beta d\tau_1\,d\tau_2\, \Sigma_{ab}(\tau_1,\tau_2) \lr{G_{ab}(\tau_1,\tau_2) - \frac{1}{N}\sum_i \psi_i^a(\tau_1) \psi_i^b(\tau_2)}\right].
\end{align}
This gives
\begin{align}
    \EE{Z^n} &= \int DG\, D\Sigma\, D\psi_i^a \exp\Bigg[-\frac{1}{2}\sum_a \int_0^\beta d\tau\sum_j \psi_j^a \partial_\tau \psi_j^a + \frac{J^2 N}{2q} \sum_{ab} \int_0^\beta d\tau_1 \, d\tau_2\, G_{ab}(\tau_1,\tau_2)^q \nonumber\\
    &\qquad - \frac{N}{2} \sum_{ab}\int_0^\beta d\tau_1\,d\tau_2\, \Sigma_{ab}(\tau_1,\tau_2) \lr{G_{ab}(\tau_1,\tau_2) - \frac{1}{N}\sum_i \psi_i^a(\tau_1) \psi_i^b(\tau_2)}\Bigg].
\end{align}
To rewrite the expectation in the form
\begin{align}
    \EE{Z^n} = \int DG \, D\Sigma\, \exp[-N S_\mathrm{Euclidean}[G,\Sigma]],
\end{align}
we do a Gaussian integral over $D\psi_i^a$. Assuming the replica diagonal solution $G_{a\neq b} = \Sigma_{a \neq b} = 0$, this gives action
\begin{align}
    S_\mathrm{Euclidean}[G,\Sigma] &= - \log \pf (\partial_\tau - \hat\Sigma_{aa}) + \frac{1}{2}\int_0^\beta d\tau_1\,d\tau_2\, \lr{\Sigma_{aa}(\tau_1,\tau_2) G_{aa}(\tau_1,\tau_2) - \frac{J^2}{q} G_{aa}(\tau_1,\tau_2)^q},
\end{align}
where $\pf$ denotes the Pfaffian in both time and replica indices, and $\hat \Sigma$ denotes the integral operator defined by the corresponding bilocal field as its kernel.

For complex $\beta$, a similar computation holds on $n=2$ copies denoted $L,R$. Since we are computing the second moment, our computation will be similar to that of the spectral form factor~\cite{khramtsov2021spectral}, which represents the special case of purely imaginary $\beta$. We write
\begin{align}
    \E |Z(\beta)|^2 = \int DG \, D\Sigma\, \exp[-N S[G, \Sigma]],
\end{align}
for replica symmetric action
\begin{align}\label{eq:S}
    S = - \log \pf\lr{\partial_\tau -\hat\Sigma_{aa}} + \frac{1}{2} \int_0^{|\beta|} d\tau_1\,d\tau_2\, \lr{\Sigma_{aa}(\tau_1,\tau_2) G_{aa}(\tau_1,\tau_2) - \frac{2^{q-1}\cJ^2}{q^2} s_{aa} G_{aa}(\tau_1,\tau_2)^q},
\end{align}
where
\begin{align}
    s_{LL} = \lr{\frac{\beta}{|\beta|}}^2, \quad s_{RR} = \lr{\frac{\bar\beta}{|\beta|}}^2.
\end{align}
We solve the saddle point equations
\begin{align}
    \partial_{\tau_1} G_{ab}(\tau_1,\tau_2) - \int d\tau\, \Sigma_{ac}(\tau_1, \tau) G_{cb}(\tau, \tau_2) &= \delta(\tau_1-\tau_2) \delta_{ab}\label{eq:sp1}\\
    \Sigma_{ab}(\tau_1,\tau_2) &= s_{ab} \frac{2^{q-1}\cJ^2}{q} G_{ab}(\tau_1,\tau_2)^{q-1},\label{eq:sp2}
\end{align}
under the RS ansatz $G_{a\neq b} = \Sigma_{a \neq b} = 0$. Note that if $|\beta|$ were superpolynomially large in $N$ (e.g., at the Thouless time), one may need to solve for the off-diagonal under a more general ansatz to describe phenomena such as the ramp in the spectral form factor~\cite{khramtsov2021spectral}.
We require periodic solutions of the form $G_{aa}(\tau_1, \tau_2)=G_{aa}(\tau_1-\tau_2)$.
Similarly to the Euclidean case, an analytical solution is only available in the large-$q$ limit; we thus make ansatz
\begin{align}
    G_{aa}(\tau) = G_{aa}^{(0)}(\tau) + \frac{1}{q}G_{aa}^{(1)}(\tau), \quad \Sigma_{aa}(\tau) = \Sigma_{aa}^{(0)}(\tau) + \frac{1}{q} \Sigma_{aa}^{(1)}(\tau).
\end{align}
We factor out the free fermion propagator
\begin{align}
    G_{LL}^{(0)}(\tau_1,\tau_2) = G_{RR}^{(0)}(\tau_1,\tau_2) = \frac{1}{2}\sgn(\tau_1-\tau_2)
\end{align}
and denote the $1/q$ correction by $g$,
\begin{align}
    G_{aa}(\tau_1,\tau_2) &= \frac{1}{2}\sgn(\tau_1-\tau_2)\lr{1 + \frac{g_{aa}(\tau_1,\tau_2)}{q}}.
\end{align}
Note that since $G_{aa}(\tau_1 + |\beta|, \tau_2) = -G_{aa}(\tau_1,\tau_2)$, we use the antiperiodic convention
\begin{align}
    \sgn(\tau+\beta) = -\sgn(\tau),
\end{align}
which requires $g$ to be periodic, i.e.,
\begin{align}\label{eq:gper}
    g_{aa}(\tau_1, \tau_2) = g_{aa}(\beta + \tau_1, \tau_2).
\end{align}
The second saddle point equation \eqref{eq:sp2} becomes in the large-$q$ limit:
\begin{align}
    \Sigma_{aa}(\tau_1,\tau_2) &= s_{aa}\cJ^2\frac{2^{q-1}}{q}\lr{G_{aa}^{(0)} + \frac{1}{q}G_{aa}^{(1)}}^{q-1} = \frac{s_{aa}\cJ^2\sgn(\tau_1-\tau_2) e^{g_{aa}(\tau_1,\tau_2)}}{q}.
\end{align}
Thus, $\Sigma_{aa}^{(0)}(\tau)=0$ and $\Sigma_{aa}^{(1)}(\tau)=s_{aa}\cJ^2\sgn(\tau) e^{g_{aa}(\tau)}$.
To analyze the first saddle point equation under the large-$q$ ansatz, we evaluate the first term in \eqref{eq:sp1} to obtain
\begin{align}
    \partial_{\tau_1}G^{(0)}(\tau_1,\tau_2) = \delta(\tau_1-\tau_2) \implies \partial_{\tau_1}G^{(1)}(\tau_1,\tau_2) - \int d\tau \,\Sigma^{(1)}(\tau_1,\tau) G^{(0)}(\tau, \tau_2) = 0.
\end{align}
Differentiating with respect to $\tau_2$ gives (using $\partial_{\tau_2}G^{(0)}(\tau,\tau_2) = -\delta(\tau-\tau_2)$)
\begin{align}
    0 = \partial_{\tau_1} \partial_{\tau_2} G^{(1)}(\tau_1,\tau_2) + \int d\tau \, \Sigma^{(1)}(\tau_1, \tau)\delta(\tau-\tau_2) = \partial_{\tau_1} \partial_{\tau_2} G^{(1)}(\tau_1,\tau_2) + \Sigma^{(1)}(\tau_1, \tau_2).
\end{align}
We combine the results from both of our saddle point equations to obtain
\begin{align}\label{eq:combsaddle}
    \frac{1}{2}\partial_{\tau_1} \partial_{\tau_2} \lr{\sgn(\tau_1-\tau_2) g_{aa}(\tau_1,\tau_2)} &= -s_{aa} \cJ^2 \sgn(\tau_1-\tau_2) e^{g_{aa}(\tau_1,\tau_2)}.
\end{align}
We restrict our attention to solutions satisfying translational symmetry, which further simplifies the saddle point equations to
\begin{align}
    \frac{1}{2}\partial_\tau^2 [\sgn(\tau)g_{aa}(\tau)] = s_{aa} \cJ^2 \sgn(\tau) e^{g_{aa}(\tau)}.
\end{align}
This is solved by
\begin{align}
    e^{g_{aa}(\tau)} = \frac{\tilde c_a^2}{s_{aa}\cJ^2 \cos^2\lr{\tilde c_a|\tau| + \tilde d_a}},
\end{align}
where we now must identify constants $\tilde c_a$ and $\tilde d_a$. Besides the initial antisymmetry constraints
\begin{align}
    G_{ab}(\tau_1,\tau_2) = -G_{ba}(\tau_2,\tau_1), \quad \Sigma_{ab}(\tau_1,\tau_2) = -\Sigma_{ba}(\tau_2,\tau_1)
\end{align}
required of $G,\Sigma$, the translational invariant solutions satisfy $G_{ab}(\tau_1,\tau_2) = G_{ab}(\tau_1-\tau_2)$ and similarly for $\Sigma$. Combined with the antiperiodicity constraint, this implies
\begin{align}
    G_{ab}(\tau) = G_{ba}(|\beta|-\tau).
\end{align}
We also require smoothness at $|\beta|/2$,
\begin{align}
    g'_{aa}(|\beta|/2) = 0,
\end{align}
and the free fermion limit
\begin{align}
    g_{aa}(0) = 0.
\end{align}
This gives solution
\begin{align}\label{eq:gaa}
    e^{g_{aa}(\tau)} = \lr{\frac{\cos(c_a/2)}{\cos\left[c_a \lr{\frac{1}{2} - \frac{\tau}{|\beta|}}\right]}}^2
\end{align}
for $c_a$ satisfying
\begin{align}
    c_a^2 = s_{aa} \lr{|\beta|\cJ}^2 \cos^2 \frac{c_a}{2}.
\end{align}

\subsubsection{Computing the action at the saddle point}
We now identify which saddle point solution dominates the action. Starting from the action \eqref{eq:S}, we expand the first term
\begin{align}
    - \log \pf(\partial_\tau - \hat \Sigma_{aa}) &= -\frac{1}{2} \Tr \log\lr{\partial_\tau - \hat\Sigma_{aa}} \\
    &= - \Tr \log(\partial_\tau) + \frac{1}{2q}\Tr\lr{G^{(0)} \cdot \Sigma^{(1)}} + \frac{1}{4q^2} \Tr\lr{G^{(0)} \cdot \Sigma^{(1)} \cdot G^{(0)} \cdot \Sigma^{(1)}} + O\lr{\frac{1}{q^3}}.
\end{align}
We similarly expand the remaining term in \eqref{eq:S}
\begin{align}
    \frac{1}{2} \int_0^{|\beta|} d\tau_1\,d\tau_2\, \lr{\Sigma_{aa}(\tau_1,\tau_2) G_{aa}(\tau_1,\tau_2) - \frac{2^{q-1}\cJ^2}{q^2} s_{aa} G_{aa}(\tau_1,\tau_2)^q}
\end{align}
using the large-$q$ result
\begin{align}
    G_{aa}(\tau_1,\tau_2)^q = \left[\frac{1}{2}\sgn(\tau_1-\tau_2)\lr{1 + \frac{g_{aa}(\tau_1,\tau_2}{q}}\right]^q = \frac{e^{g_{aa}(\tau_1,\tau_2)}}{2^q}
\end{align}
to obtain (using $G^{(0)}(\tau_1,\tau_2) = -G^{(0)}(\tau_2,\tau_1)$)
\begin{align}
    &\frac{1}{2}\int_0^{|\beta|}d\tau_1\,d\tau_2\,\Bigg[\frac{1}{q} \Sigma_{aa}^{(1)}(\tau_1,\tau_2) G_{aa}^{(0)}(\tau_1,\tau_2) + \frac{1}{q^2}\Bigg(\Sigma_{aa}^{(1)}(\tau_1,\tau_2) G_{aa}^{(0)}(\tau_1,\tau_2) g_{aa}(\tau_1,\tau_2) - \frac{s_{aa}\cJ^2}{2} e^{g_{aa}(\tau_1,\tau_2)}\Bigg)\Bigg]\\
    &= -\frac{1}{2q} \Tr\lr{G^{(0)} \cdot \Sigma^{(1)}} + \frac{1}{2q^2} \int_0^{|\beta|}d\tau_1\,d\tau_2\,\Bigg(\Sigma_{aa}^{(1)}(\tau_1,\tau_2) G_{aa}^{(0)}(\tau_1,\tau_2) g_{aa}(\tau_1,\tau_2) - \frac{s_{aa}\cJ^2}{2} e^{g_{aa}(\tau_1,\tau_2)}\Bigg).
\end{align}
Altogether, this gives action
\begin{align}\label{eq:Sq2}
    S &= - \Tr \log(\partial_\tau) + \frac{1}{4q^2} \Tr\lr{G^{(0)} \cdot \Sigma^{(1)} \cdot G^{(0)} \cdot \Sigma^{(1)}}\nonumber\\
    &\quad + \frac{1}{2q^2}\int_0^{|\beta|}d\tau_1\,d\tau_2\,\Bigg(\Sigma_{aa}^{(1)}(\tau_1,\tau_2) G_{aa}^{(0)}(\tau_1,\tau_2) g_{aa}(\tau_1,\tau_2) - \frac{s_{aa}\cJ^2}{2} e^{g_{aa}(\tau_1,\tau_2)}\Bigg)\Bigg] + O\lr{\frac{1}{q^3}}.
\end{align}
Finally, we integrate out $\Sigma^{(1)}$ in the action, rewriting
\begin{align}
    \E |Z(\beta)|^2 = \int DG \, D\Sigma\, \exp[-N S[G, \Sigma]] = \int Dg \exp[-NI[g]].
\end{align}
We apply \eqref{eq:combsaddle} and the definition of $G^{(0)}$, i.e.,
\begin{align}
    \Sigma^{(1)}(\tau_1, \tau_2) = -\frac{1}{2}\partial_{\tau_1} \partial_{\tau_2} \lr{\sgn(\tau_1-\tau_2) g_{aa}(\tau_1,\tau_2)}, \quad G^{(0)}(\tau_1,\tau_2) = \frac{1}{2}\sgn(\tau_1-\tau_2),
\end{align}
to obtain the identity (via integration by parts)
\begin{align}
    \int_0^{|\beta|} d\tau \,G^{(0)} (\tau_1,\tau) \Sigma^{(1)}(\tau, \tau_2) &= -\frac{1}{4}\int_0^{|\beta|} d\tau\, \sgn(\tau_1-\tau) \partial_\tau \partial_{\tau_2}[\sgn(\tau-\tau_2) g_{aa}(\tau-\tau_2)]\\
    &= -\frac{1}{4}\partial_{\tau_2}\Bigg\{\left[\sgn(\tau_1-\tau) \sgn(\tau-\tau_2) g_{aa}(\tau,\tau_2)\right]_0^{|\beta|}\nonumber\\
    &\qquad - \int_0^{|\beta|} d\tau\, \partial_\tau[\sgn(\tau_1-\tau)] \cdot [\sgn(\tau-\tau_2)g_{aa}(\tau,\tau_2)]\Bigg\}\\
    &= 0 - \frac{1}{2}\partial_{\tau_2} \int_0^{|\beta|} d\tau\, \delta(\tau_1-\tau) \sgn(\tau - \tau_2) g_{aa}(\tau, \tau_2)\\
    &= -\frac{1}{2}\partial_{\tau_2}[\sgn(\tau_1-\tau_2)g_{aa}(\tau_1,\tau_2)].
\end{align}
Above, we used the periodicity of $g_{aa}$ \eqref{eq:gper}
to show that the boundary term vanishes:
\begin{align}
    \left[\sgn(\tau_1-\tau) \sgn(\tau-\tau_2) g_{aa}(\tau,\tau_2)\right]_0^{|\beta|} &= [(-1)(+1)g(\beta,\tau_2)] - [(+1)(-1)g(0, \tau_2)] = 0.
\end{align}
This allows us to evaluate
\begin{align}
    &\frac{1}{2q^2} \lr{\frac{1}{2}\Tr\lr{G^{(0)} \cdot \Sigma^{(1)} \cdot G^{(0)} \cdot \Sigma^{(1)}} + \int_0^{|\beta|}d\tau_1\,d\tau_2\,\Sigma_{aa}^{(1)}(\tau_1,\tau_2) G_{aa}^{(0)}(\tau_1,\tau_2) g_{aa}(\tau_1,\tau_2)}\\
    &= \frac{1}{4q^2}\int_0^{|\beta|} d\tau_1\, d\tau_2\, \lr{-\frac{1}{2} \partial_{\tau_2}[\sgn(\tau_1-\tau_2)g_{aa}(\tau_1,\tau_2)]}\lr{-\frac{1}{2} \partial_{\tau_1}[\sgn(\tau_2-\tau_1)g_{aa}(\tau_2,\tau_1)]}\nonumber\\
    &\quad - \frac{1}{8q^2} \int_0^{|\beta|}d\tau_1\,d\tau_2\, [\partial_{\tau_1} \partial_{\tau_2} \lr{\sgn(\tau_1-\tau_2) g_{aa}(\tau_1,\tau_2)}] \sgn(\tau_1-\tau_2) g_{aa}(\tau_1,\tau_2)\\
    &= \frac{1}{16q^2} \int_0^{|\beta|} d\tau_1 \, d\tau_2\, \lr{-(\partial_{\tau_1} F)(\partial_{\tau_2} F) - 2 F \partial_{\tau_1} \partial_{\tau_2} F}
\end{align}
where we defined $F(\tau_1,\tau_2) = \sgn(\tau_1-\tau_2)g_{aa}(\tau_1,\tau_2)$ and applied $g_{aa}(\tau_1,\tau_2) = g_{aa}(\tau_2,\tau_1)$. Evaluating the first integral by parts, we find that
\begin{align}
    \int_0^{|\beta|} d\tau_1\, d\tau_2\, (\partial_{\tau_1} F)(\partial_{\tau_2} F) &= \int_0^{|\beta|} d\tau_2 \, \left[F\partial_{\tau_2} F\right]_{\tau_1=0}^{\tau_1=|\beta|} - \int_0^{|\beta|} d\tau_1\, d\tau_2\, F \partial_{\tau_1} \partial_{\tau_2} F
\end{align}
where the boundary term vanishes by periodicity of $g_{aa}(|\beta|, \tau_2) = g_{aa}(0, \tau_2)$:
\begin{align}
    \left[F\partial_{\tau_2} F\right]_{\tau_1=0}^{\tau_1=|\beta|} = g_{aa}(|\beta|, \tau_2) \partial_{\tau_2}\lr{\sgn(|\beta|-\tau_2) g_{aa}(|\beta|, \tau_2)} + g_{aa}(0, \tau_2) \partial_{\tau_2}\lr{\sgn(-\tau_2) g_{aa}(0, \tau_2)} = 0.
\end{align}
Hence, we have
\begin{align}
    &\frac{1}{2q^2} \lr{\frac{1}{2}\Tr\lr{G^{(0)} \cdot \Sigma^{(1)} \cdot G^{(0)} \cdot \Sigma^{(1)}} + \int_0^{|\beta|}d\tau_1\,d\tau_2\,\Sigma_{aa}^{(1)}(\tau_1,\tau_2) G_{aa}^{(0)}(\tau_1,\tau_2) g_{aa}(\tau_1,\tau_2)}\\
    &= -\frac{1}{16q^2}\int_0^{|\beta|} d\tau_1\, d\tau_2\, \sgn(\tau_1-\tau_2)g_{aa}(\tau_1,\tau_2)\partial_{\tau_1}\partial_{\tau_2}[\sgn(\tau_1-\tau_2)g_{aa}(\tau_1,\tau_2)]\\
    &= \frac{s_{aa} \cJ^2}{8q^2}\int_0^{|\beta|} d\tau_1\, d\tau_2\, g_{aa}(\tau_1,\tau_2) e^{g_{aa}(\tau_1,\tau_2)}
\end{align}
where we applied \eqref{eq:combsaddle} in the last line.
The final term of \eqref{eq:Sq2}
\begin{align}
    - \frac{1}{4q^2} \int_0^{|\beta|} s_{aa}\cJ^2 e^{g_{aa}(\tau_1,\tau_2)}
\end{align}
survives, producing action
\begin{align}
    I[g] &= -\Tr \log \partial_\tau - \frac{s_{aa}\cJ^2}{4q^2} \int_0^{|\beta|} d\tau_1\, d\tau_2\, \lr{1 - \frac{1}{2}g_{aa}(\tau_1,\tau_2)}e^{g_{aa}(\tau_1,\tau_2)}\\
    &=  -\Tr \log \partial_\tau - \frac{s_{aa}|\beta|\cJ^2}{2q^2} \int_0^{|\beta|/2} d\tau\,  \lr{1 - \frac{1}{2}g_{aa}(\tau)}e^{g_{aa}(\tau)}.
\end{align}
Finally, we plug in our translation-invariant solution \eqref{eq:gaa} for $g_{aa}$ to obtain
\begin{align}
    I[g] &= -\Tr \log \partial_\tau - \frac{s_{aa}}{2} \lr{\frac{|\beta|\cJ}{q}}^2 \lr{\frac{\sin c_a}{c_a} - \frac{1+ \cos c_a}{4}}, \quad c_a^2 = s_{aa}(|\beta|\cJ)^2 \cos^2 \frac{c_a}{2}.
\end{align}
Using the identity
\begin{align}
    \cos c_a = -1 + 2 \cos^2\frac{c_a}{2} = -1 + \frac{2c_a^2}{s_{aa}(|\beta|\cJ)^2}
\end{align}
we further simplify this to
\begin{align}
    I[g] &= -\Tr \log \partial_\tau - \frac{1}{2q^2} \lr{\frac{s_{aa}(|\beta|\cJ)^2\sin c_a}{c_a} - \frac{c_a^2}{2}} = -\Tr \log \partial_\tau + \frac{1}{2q^2} \lr{\frac{c_a^2}{2} - 2c_a \tan \frac{c_a}{2}}.
\end{align}
This is similar to the result in Eq. (3.24) of~\cite{khramtsov2021spectral}. Note that we assumed $\cos c_a/2 \neq 0$ due to the saddle point equation $c_a^2 = s_{aa}(|\beta|\cJ)^2 \cos^2 \frac{c_a}{2}$.
Since the term $-\Tr \log \partial_\tau$ simply gives $\exp[N \Tr \log \partial_\tau] = 2^N$, we can omit it from the remaining analysis.

\subsubsection{Identifying the leading saddle}
For a given choice of $\beta$, let $c_L, c_R$ be solutions to
\begin{align}\label{eq:clr}
    c_L^2 = (\beta \cJ)^2 \cos^2(c_L/2), \quad c_R^2 = (\bar\beta \cJ)^2 \cos^2(c_R/2).
\end{align}
Since
\begin{align}
    \overline{c_R}^2 = (\beta \cJ)^2 \cos^2 \frac{\overline{c_R}}{2},
\end{align}
we can always identify solutions $c_R$ by $c_R = \overline{c_L}$. Moreover, for each solution $c_L$, we note that $-c_L$ is also a solution; hence, we can fix the sign and look for solutions to
\begin{align}\label{eq:cl}
    c_L = \beta \cJ \cos \frac{c_L}{2}.
\end{align}
In Ref.~\cite{khramtsov2021spectral}, the leading saddle of the spectral form factor was correctly identified via a large-$|\beta|\cJ$ expansion. Solving \eqref{eq:cl} the in large-$|\beta| \cJ$ limit gives for integer $m$
\begin{align}
    \Re c_L = (2m+1)\pi + O\lr{\frac{1}{(|\beta|\cJ)^2}}, \quad \Im \frac{c_L}{2} = \frac{(2m+1)\pi}{\beta \cJ} + O\lr{\frac{1}{(|\beta|\cJ)^2}}.
\end{align}
The corresponding action is
\begin{align}
    -\frac{N}{q^2}\Re\lr{\frac{c_L^2}{2} - 2c_L \tan \frac{c_L}{2}} = -\frac{N}{q^2}\lr{4\Re \beta \cJ + 4 + \frac{(2m+1)^2\pi^2}{2} + O\lr{\frac{1}{|\beta| \cJ}}}.
\end{align}
Since $m=0$ maximizes the action, we obtain solution
\begin{align}\label{eq:cstar}
    c_* \approx \pi - O\lr{\frac{1}{\beta \cJ}}.
\end{align}
Although this does not maximize the action at all $\beta$, it is identified by~\cite{khramtsov2021spectral} as the correct saddle at all $\beta$ in the special case of the spectral form factor (Fig.~\ref{fig:sff}) under the above large-$|\beta|\cJ$ justification.

\begin{figure}
    \centering
    \includegraphics[width=\linewidth]{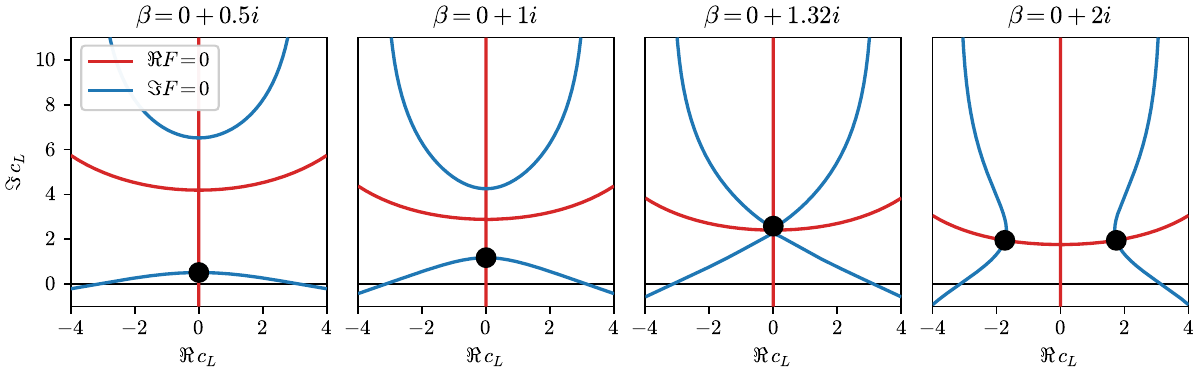}
    \caption{Red lines show the contour of $\Re(c_L - \beta \cJ \cos c_L/2) = 0$; blue lines show the contour of $\Im(c_L - \beta \cJ \cos c_L/2) = 0$; we set $\cJ = 1$. A black dot indicates the intersection at the $c_*$ branch identified as the leading saddle in~\cite{khramtsov2021spectral}. At $\beta \approx 1.32i$, the leading saddle has multiplicity 2, causing a complex zero in the partition function (as derived in \eqref{eq:132}). At $|\Im \beta| \gtrsim 1.32$, the two saddles can cancel each other in the action, leading to additional complex zeros.}
    \label{fig:sff}
\end{figure}

Note that the above argument holds for generic complex $\beta$ as well. We further extend the above argument to check that $c_*$ is the leading saddle for small $|\beta|$ as well: we find that solutions other than $c_*$ yield diverging partition functions as $|\beta| \cJ \to 0$. For $c_L/\cos(c_L/2)$ to vanish, we require that either $c_L \to 0$ or $|\cos c_L/2| \to \infty$. The first case,
\begin{align}
    c_* = \beta \cJ + O\lr{(|\beta| \cJ)^3},
\end{align}
corresponds to the previous solution $c_*$ in \eqref{eq:cstar} and yields the correct limit $\E |Z(\beta)|^2 \to 2^N$ as $\beta \to 0$. In contrast, all other solutions for $c_L$ cause $\E |Z(\beta)|^2$ diverge in the $\beta \to 0$ limit. Since
\begin{align}
    \left|\cos \frac{c}{2}\right| \leq \cosh \Im \frac{c_L}{2},
\end{align}
we must satisfy $\cosh \Im \frac{c_L}{2} \to \infty$ as $|\beta| \cJ \to 0$. Explicitly, these solutions are given by the solutions to (taking the case $\Im c_L > 0$)
\begin{align}
    \beta \cJ \approx 2 c_L e^{ic_L/2}.
\end{align}
These are solved to leading order by
\begin{align}
    c_L \approx -2iW_m\lr{\frac{i\beta \cJ}{4}}
\end{align}
for integer $m \neq 0$ and $W_m$ the Lambert $W$ function.
Expanding $W_m(z) \approx \log z + 2\pi i m$ for small $z$, we obtain
\begin{align}
    c_L \approx -2i \log \frac{i\beta \cJ}{4} + 4\pi m.
\end{align}
This satisfies $|\Im c_L| \to \infty$ as $|\beta| \cJ \to 0$. In particular, the action to leading order is
\begin{align}
    -\frac{N}{q^2}\Re\lr{\frac{c_L^2}{2} - 2c_L \tan \frac{c_L}{2}} = \frac{2N}{q^2} \log^2 \frac{4}{|\beta| \cJ},
\end{align}
which diverges in the small-$|\beta| \cJ$ limit. We conclude that the only physical solution as $|\beta| \cJ \to 0$ is the $c_*$ solution identified in \eqref{eq:cstar}.

\begin{figure}
    \centering
    \includegraphics[width=\linewidth]{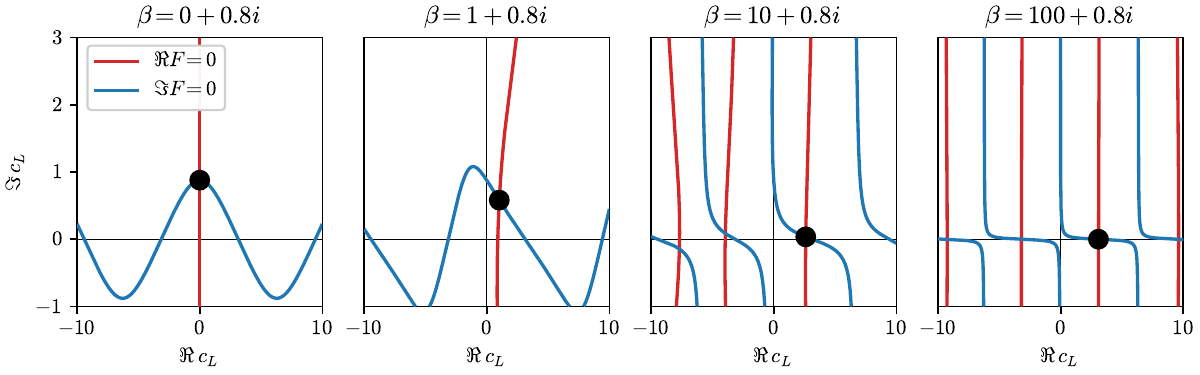}
    \caption{Red lines show the contour of $\Re(c_L - \beta \cJ \cos c_L/2) = 0$; blue lines show the contour of $\Im(c_L - \beta \cJ \cos c_L/2) = 0$; we set $\cJ = 1$. A black dot indicates the intersection at the $c_*$ branch identified as the leading saddle for all $\beta$ such that $|\Im \beta| \lesssim 1.32$. At $|\beta| \cJ \ll 1$, it corresponds to $c_* \approx \beta \cJ$; at $|\beta| \cJ \gg 1$, it approaches $c_* \approx \pi$.}
    \label{fig:contours}
\end{figure}

\subsection{Harmonic function and Jensen's formula}
We now show that for our choice of $c_*$,
\begin{align}
    \Re \lr{\frac{c_*^2}{2} - 2 c_* \tan \frac{c_*}{2}}
\end{align}
is harmonic on a region $\cR$ in order to apply Jensen's theorem. Since $c_* \to -c_*$ contributes identically to the action, we fix the sign and rewrite our condition on $c_L$ as $F(c_L, \beta \cJ) = 0$ for
\begin{align}
    F(c, b) = c - b \cos \frac{c}{2}.
\end{align}
The holomorphic implicit function theorem states that if
\begin{align}
    F(c_0, b_0) = 0 \quad \text{and} \quad \partial_c F(c_0, b_0) = 1 + \frac{b_0}{2} \sin \frac{c_0}{2} \neq 0
\end{align}
then there is a unique holomorphic solution $c(b)$ with $c(b_0) = c_0$. Let $c_*$ denote the branch characterized by $c_*(b) = b + O(b^3)$ near $b=0$; this is a unique holomorphic solution since
\begin{align}
    \partial_c F(0, 0) = 1 \neq 0.
\end{align}
We seek to identify region $\cR$ including the origin such that $F(c_*, b) = 0$ and $\partial_c F(c_*, b) \neq 0$ on all of $\cR$; by the above, $c_*$ extends to a single-valued holomorphic function on all of $\cR$. To show that $\Re\lr{\frac{c^2}{2} - 2c \tan \frac{c}{2}}$ is harmonic, we note that
\begin{align}
    H(c) = \frac{c^2}{2} - 2c \tan \frac{c}{2}
\end{align}
is holomorphic in $c$ away from the poles of $\tan c/2$, i.e., away from $c = (2k+1)\pi$ for integer $k$. However, if $c=(2k+1)\pi$ then $\cos(c/2) = 0$; then the condition $F(c, b) = 0$ would require $c_*=0$ for all $b$, which contradicts $c_*=(2k+1)\pi$. Hence, $\Re H(c_*)$ is harmonic on $\cR$.

We now find a region $\cR$ such that the pair of conditions $F(c_*, b) = 0, \partial_c F(c_*, b) = 0$ never occur simultaneously. We rewrite this condition as
\begin{align}
    F(c, b) = 0 \text{ and } \partial_c F(c, b) = 0 \quad \Longleftrightarrow \quad b = \frac{c}{\cos(c/2)} \text{ and } 1 + \frac{c}{2} \tan \frac{c}{2} = 0.
\end{align}
obtained by simplifying
\begin{align}
    \partial_c F(c, b) = 1 + \frac{b}{2}\sin\frac{c}{2} = 1 + \frac{1}{2} \frac{c}{\cos(c/2)} \sin \frac{c}{2} = 1 + \frac{c}{2} \tan \frac{c}{2}.
\end{align}
We begin be restricting the values of $c$ that can satisfy these conditions.

\begin{lemma}\label{lem:cpure}
    If $1 + \frac{c}{2} \tan \frac{c}{2} = 0$, then either $\Re c = 0$ or $\Im c = 0$.
\end{lemma}
\begin{proof}
    Let $z = \frac{c}{2} = x + iy$ so the condition is
    \begin{align}
        \tan z = -\frac{1}{z}.
    \end{align}
    Equating imaginary and real parts gives equations
    \begin{align}
        \frac{\sinh 2y}{\cos 2x + \cosh 2y} = \frac{y}{x^2 + y^2}, \quad \frac{\sin 2x}{\cos 2x + \cosh 2y} = - \frac{x}{x^2 + y^2}.
    \end{align}
    Assume $y \neq 0$. Dividing the two equations gives
    \begin{align}
        \frac{\sin 2x}{\sinh 2y} = -\frac{x}{y}.
    \end{align}
    At all $x,y \neq 0$, we have $\sinh 2y > 2|y|$ and $|\sin 2x| < 2|x|$. Hence,
    \begin{align}
        \left|\frac{x}{y}\right| = \left|\frac{\sin 2x}{\sinh 2y}\right| > \frac{2|x|}{2|y|} = \frac{|x|}{|y|}
    \end{align}
    which is a contradiction. We conclude that all solutions must have $x=0$ or $y=0$.
\end{proof}

Since $c$ is either purely real or purely imaginary by \Cref{lem:cpure}, then
\begin{align}
    b = \frac{c}{\cos c/2}
\end{align}
is purely real (if $c$ is real) or imaginary (if $c$ is imaginary). However, the branch $c_*$ corresponds to $c_* \approx b$ as $b \to 0$, so it satisfies $-\pi < c_* < \pi$. This satisfies $1 + \frac{c}{2} \tan \frac{c}{2} > 0$, leaving only the case of purely imaginary $c$. This means that we can choose $\cR$ to be the region that excludes the imaginary axis starting from the closest zero.
We compute where the first zero is by numerically solving
\begin{align}
    1 + \frac{c}{2} \tan \frac{c}{2} = 0.
\end{align}
Since $f(y) = y \tanh y$ satisfies $f'(y) > 0$, this has a unique solution satisfying $\Im c > 0$. We obtain $c_0 \approx 2.4 i$, corresponding to
\begin{align}\label{eq:132}
    b_0 = \frac{c_0}{\cos c_0/2} \approx 1.325 i.
\end{align}
Note that this corresponds to the first zero in the spectral form factor as computed in~\cite{khramtsov2021spectral}.
Combining all of the above, we conclude that
\begin{align}
    \log \E \left|\frac{Z(\beta)}{Z(0)}\right|^2 = \exp[-\frac{N}{q^2}\Re\lr{\frac{c_*^2}{2}-2c_*\tan\frac{c_*}{2}}]
\end{align}
is harmonic on the region
\begin{align}
    \cR = \left\{\beta \cJ\,:\, \Re \beta \cJ \neq 0 \text{ or } |\Im \beta| \cJ < 1.3\right\}.
\end{align}
That is, all the zeros lie on the imaginary axis and satisfy $|\Im \beta| \cJ > 1.3$. We can thus complete the SYK analysis applying Jensen's formula on $\cR$ via the disk-to-wedge mapping of \Cref{lem:wedge-avoiding-mapping} (see \Cref{lem:jensen-via-wedge}). 
This completes \Cref{res:SYK_zero_free_region}, leading to \Cref{thm:Barvinok_SYK}. 

\subsection{Local observables} \label{sec:SYK_local_observables}
Above, we showed the zero-free region of the SYK partition function $Z(\beta) = \Tr[e^{-\beta H}]$. Here, we show that perturbing $Z(\beta, \lambda) = \Tr[e^{-\beta(H + \lambda O)}]$ for any local observable $O$ and any $\lambda = o(1)$ preserves the zero-free region.
This implies that
\begin{align}
    \langle O \rangle_\beta = - \frac{1}{\beta} \partial_\lambda \log Z(\beta,\lambda)\Bigg|_{\lambda=0}
\end{align}
can be estimated by evaluating the estimate of $\log Z(\beta,\lambda)$ at small $\lambda$ and applying a finite difference with $\log Z(\beta, 0)$ to estimate the derivative, as shown in \Cref{lem:QMV_from_QPF}.

Since the computation with $H \to H +\lambda O$ is largely identical to the previous computation, we only provide a brief sketch. The path integral on two copies gives
\begin{align}
    |Z(\beta,\lambda)|^2 = \int D\psi_i^a \exp[-\sum_a s_{aa} \int_0^{|\beta|} d\tau \lr{\frac{1}{2}\sum_j \psi_j^a \partial_\tau \psi_j^a + i^{q/2} \sum_{1 \leq i_1 < \cdots < i_q \leq N} J_{i_1\cdots i_q}\psi_{i_1}^a \cdots \psi_{i_q}^a + \lambda O^a}].
\end{align}
Taking the expectation over disorder, we have
\begin{align}
    \E |Z(\beta, \lambda)|^2 &= \int D\psi_i^a \exp[-\sum_a s_{aa} \int_0^{|\beta|} \lr{\frac{1}{2}\sum_j \psi_j^a \partial_\tau \psi_j^a + \lambda O^a}] \nonumber\\
    &\quad \times \int DG\, D\Sigma \exp[-\frac{N}{2} \sum_a \int_0^{|\beta|} d\tau_1\,d\tau_2\, \lr{\Sigma_{aa}(\tau_1,\tau_2) G_{aa}(\tau_1,\tau_2) - \frac{2^{q-1}\cJ^2}{q^2} s_{aa} G_{aa}(\tau_1,\tau_2)^q}].
\end{align}
Assuming $O$ is $m$-local, the first integral can be evaluated via Gaussian integration over the $N-m$ fermions not in $O$; this gives the same Pfaffian term as before. The remaining $m$ fermions produce a different contribution to the action, but since $m=O(1)$ this difference is a vanishingly small correction. That is, for
\begin{align}
    \E |Z(\beta, \lambda)|^2 = \int DG \, D\Sigma\, \exp[-N S[G, \Sigma]],
\end{align}
we have action
\begin{align}
    S &= \frac{1}{2} \int_0^{|\beta|} d\tau_1\,d\tau_2\, \lr{\Sigma_{aa}(\tau_1,\tau_2) G_{aa}(\tau_1,\tau_2) - \frac{2^{q-1}\cJ^2}{q^2} s_{aa} G_{aa}(\tau_1,\tau_2)^q} \nonumber\\
    &\quad - \lr{1 - O\lr{\frac{\lambda}{N}}}\log \pf\lr{\partial_\tau -\hat\Sigma_{aa}} + O\lr{\frac{\lambda}{N}},
\end{align}
which is an $O(\lambda/N)$ correction to the previous action \eqref{eq:S} since $m = O(1)$. Hence, the saddle point equations are unchanged, yielding the same solutions $G_*, \Sigma_*$ as before. In particular, for any $m=O(1)$, the zero-free region is preserved since
\begin{align}
    \E |Z(\beta, \lambda)|^2 = C_\lambda \cdot \E |Z(\beta, 0)|^2
\end{align}
for some $C_\lambda$ satisfying $|\log C_\lambda| = O(\lambda)$. In Jensen's formula applied to the wedge \eqref{eq:jensen-phi}, $C_\lambda$ contributes a remainder term: if $N_Z$ is the expected number of zeros of the unperturbed partition function (at constant $\beta$), the expected number of zeros of the perturbed partition function is upper-bounded by $N_Z + O(\lambda) = N_Z + o(1)$. We confirm this result numerically in the following section, which checks that the zero-free region of $Z(\beta, 0)$ is preserved in $Z(\beta, \lambda)$.

\subsection{Numerical experiments}\label{app:num}
We fix a plotting rectangle
\begin{equation}
R=[a,b]+i[c,d]\subset\mathbb{C}
\end{equation}
and, after verifying that $Z(\beta)\neq 0$ for $\beta\in\partial R$, count the number of zeros in $R$ by numerically integrating
\begin{equation}
N_R \;=\; \frac{1}{2\pi i}\oint_{\partial R} \partial_\beta \log Z(\beta)\, d\beta,
\label{eq:arg-principle}
\end{equation}
via a Riemann sum along the edges of $\partial R$. We check that our integral error is sufficiently small to produce an integer answer for $N_R$: for example, when $\beta \in [-300, 300] + i[-300, 300]$, numerical integration gives $N_R = 263.9997$. From the phase/magnitude plot of $Z$, we measure a the phase differences around each $\beta$ in the plot and establish a guess for $\beta_\mathrm{seed}$ where a complex zero is located. We then refine smaller rectangles around $\beta_\mathrm{seed}$
\begin{equation}
R_{\text{loc}} = [\Re\beta_{\text{seed}}-\Delta_x,\Re\beta_{\text{seed}}+\Delta_x]
+ i[\Im\beta_{\text{seed}}-\Delta_y,\Im\beta_{\text{seed}}+\Delta_y]
\end{equation}
and iteratively apply Jensen's formula to obtain $N_{R_\mathrm{loc}}$. Once $N_{R_\mathrm{loc}}=1$, we apply Newton's method to solve $Z(\beta)=0$ via 
\begin{equation}
\beta_{n+1}=\beta_n-\frac{Z(\beta_n)}{\partial_\beta Z(\beta_n)}.
\end{equation}
This is repeated until the all $N_R$ zeros in $R$ are explicitly identified. In the main text, we show the square of $\beta \in [-30, 30] + i[-30, 30]$; below, we extend this to a side length of 600. We similarly numerically verified that the zero-free region of the partition function is preserved upon perturbing by a local observable.

\begin{figure}[H]
    \centering
    \includegraphics[width=0.49\linewidth]{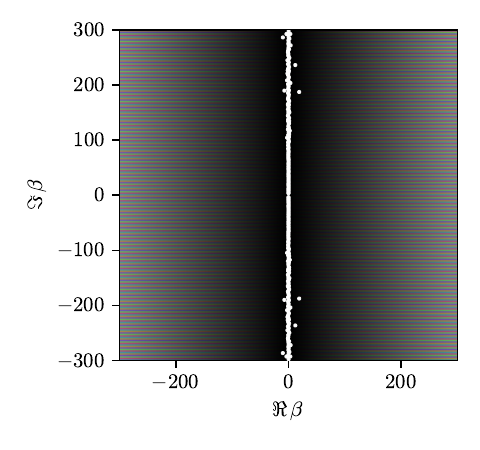}
    \caption{Complex zeros of $Z(\beta)$ obtained from exact diagonalization of a single $N=30, q=4$ SYK instance ($J=1$). Brightness indicates magnitude and hue indicates phase; zeros are indicated by white markers. Perturbing the Hamiltonian by a local observable (e.g., $H \to H + 0.01 i \psi_1 \cdots \psi_6$) yields a similar plot, where the complex zeros shift by a small but visually imperceptible amount.}
    \label{fig:sykzeros300}
\end{figure}

We repeat this procedure for several other models. In Fig.~\ref{fig:glasses}, we show two spin glasses known to have a glass transition that obstructs sampling algorithms~\cite{el2022sampling,el2025sampling,anschuetz2025efficient,zlokapa2025average}, and we numerically verify that complex zeros appear at similar values of $N$ as we take for the SYK model. In Fig.~\ref{fig:glasses}a, we show an Ising $p$-spin model
\begin{align}
    H = \sum_{1 \leq i_1 < \cdots < i_p \leq n} J_{i_1\cdots i_p} \sigma^Z_{i_1} \cdots \sigma^Z_{i_p}
\end{align}
for i.i.d. coefficients $J_{i_1\cdots i_p} \sim \cN(0, p! / 2N^{p-1})$ at $p=3$. This has dynamical and static 1RSB phases at constant temperature~\cite{montanari2003nature,ferrari2012two,talagrand2000rigorous,gamarnik2025shattering,alaoui2024near}. In Fig.~\ref{fig:glasses}b, we show a quantum Heisenberg model
\begin{align}
    H = \sum_{i < j} J_{ij} \vec S_i \cdot \vec S_j
\end{align}
for SU(2) spins $\vec S^2 = S(S+1)$ and Gaussian coefficients $J_{ij} \sim \cN(0,1/N)$. This has a full RSB phase transition at constant temperature~\cite{georges2001quantum,christos2022spin,kavokine2024exact} and hence a similar algorithmic obstruction. For both of these models, we see complex zeros emerging completely unlike the SYK model, which ultimately obstructs Barvinok's algorithm. Note that our SYK model analysis at $N=30$ corresponds to a Hilbert space of dimension $2^{15}$, and thus we take an $N=15$ quantum Heisenberg model.

We also consider a $3 \times 3$ 2D Fermi-Hubbard model with periodic boundary conditions
\begin{align}
    H = -t \sum_{\langle i,j \rangle} \sum_{\sigma \in \{\uparrow,\downarrow\}} \lr{c_{i\sigma}^\dagger c_{j\sigma} + c_{j\sigma} c_{i\sigma}^\dagger} + U\sum_i n_{i\uparrow}n_{i\downarrow},
\end{align}
where $n_{i\sigma} = c_{i\sigma}^\dagger c_{i\sigma}$. Since the number of orbitals is $2N_\mathrm{sites}$, the Hilbert space dimension is $2^{18}$. We choose $t=1, U=8, \mu=2$ and compute the resulting doping at $\beta=10$ via
\begin{align}
    \delta = 1 - \frac{\langle N \rangle}{N_\mathrm{sites}}, \quad \langle N \rangle = \frac{1}{Z} \sum_{N,\alpha} N e^{-\beta\lr{E_{N,\alpha}- \mu N}}
\end{align}
for $N = \sum_i \lr{n_{i\uparrow} + n_{i\downarrow}}$ and eigenvalues $E_{N,\alpha}$ corresponding to the $N$-particle subspace with eigenvalues indexed by $\alpha$. Note that $\delta = 0$ corresponds to half-filling; because $\delta > 0$ the hole-doped model has a sign problem. Fig.~\ref{fig:fh} suggests that a disk-to-strip map can be used to compute thermal expectations at similar temperatures to state-of-the-art experimental results~\cite{xu2025neutral}, and more tailored analytic continuations (e.g., patching disks~\cite{eldar2018approximating}) may yield classical algorithms with better parameters.

\begin{figure}[H]
    \centering
    \begin{minipage}{0.49\textwidth}
        \centering
        \includegraphics[width=\textwidth]{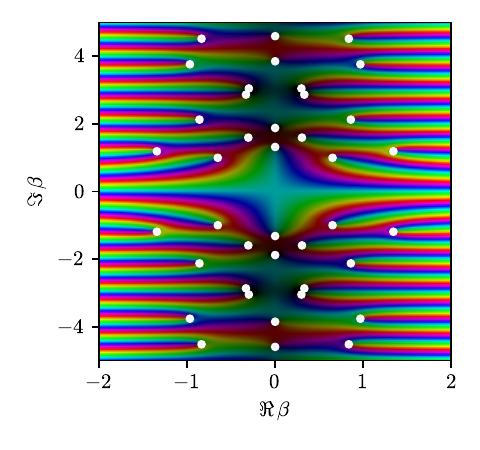}\\
        \small\textbf{(a)} Ising 3-spin model ($N=18$)
    \end{minipage}
    \hfill
    \begin{minipage}{0.49\textwidth}
        \centering
        \includegraphics[width=\textwidth]{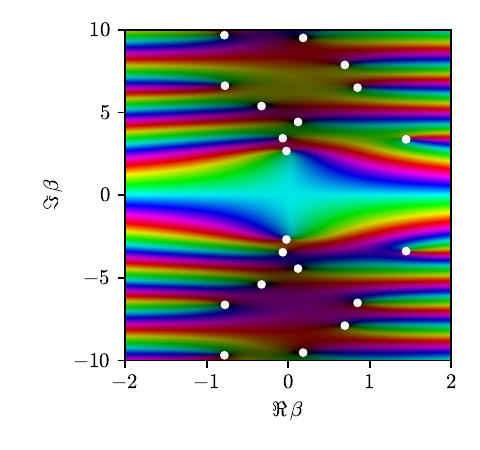}\\
        \small\textbf{(b)} Quantum Heisenberg model ($N=15$)
    \end{minipage}

    \caption{Complex zeros of $Z(\beta)$ for two models with a glass phase. Brightness indicates magnitude ($\gamma$-corrected for visual clarity) and hue indicates phase; zeros are indicated by white markers. In the Ising model, zeros appear outside a circle of constant radius. In the quantum Heisenberg model, a line of zeros appears parallel to a line of zeros around $\Re \beta = 0$ and with similar density. These are distinct from the SYK model (Fig.~\ref{fig:sykzeros} in the main text and Fig.~\ref{fig:sykzeros300} above), which exclusively has zeros on the $\Re \beta = 0$ axis for similar system sizes.}
    \label{fig:glasses}
\end{figure}

\begin{figure}[H]
    \centering
    \begin{minipage}{0.49\textwidth}
        \centering
        \includegraphics[width=\textwidth]{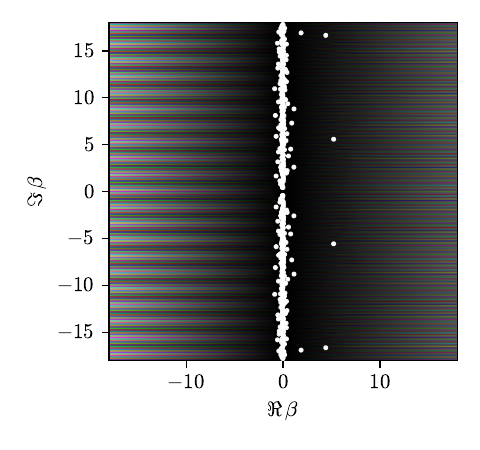}\\
        \small\textbf{(a)} $|\Re \beta|, |\Im \beta| \leq 18$
        \label{fig:fhbeta18}
    \end{minipage}
    \hfill
    \begin{minipage}{0.49\textwidth}
        \centering
        \includegraphics[width=\textwidth]{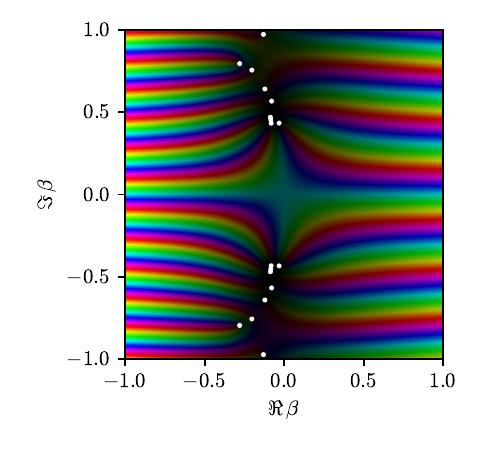}\\
        \small\textbf{(b)} $|\Re \beta|, |\Im \beta| \leq 1$
        \label{fig:fhbeta1}
    \end{minipage}

    \caption{Complex zeros of $Z(\beta)$ for the $3\times 3$ doped Fermi-Hubbard model with $U/t=8$ and $\delta \approx 0.1$ ($\mu=2$) obtained from exact diagonalization. Brightness indicates magnitude and hue indicates phase; zeros are indicated by white markers. A zero-free strip of height $|\Im \beta| \leq 0.4$ exists for all $-18 < \Re \beta < 18$.}
    \label{fig:fh}
\end{figure}

\section{Out-of-time-order correlators}
\label{app:otoc}
We consider an $n$-qubit system with Hilbert space $\mathcal{H}=(\mathbb{C}^2)^{\otimes n}$.
We consider the task of estimating OTOC function $f_{2\OTOCorder}(t)$ as formalized in \cite{google2025observation,king2025simplified} and defined below.
Let
\begin{equation}
B(s):=e^{iHs} B e^{-iHs},\qquad X_{2\OTOCorder}(s):=(B(s)M)^{2\OTOCorder},
\end{equation}
where $B$ and $M$ are single-qubit Pauli observables acting on distinct qubits $b\neq m$. 
At times, we will denote $X_{2\OTOCorder}(s)$ as simply $X(s)$ to simplify notation.
Throughout, $H=\sum_{e=1}^m h_e$ is a $\locality$-local Hamiltonian satisfying:
\begin{itemize}
\item (Locality) Each $h_e$ acts nontrivially on at most $\locality$ qubits.
\item (Bounded degree) Each qubit participates in at most $D$ local terms $h_e$.
\item (Uniform norm bound) $\|h_e\|\le J$ for all $e$.
\end{itemize}
We aim to approximate the OTOC function defined below for real valued $t$:
\begin{equation}
    f_{2\OTOCorder}(t) := \Tr \left[\rho X_{2\OTOCorder}(t)\right].
\end{equation}
Above, $\rho$ is a specified initial state.
Note that for real valued $t$, $|f_{2\OTOCorder}(t)|\leq 1$ and we need to approximate $f_{2\OTOCorder}(t)$ to $1/\operatorname{poly}(n)$ additive error to compete with quantum methods. 
This is a marked difference from partition functions $Z_H(\beta)$ where $|Z_H(\beta)|=\exp(\Theta(n))$ and relative error is often desirable (approximating $Z_H(\beta)$ to $1/\operatorname{poly}(n)$ additive error is a much harder problem).

We will run Barvinok's interpolation on the function $\log(2-f_{2\OTOCorder}(s))$, the constant $2$ being chosen by convenience to form a zero-free region.
Alternatively, we could interpolate directly the function $f_{2\OTOCorder}(s)$ using \Cref{lem:disk-analytic}, but we stick to the form $\log(2-f_{2\OTOCorder}(s))$ to frame our results in the zero-free region setting.
Namely, the zero-free region of $2-f_{2\OTOCorder}(s)$ will be closely related to the region in which $|f_{2\OTOCorder}(s)|$ is analytic and bounded. 
Outside the zero-free region, $f_{2\OTOCorder}(s)$ can grow as $|f_{2\OTOCorder}(s)| = \exp(\Omega(n))$, and interpolation will fail due to the exponential blow up.
We show below that $2-f_{2\OTOCorder}(s)$ is zero-free in a real valued strip. 
We will first show the main result regarding the zero-free region and relegate proofs of the bounds on the operators in OTOC which give this zero-free region to later sections.
\begin{lemma}[Zero-free strip for $2-f_{2\OTOCorder}$]\label{lem:strip-zero-free}
For fixed $\eta\in(0,1)$, define
\begin{equation}
    \sigma_\eta:=\frac{1-(2-\eta)^{-1/(2\OTOCorder)}}{2 \locality JD} .
\end{equation}
Then for all $s$ with $|\Im s|\le \sigma_\eta$,
\begin{equation}
|f_{2\OTOCorder}(s)|\le \|X_{2\OTOCorder}(s)\|\le (1-2 \locality JD|\Im s|)^{-2\OTOCorder}\le 2-\eta,
\end{equation}
so $2-f_{2\OTOCorder}(s)\neq 0$ throughout the horizontal strip $\{s:\ |\Im s|\le \sigma_\eta\}$.
\end{lemma}

\begin{proof}
Apply \Cref{prop:B-bound-2k} with $b=\Im s$ obtaining $(1-2 \locality JD|b|)^{-2\OTOCorder}\le 2-\eta$. 
Solving this for $|b|$ completes the proof.
\end{proof}

Using this zero-free strip, we can estimate $\log(2-f_{2\OTOCorder}(t))$ by applying interpolation to the composed function $F(z) = \log(2-f_{2\OTOCorder}(t\phi(z)))$, where $\phi$ is a mapping from the disk to the strip such that $\phi(0)=0$ and $\phi(1)=1$ from \Cref{lem:barvinok-strip-map}.
This procedure follows from standard practice in using Barvinok's interpolation and its performance is shown below \cite{barvinok2016combinatorics}.

\begin{proposition}[Zero-free disk and Taylor truncation for OTOC via strip mapping]\label{prop:otoc-strip-trunc}
Let $H=\sum_e h_e$ be a $\locality$-local, degree-$D$ Hamiltonian on $n$ qubits with $\|h_e\|\le J$. 
Fix a real target time $t>0$.
For $\eta\in(0,1)$, let $\sigma := \sigma_\eta / t$ where $\sigma_\eta$ is taken from \Cref{lem:strip-zero-free}.
Then there exist constants $C_1,C_2>0$ and a polynomial $\phi$ from \Cref{lem:barvinok-strip-map} with
\begin{equation}
    \phi(0)=0,\qquad \phi(1)=1,\qquad \deg(\phi)\le \exp(C_1 t),
\end{equation}
such that for all $|z|\le \hat r$,
\begin{equation}
    |\Im \phi(z)| \le \sigma,\qquad \hat r=1+\exp(-C_2 t).
\end{equation}
Define the composed analytic function on $\{|z|\le \hat r\}$, equal to
\begin{equation}
    F(z) :=  \log\left(2-f_{2\OTOCorder}\left(t \phi(z)\right)\right).
\end{equation}
Then, $F$ is holomorphic on the disk $\overline{\D_{\hat r}}$ and $F(1)=\log(2-f_{2\OTOCorder}(t))$. Let $F(z)=\sum_{\ell\ge 0} a_\ell z^\ell$ be the Maclaurin series and denote its order $K$ truncation as $T_K(1):=\sum_{\ell=0}^K a_\ell$. Then for any $\epsilon\in(0,1)$, choosing
\begin{equation}\label{eq:K-OTOC-prop}
    K \ge e^{\Theta(t)}O( \log(1/\epsilon) ),
\end{equation}
guarantees $\left|F(1)-T_K(1)\right| \le \epsilon$.
\end{proposition}

\begin{proof}
Following \Cref{lem:strip-zero-free}, define
\begin{equation}
    \sigma_\eta := \frac{1-(2-\eta)^{-1/(2\OTOCorder)}}{c_0}.
\end{equation}
\Cref{lem:strip-zero-free} states that $|2-f_{2\OTOCorder}(s)|\ge \eta$ whenever $|\Im s|\le \sigma_\eta$, which gives a zero-free horizontal strip for $2-f_{2\OTOCorder}$.
By \Cref{lem:barvinok-strip-map}, for $\sigma=\sigma_\eta/t$ there exists a polynomial $\phi$ of degree $\deg(\phi)\le \exp(C_1/\sigma)=\exp(C_1 t)$ such that $|\Im \phi(z)|\le \sigma$ for all $|z|\le \hat r$ with $\hat r=1+\exp(-C_2/\sigma)=1+\exp(-C_2 t)$, and $\phi(0)=0$, $\phi(1)=1$.
For $|z|\le \hat r$ we have $|\Im(t\phi(z))| \le \sigma_\eta$ implying $|f_{2\OTOCorder}(t\phi(z))|\le 2-\eta$ and $2-f_{2\OTOCorder}$ stays $\eta$-away from $0$. 
Choose the principal logarithm and $F$ is analytic on the closed disk $|z|\le \hat r$.

Setting $\varrho=(1+\hat r)/2$, we have for $|z| \le \varrho$,
\begin{equation}
    |F(z)|=\left|\log 2+\log(1-\tfrac12 f_{2\OTOCorder}(t\phi(z)))\right|
    \le \log 2+\frac{|f_{2\OTOCorder}(t\phi(z))|/2}{1-|f_{2\OTOCorder}(t\phi(z))|/2}
    \le \log 2+\frac{2-\eta}{\eta} =: C_\eta,
\end{equation}
using $|\log(1-w)|\le |w|/(1-|w|)$ for $|w|<1$.
Applying \Cref{lem:disk-analytic} with $R=\hat r$, $\rho=\varrho$, and $M_\varrho\le C_\eta$ we get
\begin{equation}
    |F(1)-T_K(1)| \le  \frac{C_\eta}{\varrho^{K+1}(1-1/\varrho)}.
\end{equation}
Solving for $K$ yields \Cref{eq:K-OTOC-prop}. 

\end{proof}

The upshot of the above is that we can efficiently estimate OTOC whenever $t=O(\log \log n)$.
\begin{corollary}\label{cor:barvinok-final}
With $\eta \in (0,1)$ constant, $\hat r=1+\exp(-\Theta(t))$ and $F(z)=\log(2-f_{2\OTOCorder}(t\phi(z)))$ as in \Cref{prop:otoc-strip-trunc}, truncating the Taylor series at order $K=O\left(e^{\Theta(t)}\log(1/\epsilon) \right)$ yields an $\epsilon$-additive approximation to $f_{2\OTOCorder}(t)$. Computing the first $K$ coefficients costs $O\left(m^{K} \mathrm{poly}(n,K)\right)$ time.
In particular, if $t=O(\log \log n)$, then $K=O(\log (n) \log(1/ \epsilon))$ resulting in a quasi-polynomial time algorithm.
\end{corollary}

Barvinok's algorithm in this setting performs interpolation on the function $F(z) :=  \log\left(2-f_{2\OTOCorder}\left(t \phi(z)\right)\right)$ where $\phi:\mathbb{D}_R\to \mathbb{C}$ is a map from the disk to the strip $\{z:|\Im z| \leq \sigma \}$ for $\sigma = O(1/t)$.
In \Cref{app:disk-to-strip-mapping}, we showed that the disk to strip mapping is optimal in its parameter $\hat r=1+e^{-\Theta(1/\sigma)}$ which is the dominating factor in the truncation of $K = e^{\Theta(t)}O( \log(1/\epsilon) )$ in \eqref{eq:K-OTOC-prop} with respect to the time parameter. 
This lends evidence for potential quantum advantage for OTOC in settings where $t$ grows with $n$, such as in lattice OTOC problems.
\begin{remark}[Limitations of polynomial truncation for growing $t$]
    With truncation at degree $K = e^{\Theta(t)}O( \log(1/\epsilon) )$ in \Cref{eq:K-OTOC-prop}, the runtime of Barvinok's interpolation is doubly exponential in $t$.
    Every holomorphic map from the disk $\mathbb{D}_R$ to the strip of height $\sigma$ must satisfy $R\leq 1+e^{-\Omega(1/\sigma)}$ (see \Cref{prop:disk-strip-optimality}) as $\sigma = \Theta(1/t)$. 
    Therefore, any exponential improvements to the runtime of estimating OTOC using Barvinok's interpolation with respect to the time parameter $t$ must somehow incorporate a better zero-free region.
\end{remark}

\subsection{Bounds on operators in OTOC}

Let $\ad_A(\cdot)=[A,\cdot]$ and let $\ad_A^r$ denote the $r$-fold iteration of the $\ad$ operation. 
Note that
\begin{equation}
B(s)=e^{is  \ad_H}(B)=\sum_{r\ge 0}\frac{(is)^r}{r!} \ad_H^r(B).
\end{equation}
We bound the size and norm of $\ad_H^r(B)$ using the bounded degree and locality assumptions on $H$. 
The following bound on $\|\ad_H^r(B)\|$ is standard and can be found in other works, e.g. see \cite{alhambra2023quantum}.

\begin{lemma}\label{lem:comb}
For all $r\ge 0$,
$\|\ad_H^r(B)\| \le  (c_0)^r r!$, where $c_0:=2 \locality JD$.
\end{lemma}

\begin{proof}
Expand $\ad_H^r(B)=\sum_{(e_1,\dots,e_r)} \ad_{h_{e_r}}\cdots \ad_{h_{e_1}}(B)$.
A term is nonzero only if $h_{e_j}$ overlaps the support of the prior $j-1$ steps.
After $j-1$ commutators the support size is at most $1+\locality(j-1)$, so the number of admissible $h_{e_j}$ is at most $D(1+\locality(j-1))\le \locality Dj$ for $j\ge 1$.
Thus the number of nonzero ordered $r$-tuples is at most
$\prod_{j=1}^r \locality Dj = (\locality D)^r r!$.
For an operator $O$, we have $\|\ad_{h_e}(O)\|\le 2\|h_e\| \|O\|\le 2J\|O\|$. Applying this $r$ times yields a factor $(2J)^r$.
In conclusion, we get $\|\ad_H^r(B)\|\le (2J)^r(\locality D)^r r!=(2 \locality JD)^r r!$.
\end{proof}

\begin{lemma}[Geometric bounds for $\|B(s)\|$, $\|Y(s)\|$ and $\|X_{2\OTOCorder}(s)\|$]\label{prop:B-bound-2k}
For all $s=a+ib\in\mathbb{C}$ with $|b|<1/c_0=1/(2 \locality JD)$,
\begin{equation}
\|B(s)\| \le \frac{1}{1-c_0|b|},\qquad
\|X_{2\OTOCorder}(s)\| \le \frac{1}{(1-c_0|b|)^{2\OTOCorder}}.
\end{equation}
\end{lemma}

\begin{proof}
In what follows we write $Y(s):=B(s)M$ and $X_{2\OTOCorder}(s):=Y(s)^{2\OTOCorder}$.
By \Cref{lem:comb},
\begin{equation}
\|B(ib)\|\le \sum_{r\ge 0}\frac{|b|^r}{r!} \|\ad_H^r(B)\|
\le \sum_{r\ge 0} (c_0|b|)^r = \frac{1}{1-c_0|b|}.
\end{equation}
Then $\|Y(s)\|\le \|e^{-iHa}\|  \|B(ib)\|  \|e^{iHa}\|  \|M\|=\|B(ib)\|$ and $\|X_{2\OTOCorder}(s)\|\le \|Y(s)\|^{2\OTOCorder}$. 
\end{proof}

\subsection{Comparison to Lieb Robinson truncation}\label{subsec:LR-degree-d}

Let $B$ be supported at site $b$ and $M$ at site $m$ with distance $r:=\mathrm{dist}(b,m)$. 
For radius $R\in\mathbb{N}$, denote $B_R$ as the set containing all sites at graph distance at most $R$ around $b$. 
Then, let $H_{B_R}$ be the restriction of $H$ to $B_R$ dropping all terms that are in $B_R^c$. 
Define corresponding OTOC terms for $H_{B_R}$:
\begin{equation}
    B_{B_R}(t) := e^{i H_{B_R} t} B e^{-i H_{B_R} t},\qquad
    X_{2\OTOCorder,B_R}(t) := \left(B_{B_R}(t) M\right)^{2\OTOCorder},\qquad
    f^{\mathrm{LR}}_{2\OTOCorder}(t) := \Tr\left[\rho X_{2\OTOCorder,B_R}(t)\right].
\end{equation}
Using Lieb-Robinson bounds, one can obtain an estimate of the OTOC using the surrogates above as we show below.
To make this comparison realistic, we will assume throughout that $\rho$ is a state which is efficiently representable on a classical computer (e.g. product or stabilizer state).
\begin{proposition}[LR truncation complexity on degree-$D$ graphs]\label{prop:LR-degree-d}
Let $\mu,v,c>0$ be the constants from the truncated‑evolution Lieb-Robinson bound of \cite[Prop.~4.3]{chen2023speed}. Then for any $t\in\mathbb{R}$ and $\epsilon>0$, choosing
\begin{equation}\label{eq:R-choice}
    R = v|t| + \mu^{-1} \log \left(\frac{C_0}{\epsilon} \right)
\end{equation}
guarantees the additive error in the OTOC
\begin{equation}\label{eq:LR-err}
    \left| f_{2\OTOCorder}(t)-f^{\mathrm{LR}}_{2\OTOCorder}(t)\right|
    \le  2\OTOCorder C_0 e^{-\mu (R-v|t|)} \le  \epsilon .
\end{equation}
Moreover, classical Heisenberg simulation over the region $B_R$ has runtime
\begin{equation}\label{eq:TLR-degree-d}
    T_{\mathrm{LR}}(t,\epsilon;D)
     = 2^{ O\left(D^{ v|t|} \left(\frac{1}{\epsilon}\right)^{\frac{\log D}{\mu}}\right)} .
\end{equation}
\end{proposition}

\begin{proof}
Let us denote $N_R$ as the number of qubits in $B_R$ bounded as
\begin{equation}\label{eq:ball-volume}
    N_R := |B_R| \le 1+D\sum_{j=0}^{R-1}(D-1)^j = O \left((D-1)^R\right),
\end{equation}
Applying \cite[Prop.~4.3]{chen2023speed} setting $A$ in their proposition to $A=B$ (so $|\partial A|=O(1)$, $\|A\|=1$ in their notation) yields
$\|B(t)-B_{B_R}(t)\|\le C_0 e^{-\mu (R-v|t|)}$ with $C_0,\mu,v$ constant.
Using the bound
$\|(XM)^{2\OTOCorder}-(\tilde X M)^{2\OTOCorder}\| \le 2\OTOCorder\|X-\tilde X\|$ for $X,\tilde X, M$ unitary
gives \Cref{eq:LR-err}. 
Evolving $B$ on $B_R$ and evaluating the trace costs $2^{O(N_R)}$, resulting in \Cref{eq:TLR-degree-d}.
\end{proof}

\begin{remark}[Comparison with Barvinok's method]
Barvinok interpolation (\Cref{cor:barvinok-final}) has truncation order $K=O \left(\log(1/\epsilon)e^{\Theta(t)}\right)$ which has logarithmic dependence on $1/\epsilon$. 
In comparison to the LR runtime in \eqref{eq:TLR-degree-d} which is exponential in $\mathrm{poly}(1/\epsilon)$, the runtime of Barvinok's interpolation is exponential in $\log(1/\epsilon)$. 
This is a significant improvement in expander type models.
For $d$‑dimensional lattices, Lieb-Robinson methods achieve a similar scaling however as $T_{\mathrm{LR}}=2^{O\left((|t|+\log(1/\epsilon))^{d}\right)}$ follows from $|B_R|=O(R^d)$.
\end{remark}

\end{document}